\documentclass[envcountsame, runningheads]{llncs}
\usepackage{amssymb}
\usepackage{amsmath}
\usepackage{mathtools}
\usepackage{relsize}

\usepackage{framed}
\usepackage{url}

\usepackage{algpseudocode}
\usepackage{algorithm}
\usepackage{cite}
\usepackage{array}

\usepackage{tabularx}
\usepackage{todonotes}

\usepackage[bookmarks=false, colorlinks,
linkcolor=red!60!black, citecolor=blue!60!black]{hyperref}

\usepackage{enumitem}
\usepackage{xcolor}
\usepackage{comment}
\usepackage{tikz}
\usepackage{epsfig}
\usepackage{wrapfig}
\usetikzlibrary{shapes,backgrounds, patterns}

\usepackage{booktabs, multirow}

\newsavebox\savefigure

\colorlet{bscolor}{blue}
\colorlet{skcolor}{orange}
\colorlet{vrcolor}{red}

\newcommand{\Omit}[1]{}

\newcommand{\chicn}{\chi_{CN}(G)}
\newcommand{\chion}{\chi_{ON}(G)}
\newcommand{\chionp}{\chi^*_{ON}(G)}
\newcommand{\chicnp}{\chi^*_{CN}(G)}
\newcommand{\chionps}{\chi^*_{ON}}

\newcommand{\cfon}{CFON}
\newcommand{\cfcn}{CFCN}
\newcommand{\cfonp}{CFON*}
\newcommand{\cfcnp}{CFCN*}

\newcommand{\NP}{{\sf NP}}

\newtheorem{Observation}[theorem]{Observation}

\usepackage{array}
\usepackage{comment}
\usepackage{tcolorbox}
\tcbuselibrary{breakable}
\newtcolorbox{mybox}[2][]{colbacktitle=white,colback=white,coltitle=black,title={#2},fonttitle=\bfseries,#1, left = 2mm, right = 2mm, breakable}

\title{Conflict-Free Coloring: Graphs of Bounded Clique-Width and Intersection Graphs\thanks{A subset of the results of this paper appeared in the Proceedings of  
IWOCA 2021 \cite{iwoca-sri}. } 
}
\titlerunning{Conflict-Free Coloring:  Bounded Clique-Width and Intersection Graphs}
\author{Sriram Bhyravarapu\inst{1}, Tim A.\ Hartmann\inst{2}, Hung P. Hoang\inst{3}, Subrahmanyam Kalyanasundaram\inst{4, }\thanks{Corresponding author. Email: {\tt subruk@cse.iith.ac.in}} \and I. Vinod Reddy\inst{5}}
\institute{
The Institute of Mathematical Sciences, HBNI, Chennai, India
\and
CISPA Helmholtz Center for Information Security, Germany \and 
Algorithms and Complexity Group, TU Wien, Austria \and
 Department of Computer Science and Engineering, IIT Hyderabad \and 
Department of Electrical Engineering and Computer Science,  IIT Bhilai \\
\email{sriramb@imsc.res.in, tim.hartmann@cispa.de,
phoang@ac.tuwien.ac.at, subruk@cse.iith.ac.in,  vinod@iitbhilai.ac.in}
}

\authorrunning{Bhyravarapu, Hartmann, Hoang, Kalyanasundaram and Reddy}
\begin{document}

\maketitle

% \footremember{A part of this work happened the author was at IIT Hyderabad. },
\begin{abstract}
A conflict-free coloring of a graph $G$ is a (partial) coloring of its vertices
	such that every vertex $u$
	has a neighbor whose assigned color is unique in the neighborhood of $u$.
There are two variants of this coloring,
	one defined using the open neighborhood and one using the closed neighborhood.
For both variants, we study the problem of deciding whether the conflict-free coloring
of a given graph~$G$ is at most a given number~$k$.
%The decision version of these problems are {\sf NP}-complete for general graphs.

In this work, we investigate the relation of clique-width and minimum number of colors needed (for both variants) and show that these parameters do not bound one another. 
Moreover, we consider specific graph classes, particularly graphs of bounded clique-width and types of intersection graphs, such as distance hereditary graphs, interval graphs and unit square and disk graphs.
We also consider Kneser graphs and split graphs.
We give (often tight) upper and lower bounds and determine the complexity of the decision problem on these graph classes, which improve some of the results from the literature.
Particularly, we settle the number of colors needed for an interval graph to be conflict-free colored under the open neighborhood model,	which was posed as an open problem.

\end{abstract}

\section{Introduction}
\label{sec:intro}
%Vertex coloring is one of the most fundamental concepts in graph theory.
Graph coloring is one of the most fundamental problems in graph theory. 
A \emph{proper coloring} of a given undirected graph $G$ is an assignment of colors to the vertices of $G$ such that no two adjacent vertices have the same color.
The minimum number of colors for which a proper coloring of $G$ exists is called the \emph{chromatic number} of $G$.
There have been extensive studies on this parameter, both algorithmically (e.g., determining or approximating the chromatic number) and structurally (e.g., worst-case bounds on the chromatic number of a given graph class, notably planar graphs); see~\cite{MR1304254} for example for an overview.

Besides the classical coloring above, there have been many variants.
One such variant is introduced in 2002 by Even, Lotker, Ron and Smorodinsky \cite{Even2002}, motivated by the frequency assignment problem in cellular networks, where base stations and clients communicate with one another.
To avoid interference, it is required that for each client, among the base stations that it connects to, there exists one with a unique frequency. 
This is formalized as a \emph{conflict-free coloring}. In the below definition, the open neighborhood
of a vertex is the set of its adjacent vertices. 
\begin{definition}[Conflict-Free Coloring]\label{def:cfonp}
%\textcolor{blue}{
%    A \emph{partial conflict-free open-neighborhood coloring} (\cfonp{} coloring) of a graph $G = (V, E)$ using $k$ colors
%    is an assignment $C: V' \rightarrow \{1, 2, \ldots, k\}$ where $V'$ is a subset of $V$,
%    such that for every vertex $v \in V$, %there is an $i \in \{1, 2, \ldots, k\}$
%    there is a neighbor whose color is unique in the open neighborhood of $v$.}
%\textcolor{red} 
{
    A \emph{partial conflict-free open-neighborhood coloring} (\cfonp{} coloring) of a graph $G$, $G = (V, E),$ using $k$ colors
    is an assignment $C: V \rightarrow \{0\} \cup \{1, 2, \ldots, k\}$ 
    such that for every vertex $v \in V$, %there is an $i \in \{1, 2, \ldots, k\}$
    there is a vertex with a unique non-zero color in the open neighborhood of $v$.}
    
    The smallest $k$ for which there is a \cfonp{} coloring of $G$
    is called the \emph{\cfonp{} chromatic number} of $G$, denoted by $\chionp$. Given a graph $G$ and a natural number $k$, the \emph{\cfonp{} coloring problem} asks whether $\chionp$ is at most $k$.    
%We consider two variants, one where we consider the open neighborhood
%    and one where we consider the closed neighborhood. 
\end{definition} 

%Hence, instead of using a partial mapping we use color 0
%    for the `uncolored' vertices.
%For the full (that is not partial) conflict-free coloring
%    then we require that $c$ assigns positive colors only.

%It differs from the classical coloring in three aspects: 
%\begin{itemize}
%    \item We consider a \emph{partial coloring} (that is, a vertex %may not be assigned a color);
%    \item The coloring is not necessarily proper (that is, two adjacent vertices may share the same color); and
%    \item For every vertex, there is a color that appears exactly %once in its open neighborhood, i.e., the set of its adjacent vertices.
%\end{itemize} 
%To be more precise, we call this a \emph{partial conflict-free open-neighborhood coloring} or \emph{\cfonp{} coloring}.\todo{Subruk: We have to note somewhere that we sometimes use the color 0 to indicate a vertex that is not colored.}
%The smallest number of colors required for a \cfonp{} coloring of a graph $G$ is called the \emph{\cfonp{} chromatic number of $G$}, denoted by $\chionp$. 

Similarly, we can define a coloring variant for closed neighborhoods, where a closed neighborhood of a vertex contains the vertex itself and its adjacent vertices.
We call this a \emph{partial conflict-free closed-neighborhood coloring} or \emph{\cfcnp{} coloring}. 
The \cfcnp{} chromatic number $\chicnp$ and the \cfcnp{} coloring problem are defined analogously.
Collectively, we refer to the chromatic numbers of these two variants as conflict-free chromatic numbers.
The \cfonp{} and \cfcnp{} colorings are referred to as ``partial colorings'' because the vertices colored 0 are treated as ``uncolored''. 

Conflict-free coloring has been well studied
for nearly 20 years (e.g., see the survey by Smorodinsky~\cite{smorosurvey}) and also found applications in the area of sensor networks \cite{vijayan, collision} and coding theory \cite{KrishnanMK21}.
Similar to the classic setting, these works explored various combinatorial and algorithmic questions on conflict-free coloring.
What are worst-case bounds on any of the conflict-free chromatic numbers?
What is the computational complexity of the conflict-free coloring problems?
For what kind of graph classes can we get better bounds and complexity for the questions above?
While many papers have addressed these questions, there are still many gaps, which we bridge in this paper.

\subsection{Results and Discussion}\label{subsec:results}
This paper is the extended version of the preliminary version \cite{iwoca-sri} published in the proceedings of the IWOCA 2021 conference.
In this paper, we obtain many new results and provide the full proofs that were omitted in the preliminary version.
In the following, we briefly present the results of the current paper and highlight the changes from the preliminary paper.
A summary of the results for \cfonp{} and \cfcnp{} colorings are also presented in Tables \ref{tab:summary} and \ref{tab:summary-cfcn}, respectively. 

\begin{itemize}
    \item
    In Section~\ref{sec:cliquewidth}, we discuss the conflict-free chromatic numbers in relation with the parameter clique-width.
    In the preliminary version~\cite{iwoca-sri}, we presented fixed parameter tractable (FPT) algorithms for all the conflict-free coloring problems with respect to the number of colors and clique-width, with full correctness proof.
    %full proof given an $n$-vertex graph with clique width $w$ and a natural number $k$, we presented an algorithm for each of the conflict-free coloring problems in time $2^{O(w3^k)} n^{O(1)}$ with a full correctness proof. 
    % By a recent result by Bergougnoux, Dreier and Jaffke~\cite{logic-mim-width}, a better runtime of $2^{O(w3^k)} n^{O(1)}$ can be achieved.
    If the \cfonp{} and \cfcnp{} chromatic numbers are bounded by a function of clique-width, the result above will translate to an FPT algorithm with respect to only the clique-width.
    As a new result, we show that the conflict-free chromatic numbers cannot be bounded 
    by a function of clique-width.
    %this boundedness of the conflict-free chromatic numbers is actually false. 
    Towards this end, we show the existence of graphs with clique-width three and conflict-free chromatic numbers $\Omega(\log n)$. 
    The existence of an FPT algorithm with respect to the clique-width remains open.

    \item In Section \ref{sec:dist_here}, we discuss certain graphs with bounded clique-width. 
    In particular, for distance-hereditary graphs, we show that the \cfcnp{} chromatic number is at most three.
    Consequently, we can obtain a polynomial time algorithm for this graph class, by applying any FPT or XP algorithm with respect to the clique-width and the number of colors~\cite{iwoca-sri, logic-mim-width, DBLP:journals/corr/abs-2203-15724}.
    The \cfonp{} chromatic number for this graph class, however, is unbounded.
    Still, we show that it is bounded for two subclasses, cographs and block graphs, and hence the \cfonp{} coloring problem is polynomial time solvable on them.
    The results related to these two subclasses have been announced in the preliminary version~\cite{iwoca-sri}, and here we provide the full proof.
    The results for distance-hereditary graphs are new.
    
    \item In Section \ref{sec:interval}, we 
    %provide the full proof for the facts 
    show that for an interval graph $G$, $\chionp{}\leq 3$ and that this bound is tight.
    This result answers an open question posed in \cite{vinod2017}.
    Moreover, we show that two colors are sufficient to \cfonp{} color proper interval graphs.
    All these results were announced in the preliminary version~\cite{iwoca-sri}, but the full proof was only provided for the upper bound on proper interval graphs. 

    \item In Section \ref{sec:unitsquare}, we provide the full proof for the upper bounds of the \cfonp{} chromatic numbers of 27 and 54 for unit square and unit disk intersection graphs, respectively. 
    
    Further, in Section \ref{sec:np-hard}, we show a new {\sf NP}-completeness result for the \cfonp{} coloring problem on unit square and unit disk intersection graphs.
    These results complement the corresponding bounds and complexity for the closed-neighborhood variant studied previously by Fekete and Keldenich~\cite{Sandor2017}.

    \item In the last two sections, we provide the full proofs related to Kneser graphs and split graphs, as announced in~\cite{iwoca-sri}.
    In particular, in Section~\ref{sec:kne_split}, for the Kneser graph $K(n,\kappa)$, $\kappa+1$ colors are sufficient when $n\geq 2\kappa+1$ and are also necessary when $n \geq 2\kappa^2 + \kappa$.
    For \cfcnp{} coloring of $K(n,\kappa)$, we show an upper bound of $\kappa$ colors for $n\geq 2\kappa+1$. 

    In Section~\ref{sec:split}, we prove that for split graphs, the \cfonp{} coloring problem is {\sf NP}-complete and that the \cfcnp{} coloring problem is polynomial time solvable. 
    %\todo[]{Remove the first result and change $k$ in the Kneser graphs. }
\end{itemize}

\begin{table}[h!]
\centering
\begin{tabular}{|c c c c|} 
 \hline
 \textbf{Graph Class}  &\textbf{Upper Bound} &  \textbf{Lower Bound}&  \textbf{Complexity}\\ [0.5ex] 
 
 \hline\hline 
%  $(G,{\sf cw},k)$ & -  & - & {\sf FPT} \cite{iwoca-sri, logic-mim-width}\\ 
 % \hline
  Distance hereditary graphs & - & $\Omega(\log n)$ (Cor. \ref{cor:dist-here-lowerbound}) & - \\
  \hline
 Block graphs & 3  & 3 (Fig.~\ref{figure:block:graph}) & {\sf P}\\ 
  \hline
 Cographs & 2 & 2 ($K_3$) & {\sf P}\\
 \hline
  Interval graphs &3 & 3 (Fig.~\ref{figure:interval:lower:bound}) & {\sf P} \cite{DBLP:journals/corr/abs-2203-15724, interval-sri-mfcs} \\ 
 \hline
 Proper interval graphs & 2 & 2 ($K_3$) & {\sf P} \cite{DBLP:journals/corr/abs-2203-15724, interval-sri-mfcs} \\
 \hline
 Unit square graphs & 27 & 3 (Fig.~\ref{fig:unit_square_example}) & {\sf NP}-hard\\ 
 \hline
 Unit disk graphs & 54 & 3 (Fig.~\ref{fig:unit_square_example}) & {\sf NP}-hard\\
  \hline
 Kneser graphs $K(n,\kappa)$ & $\kappa+1$  & $\kappa+1$ (Lem.~\ref{lem:lowerbound}) & - \\ 
 \hline
 Split graphs & - & - & {\sf NP}-hard 
 \\  [1ex]
 \hline
\end{tabular}
\vspace{0.5cm}
\caption{Bounds and algorithmic status on various graph classes for the \cfonp{} coloring problem. 
The results that were previously known are indicated by providing citations to the papers. 
The absence of a citation indicates that the result is shown in this paper.
Here a ``Lower Bound'' of $\ell$ indicates the existence of a graph $G$ such that 
$\chionp{}=\ell$. Such graphs are indicated in parenthesis. 
If the bounds or the algorithmic status (whether {\sf P} or {\sf NP}-hard) for a graph class is unknown, we indicate it by ``-''. 
%The graphs for lower bounds are indicated in parentheses. 
%$K_3$ refers to  the complete graph on 3 vertices.
%The lower bound proof for distance hereditary graphs is shown in Lemma~\ref{lem:dist_here_ON}. 
%The lower bound for Kneser graphs only applies for $n \geq 2\kappa^2 + \kappa$; see Lemma~\ref{lem:lowerbound}. 
}
\label{tab:summary}
\end{table}

\begin{table}[h!]
\centering
\begin{tabular}{|c c c|} 
 \hline
 \textbf{Graph Class}  &\textbf{Upper Bound} &  \textbf{Complexity}\\ [0.5ex] 
 
 \hline\hline 
%  $(G,{\sf cw},k)$ & -  & {\sf FPT} \cite{iwoca-sri, logic-mim-width}\\ 
%  \hline
   Distance Hereditary Graphs & 3 
   %(Fig.~\ref{figure:bull-graph}) 
   & {\sf P}\\
  \hline
  Block graphs & 2   
  %(Fig \ref{figure:bull-graph}) 
  & {\sf P}\\ 
  \hline
 Cographs & 2 
 %($K_{2,2}$) 
 & {\sf P}\\
 \hline
  Interval graphs &2\cite{Sandor2017}  
  %(Fig.~\ref{figure:bull-graph}) 
  & {\sf P} \cite{DBLP:journals/corr/abs-2203-15724, iwoca-sri}\\ 
 \hline
 Proper interval graphs & 2\cite{Sandor2017}  
 %(Fig.~\ref{figure:bull-graph}) 
 & {\sf P} \cite{DBLP:journals/corr/abs-2203-15724, iwoca-sri}\\
 \hline
 Unit square graphs & 4\cite{Sandor2017} 
 %(Fig.~\ref{figure:bull-graph}) 
 & {\sf NP}-hard \cite{Sandor2017}\\ 
 \hline
 Unit disk graphs & 6\cite{Sandor2017}  
 %(Fig.~\ref{figure:bull-graph}) 
 & {\sf NP}-hard \cite{Sandor2017}\\
  \hline
 Kneser graphs $K(n,\kappa)$ & $\kappa$ & - \\ 
 \hline
 Split graphs & 2  
 %(Fig.~\ref{figure:bull-graph}) 
 & {\sf P}
 \\  [1ex]
 \hline
\end{tabular}
\vspace{0.5cm}
\caption{Bounds and algorithmic status on various graph classes for the \cfcnp{} coloring problem. 
%Here a ``Lower Bound'' of $\ell$ indicates existence of a graph $G$ such that 
%$\chicnp{}=\ell$. 
The results that were previously known are indicated by providing citations to the papers. 
The absence of a citation indicates that the result is shown in this paper.
If the algorithmic status (whether {\sf P} or {\sf NP}-hard) for a graph class is unknown, we indicate it by ``-''. %For all graph classes except cographs and Kneser graphs, the bull graph in Figure \ref{figure:bull-graph}
%is a graph that requires 2 colors. 
%For cographs, the complete bipartite graph $K_{2,2}$ is a graph that requires 2 colors. 
}
\label{tab:summary-cfcn}
\end{table}

\subsection{Related works}
Here, we briefly review a few results related to the materials of this paper.
We will elaborate the relevant results further at each subsequent section.
For a more general overview on conflict-free coloring, we refer the reader to the survey by Smorodinsky~\cite{smorosurvey}. %and the review by Abel et al.~\cite{planar}. \todo[]{the paper by Abel is not a review paper. Should we remove?}

Many papers in the literature of conflict-free coloring considered the variants where the partial coloring is a full coloring.
We remove the asterisks to denote these cases.
That is, we denote by \cfon{} coloring the \emph{full conflict-free open-neighborhood coloring}, and denote by \cfcn{} coloring for the closed neighborhood variant.
The corresponding chromatic numbers and decision problems are defined analogously.

In terms of asymptotic worst-case bounds, Pach and T\'ardos~\cite{Pach2009} showed a bound of $O(\log n)$ for a general $n$-vertex graph for any of the four variants of conflict-free coloring.
Glebov, Szab\'o, and Tardos~\cite{MR3189420} proved that this bound is tight, using a randomized construction. 
This bound can be improved for special graph classes, such as random graphs $G(n,p)$ for $p \in \omega(1/n)$~\cite{MR3189420} and graphs with bounded degrees~\cite{gupta2023extremal, Pach2009}.
% Specially, due to the original motivation of the problem, many geometric graph classes have been given attention, including graphs induced by coverage of point sets on the plane by convex regions~\cite{Even2002}, planar and outerplanar graphs~\cite{planar, planartight} and various intersection graphs, such as unit disk graphs, unit square graphs, and interval graph~\cite{Sandor2017, vinod2017}. 
% The four conflict-free chromatic numbers are bounded by constants for these classes, except for the first class. 
% For example, the related result for planar graphs by Huang, Guo, and Yuan~\cite{planartight} showed that four colors are sufficient to \cfonp{} color these graphs.

% The CFON$^*$ variant is generally considered to be harder than the CFCN$^*$ variant, see for instance, remarks in \cite{Keller2018,Pach2009}. 
From an algorithmic perspective, Gargano and Rescigno~\cite{gargano2015} showed that the \cfon{} and \cfcn{} coloring problems are {\sf NP}-complete, and the corresponding chromatic numbers are hard to approximate within a factor less than 3/2.
Abel et al.~\cite{planar} later showed the {\sf NP}-completeness for the \cfonp{} and \cfcnp{} coloring problems, where the former problem is {\sf NP}-complete even for planar bipartite graphs. 
Given these complexity results, two natural approaches for further investigation are to study the parameterized complexity of these problems or to restrict the classes of graphs for which the problems may be efficiently solved.

The parameterized complexity of conflict-free coloring problems
has captured the interest of the research community recently.
The \cfon{} and \cfcn{} coloring problems are FPT\footnote{For the formal definition of FPT and more details on parameterized complexity, we refer the reader to \cite{CyganFKLMPPS15,downey2013fundamentals}.} when parameterized by tree-width~\cite{Bod2019, aatw},
distance to cluster (distance to disjoint union of cliques) \cite{vinod2017}, 
neighborhood diversity \cite{gargano2015}.
%Agarwal et al. \cite{aatw} also proved that the \cfcn{} coloring problem is FPT with respect to tree width. 

Our preliminary paper~\cite{iwoca-sri} and a recent paper by Bergougnoux, Dreier, and Jaffke~\cite{logic-mim-width} showed that the four variants of conflict-free coloring problems are FPT with respect to the combined parameters clique-width and the number of colors.
Clique-width is more general than the parameters mentioned in the preceding paragraph. 
%In other words, a bounded clique-width implies that the other parameters are also bounded.
In other words, if these parameters are bounded, then the clique-width is also bounded. 
More recently, Gonzalez and Mann~\cite{DBLP:journals/corr/abs-2203-15724} showed that the problems are polynomial time solvable when the mim-width and the number of colors are constant.
Although this parameter mim-width is more general than clique-width, the algorithms by Gonzalez and Mann are not FPT algorithms.

In the direction of restricted classes of graphs, many geometric graph classes have been given attention specially, due to the original motivation from the frequency assignment problem.
The original paper~\cite{Even2002} considered graphs induced by coverage of point sets on the plane by convex regions.
Abel et al.~\cite{planar} presented many combinatorial and algorithmic results for the planar graphs. 
The latest paper for these graphs by Huang, Guo, and Yuan~\cite{planartight} gives a tight bound of 4 for the \cfonp{} chromatic number.
Intersection graphs are also natural classes of geometric graphs to consider.
Fekete and Keldenich \cite{Sandor2017} studied  CFCN$^*$ coloring on common intersection graphs such as interval graphs, unit disk graphs and unit square graphs. (See the references therein for further related works on intersection graphs.) 
This paper poses an open question on the existence of a polynomial time algorithm for the CFON$^*$ problem on interval graphs. 
This was recently proved affirmatively, independently by Bhyravarapu, Kalyanasundaram and Mathew \cite{interval-sri-mfcs}, and Gonzalez and Mann~\cite{DBLP:journals/corr/abs-2203-15724}. 
Beside these intersection graphs, several others have been considered, such as string graphs and circle graphs \cite{kellerstring}.

\section{Preliminaries}
Throughout the paper, we assume that the input graph $G = (V, E)$ is connected. 
We also assume that $G$ does not contain any isolated vertices 
as the \cfonp{} problem is not defined for an isolated vertex. 
All the results of this paper hold for disconnected graphs without isolated vertices by the application of the respective theorems on each connected component.
%For the \cfonp{} problem, 
We use $[k]$ to denote the set $\{1, 2, \dots, k\}$ 
and %the coloring function 
$C:V \rightarrow \{0\} \cup [k]$ to denote the coloring function. % assigned to a vertex. 
A \emph{universal vertex} is a vertex that is adjacent to all other vertices of the graph.  
In some of our
algorithms and proofs, it is convenient to distinguish between vertices that are intentionally
left uncolored, and the vertices that are yet to be assigned any color.
The assignment of color 0 is used to denote that a vertex is left intentionally uncolored by the 
coloring function.

To simplify the notation and for ease of readability, we use the shorthand notation
$vw$ to denote the edge $\{v, w\}$.
The open neighborhood of a vertex $v\in V$ is the 
set of vertices 
%$\{w:\{v,w\}\in E(G)\}$ 
$\{w: vw\in E\}$ 
and is denoted by 
$N(v)$. 
Given a conflict-free coloring $C$, 
a vertex $w\in N(v)$ 
is called a \emph{uniquely colored neighbor} of $v$ if $C(w)\neq 0$ and $\forall x\in N(v)\setminus \{w\}, C(w)\neq C(x)$. 
The closed neighborhood of $v$ is the set $N(v)\cup \{v\}$, denoted by $N[v]$. 
The notion of a uniquely colored neighbor in the closed neighborhood variant is analogous to the open neighborhood variant, and is obtained by replacing $N(v)$ by $N[v]$. 
%We sometimes use the mapping $h: V \rightarrow V$ to denote a uniquely colored
%neighbor of a vertex.
Given a vertex set $V' \subseteq V$, we define $C(V')$
as follows: $C(V') = \bigcup_{v \in V'} \{ C(v)\}$. 
%for  $V' \subseteq V(G)$.
%$C(V') = \{ C(v) \mid v \in V' \}$ for $V' \subseteq V(G)$.
%To refer to the multi-set of colors used in $V'$, we use $\cm(V')$. The difference between
%$\cm(V')$ and $C(V')$ is that we use multiset union in the former. 

%A parameterized problem is denoted as $(I,k)\subseteq \Sigma^*\times \mathbb{N}$, where $\Sigma$ is fixed alphabet and $k$ is called the parameter. We say that the problem $(I,k)$ is {\it fixed-parameter tractable} (FPT in short) with respect to the parameter $k$ if there exists an algorithm which solves the problem in time $f(k) |I|^{O(1)}$, where $f$ is a computable function.

In many of the sections, in order to establish bounds on the \cfonp{} and \cfcnp{} chromatic numbers, we use the full coloring conflict-free coloring variants, defined as follows. 

\begin{definition}[Conflict-Free Coloring -- Full Coloring Variant]\label{def:cfon}
%Let $G= (V, E)$ be a graph. 
A \emph{\cfon{} coloring} of a graph $G$, $G = (V,E)$, using $k$ colors is an 
assignment $C:V \rightarrow \{1, 2, \ldots, k\}$ such that for every $v \in V$, 
there exists an $i \in \{ 1, 2, \ldots, k\}$ such that  $|N(v) \cap C^{-1}(i)| = 1$.
The smallest number of colors required for a \cfon{} coloring
of $G$ is called the \cfon{} {chromatic number of $G$}, denoted by $\chion{}$.

The corresponding closed neighborhood variant is denoted by \cfcn{} coloring, and 
the CFCN chromatic number is denoted by $\chicn$.
\end{definition}
\begin{remark}
    A full conflict-free coloring, where all the vertices are colored with a non-zero color, is 
also a partial conflict-free coloring (as defined in Definition \ref{def:cfonp}), while the converse is not necessarily true. 
   One extra color suffices to obtain a full conflict-free coloring  from a partial conflict-free coloring. 
However, it is not always clear if the extra color is actually necessary.  
\end{remark}

% Related to the conflict-free coloring problem is the classical graph coloring problem.
% In this problem, given a graph~$G$, we would like to minimize the number of colors required to color every vertex in $G$ such that adjacent vertices have different colors.
% The solution to this problem is called the \emph{chromatic number} of $G$, denoted by $\chi(G)$.
Recall the proper coloring defined in Section~\ref{sec:intro}.
For a graph~$G$, denote by $\chi(G)$ the chromatic number of $G$.
Observe that such a proper coloring gives a \cfcn{} coloring, but in general, the \cfcn{} chromatic number is much lower than the chromatic number. For example, $\chi(K_n)=n$ but $\chi_{CN}(K_n)=2$ where $K_n$ is a clique on $n$ vertices.

\section{Clique-width}\label{sec:cliquewidth} 
In this section, we study conflict-free coloring on graphs of bounded clique-width. 

\begin{definition}[Clique-width~\cite{courcellecw}]
    Let $w \in \mathbb{N}$.
    A $w$-expression $\Phi$ defines a graph $G_\Phi$ where each vertex receives a 
    label from $[w]$. The graph consisting of a solitary vertex $v$ with label $i$ has the $w$-expression $v(i)$. Graphs that contain two or more vertices are defined inductively using the three operations described below.   Let $G_{\Phi'}$ and $G_{\Phi''}$ be graphs given by the $w$-expressions $\Phi'$ and $\Phi''$ respectively.

\begin{enumerate}
    \item Disjoint union: The graph $G_{\Phi}$ which is the disjoint union of $G_{\Phi'}$ and $G_{\Phi''}$ is given by the $w$-expression
    $\Phi = \Phi' \oplus \Phi''$.

	\item  Relabel: Let the graph $G_{\Phi}$
            be $G_{\Phi'}$ where each vertex labeled $i$ in $G_{\Phi'}$ is relabeled with the label $j$. The graph $G_{\Phi}$ is given by the $w$-expression
    $\Phi = \rho_{i\rightarrow j}(\Phi')$.
    
    \item Join: Let the graph $G_{\Phi}$ obtained from  $G_{\Phi'}$ by adding edges between all the vertex pairs
    $(u,v)$, where $u$ has label $i$ and $v$ has label $j$. The graph $G_{\Phi}$ is given by the $w$-expression $\Phi=\eta_{i,j}(\Phi')$.
	\end{enumerate}

The \emph{clique-width} of a graph $G$ denoted by \emph{cw(G)} is the minimum number $w$
    such that there is a $w$-expression $\Phi$ that defines $G$.
\end{definition}

Given a graph $G=(V,E)$ 
%with clique-width $w$ and a decomposition of $G$ corresponding to a 
and its $w$-expression 
%that defines~$G$,\todo{Hung: I didn't understand what the decomposition refers to. So added a little bit of context} 
%$w$-expression of $G$its clique decomposition, 
it was shown in \cite{iwoca-sri} that all the four variants of the conflict-free coloring problem 
(\cfon{}, \cfcn{}, \cfonp{} and \cfcnp{}) %problems 
can be solved in  time
$ 2^{O(w3^k)} n^{O(1)}$,  where $w$ is the clique-width of the graph, $k$ is the number of colors, and $n$ is the number of vertices of $G$. 
Recently, Bergougnoux, Dreier and Jaffke in \cite{logic-mim-width} introduced a logic called \emph{distance neighborhood} (DN) 
logic which extends existential MSO$_1$. %within which the problems can be expressible. 
%The DN logic permits the use of predicates for querying neighborhoods of vertex sets. 
It was shown in the same paper that the CFON and CFCN coloring problems 
%all the four variants of conflict-free coloring problem 
can be expressed in DN logic. Using similar ideas, the CFON$^*$ and CFCN$^*$ coloring problems can also be expressed in DN logic.
%For example, the \cfon{} coloring can be expressed in DN logic as follows\footnote{We do not provide an explanation for the DN logic formulation 
%of CFON coloring given in equation (\ref{eq:DNlogic}). This involves explaining the
%terminology involved and takes us away from the focus of our paper. The 
%interested reader is referred to \cite{logic-mim-width} for more details.}: 
%\begin{equation}
%\label{eq:DNlogic}
%    \exists X_1, X_2, \dots, X_k  \mbox{ partition}(X_1, X_2, \dots, X_k) \land 
%\bigcup\limits_{i\in [k]}N_1^1(X_i)\setminus N_2^1(X_i) = V(G). 
%\end{equation}
%\todo[]{Check the formula in the paper. It is written a bit differently. }
By applying 
%the algorithmic meta-theorem 
%stated as 
Theorem 1.1  in  \cite{logic-mim-width}, we obtain an algorithm that runs in time $2^{O(wk^2)}n^{O(1)}$. 
%\todo{Hung: Can we give theorem/section numbers of this paper, so that the reader can verify? Also, is it just one theorem? \bscomment{It is implied by Theorem 1.1 from the arxiv version of the paper. }}
Thus, all the variants of conflict-free coloring 
are fixed-parameter tractable when parameterized by the clique-width of the graph and the number of colors. 
As a consequence, we have the following.  

\begin{theorem}\label{thm:fpt}
Given a $w$-expression of a graph $G$, all the four variants of the conflict-free coloring problem 
(\cfon{}, \cfcn{}, \cfonp{} and \cfcnp{}) can be solved in time $2^{O(wk^2)}n^{O(1)}$ where $k$ is the number of colors and $n$ is the number of vertices of $G$. 
\end{theorem}

\subsection{Graphs of bounded clique-width and unbounded $\chi_{CN}$ and $\chi_{ON}$}

\newcommand{\BB}{\mathcal{B}}
Since the \cfcn{} and  \cfon{} coloring problems are FPT when parameterized by clique-width and the number of colors, 
%\todo[]{recheck after defining the ``problems''. Also we need to add somewhere that we study only full coloring variants in this section.} 
an open question is then whether there exists an FPT algorithm with respect to only the clique-width.
One solution to this question would be to bound the \cfon{} and \cfcn{} chromatic numbers by a function of the clique-width. 
%Unfortunately this approach does not yield an FPT algorithm, 
However, this turns out to be impossible, 
even for graphs of clique-width three. 
% Finally we also observe that the \cfon{} and the \cfcn{} chromatic numbers are unbounded even for graphs of clique-width three. 
%\bscomment{does it work for the partial variants as well by making some changes? Though asymptotically it remains the same for partial variants. We need to add this somewhere.} 
%\skcomment{I feel it is better to not distract the reader with this discussion.}
We construct graphs $G_2, G_3,\dots, G_k$ of clique-width at most 3 such 
that a conflict-free coloring of $G_i$ requires at least $i$ colors.
Interestingly, graphs of clique-width at most $2$, i.e.,\ cographs (see \cite{courcellecw} for a reference), have bounded \cfon{} and \cfcn{} chromatic numbers, as shown in Theorems~\ref{lem:dist_here_CN_restr} and~\ref{lem:cograph_ON} in the next section.
In the following, we consider the full coloring variant. 
%without the color $0$.
%TODO justify that this is ok
Let us first consider \cfcn{} colorings.

\begin{theorem}
\label{thm:unbounded_cn}
For any given integer $k$ $\geq 2$, there exists  a graph $G_k$ of clique-width at most 3 with $\chi_{CN}(G_k) \geq k$.
%\textcolor{red}{There is a series of graphs $G_k,k \geq 2,$ of clique-width $\leq 3$ but $\chi_{CN}(G_k) \geq k$.}
\end{theorem}

\begin{figure}
\vspace{-0.3cm}
\begin{center}
\begin{tikzpicture}
[scale=0.7,auto=left, node/.style={circle,fill=white, draw, scale=0.5}
	,max/.style={circle,fill=black, draw, scale=1}
	,bubble/.style={ellipse,fill=white, draw, scale=1}]
    
	\foreach \x in {1,...,4}{
	    \node[node] (b\x) at (\x,0) {};
	}
	\foreach \x in {1,...,4}{
	    \foreach \y in {1,...,4}{
	        \draw ({b\x}) to[out=-90, in=-90] ({b\y});
	    }
	}
	
	\node[bubble] (s1) at (0,2) {$G_2 \quad G_2$};
	\node[bubble] (s2) at (5,2) {$G_2 \quad G_2$};
	\node[bubble] (s12) at (2.5,3) {$G_2 \quad G_2$};
	
	\foreach \x in {1,...,4}{
	    \draw (s12) -- (b\x);
	}
	\foreach \x in {1,2}{
		   \draw (s1) -- (b\x);
	}
	\foreach \x in {3,4}{
		   \draw (s2) -- (b\x);
	}	
	
\end{tikzpicture}
\begin{tikzpicture}
[scale=0.7,auto=left, node/.style={circle,fill=white, draw, scale=0.5}
	,max/.style={circle,fill=black, draw, scale=1}
	,bubble/.style={ellipse,fill=white, draw, scale=0.75}]
    
	\foreach \x in {1,...,8}{
	    \node[node] (b\x) at (\x,0.5) {};
	}
	\foreach \x in {1,...,8}{
	    \foreach \y in {1,...,8}{
	        \draw (b\x) to[out=-90, in=-90] (b\y);
	    }
	}

	\node[bubble] (s1) at (1,-2) {$G_3 \quad G_3$};
	\node[bubble] (s2) at (3,-3) {$G_3 \quad G_3$};
	\node[bubble] (s3) at (6,-3) {$G_3 \quad G_3$};
	\node[bubble] (s4) at (8,-2) {$G_3 \quad G_3$};
    
	\foreach \x in {1,2} {
	    \draw (s1) -- (b\x);
	}
	\foreach \x in {3,4} {
	    \draw (s2) -- (b\x);
	}
	\foreach \x in {5,6} {
	    \draw (s3) -- (b\x);
	}
	\foreach \x in {7,8} {
	    \draw (s4) -- (b\x);
	}

	\node[bubble] (s12) at (1,2) {$G_3 \quad G_3$};
	\node[bubble] (s34) at (8,2) {$G_3 \quad G_3$};
	\node[bubble] (s1234) at (4.5,3) {$G_3 \quad G_3$};
	
	\foreach \x in {1,...,8}{
	    \draw (s1234) -- (b\x);
	}
	\foreach \x in {1,...,4}{
		   \draw (s12) -- (b\x);
	}
	\foreach \x in {5,...,8}{
		   \draw (s34) -- (b\x);
	}	
	
\end{tikzpicture}
\end{center}
\vspace{-0.3cm}
\caption{$G_3$ (left) and $G_4$ (right) have clique-width $3$ but cannot be \cfcn{} colored with $2$ and $3$ colors, respectively.
Each $G_i$,$i\geq 2$ stands for a copy of the graph $G_i$.
Every vertex in an ellipse is adjacent to every vertex that is connected to that ellipse.
}
\label{figure:not:cw:bounded}

\vspace{-0.5cm}
\end{figure}

\begin{proof}
We construct graphs $G_i$, $i\geq 2$ inductively. 
The graph $G_{k+1}$ is such that it cannot be \cfcn{} colored with $k$ colors.
Thus at least $k+1$ colors are required.
\begin{itemize}
\item
Let $G_2$ be the graph isomorphic to $K_2$.
\item
The graph $G_{k+1}$, for $k \geq 2$, is constructed as follows. It consists of $2^k$ \emph{bottom vertices} $B = \{b_0,\dots,b_{2^k-1}\}$ and $2(2^k-1)$ copies of $G_k$. The  vertices of $B$ form a clique. 
To describe the edges between the vertices in the copies of $G_k$ and those in $B$, it will be simpler
to consider an imaginary binary tree $T$.
Let $T$ be the full binary tree with $k$ levels and with leaves $B$.
That is, $T$ consists of $k+1$ levels $L_0,\dots, L_k$,
    where level $L_i$ contains $2^{k-i}$ vertices $b_0^i,\dots,b_{2^{k-i}-1}^i$ for $0 \leq i \leq k$.
Each vertex $b_j^i$ has children $b_{2j}^{i-1}$ and $b_{2j+1}^{i-1}$ for $1 \leq i \leq k$ and $0 \leq j < 2^{k-i}$.
Then we identify the bottom vertices $B$ with the leaves $L_0$,
    which is $b^0_j = b_j$ for $0\leq j < 2^k$.
For a non-leaf $x$ of $T$, let $B(x) \subseteq B$ be the set of descendants of $x$ among the leaves $B$.
Let $\BB = \{ B(x) \mid x \in V(T) \setminus L_0 \}$ be the family of such sets.
For every set $S \in \BB$, introduce two disjoint copies of $G_k$ and make them adjacent to $S$, i.e., all the vertices in the two copies of $G_k$ are 
adjacent to all the vertices in $S$.
%\bscomment{May be better if we add this- Each inner node $x$ of $T$ at level $k-i$ has $|B(x)|=2^i$, where $1\leq i\leq k$. Also should we give details on how the neighborhood is chosen ?}
See Fig.~\ref{figure:not:cw:bounded} for illustrations of $G_3$ and $G_4$.
\end{itemize}

Inductively we show that $G_k$ has clique-width at most $3$.
That is, there is a $3$-expression $\Phi_k$ where $G_{\Phi_k}$ equals $G$ when ignoring the labels.
We will use the labels $\{\alpha,\beta,\gamma\}$ instead of numbers, since numbers are already used for colors.
\begin{itemize}
\item
Graph $G_2$, a single edge, can be constructed using 2 labels.
\item
Consider the graph $G_{k+1}$.
By the induction hypothesis, there is a $3$-expression $\Phi_k$ that describes $G_k$. 
We may assume that every vertex of $G_{\Phi_k}$ has label $\beta$ since we can apply relabelling operations at the end.
Let vertex sets $B$ and $T$ with levels $L_0,\dots,L_k$ be as in the construction of $G_{k+1}$.
%That is $B = \{b_1,\dots,b_{2^{k}}\}$,
%    and $T$ is the binary tree with $k$ levels and leaves $B$.
We show the following properties for every node $x \in L_i$ of $T$
    by induction on the level $i = 0,\dots,k$:
\begin{itemize}
	\item[(*)]
	There is a 3-expression $\Phi_{k+1,x}$ where $G_{\Phi_{k+1,x}}$
	equals the induced subgraph of $G_{k+1}$ that contains $B(x)$ and the copies of $G_k$ whose neighborhoods 
    are subsets of $B(x)$; and
	\item[(**)]
	$B(x)$ has label $\alpha$ and the copies of $G_k$ have label $\gamma$. %\bscomment{Above we said vertices of $G_k$ have label $\beta$}\thcomment{What I mean: Wlog we may use an expression for the recursively defined graph $G_k$ with all labels $\beta$.
	%Here, in the middle of constructing $G_{k+1}$ the labeling of some used $G_k$ may have changed in the meantime.}
\end{itemize}
Then $\Phi_{k+1,r}$, where $r$ is the root of $T$, is the desired $3$-expression.
\begin{itemize}
	\item
	For the induction basis, let $i=0$ and $x \in L_0$. Hence $x$ is a leaf $b_i$:
	Simply introduce the single vertex of label $\alpha$.
	\item
	For the induction step, let $i \geq 1$ and $x \in L_i$.
	The vertex $x$ has two children, say $y$ and $z$, in level $L_{i-1}$.
	Thus by induction hypothesis there are $3$-expressions $\Phi_{k+1,y}$ and $\Phi_{k+1,z}$
	    with the properties (*) and (**) described above.

%    \begin{itemize}
%        \item We do a relabel operation on $\Phi_{k+1,y}$ to get $\rho_{\alpha \to \beta}( \Phi_{k+1,y} )$.
%        \item We do a disjoint union: $\rho_{\alpha \to \beta}( \Phi_{k+1,y} ) \oplus \Phi_{k+1,z}$.
%        \item On the graph obtained above, we do a join operation $\eta_{\alpha, \beta}$ that 
%        adds all edges between $B(y)$ and $B(z)$.
%        \item We do a further relabel operation and relabel all the vertices assigned the label 
%        $\beta$ to $\alpha$. 
%        \item We now introduce two copies of $G_k$ with label $\beta$, by induction. We use the disjoint union operation to combine with the graph constructed above.
%        \item We do a join operation $\eta_{\alpha, \beta}$.
%        \item Finally we relabel all the vertices labelled $\beta$ to $\gamma$.
%    \end{itemize}

    We construct $\Phi_{k+1,x}$: We first need to add all the edges between $B(y)$ and $B(z)$ by combining the respective 3-expressions.
	Towards this end, we do (i) a relabel operation $\rho_{\alpha \to \beta}( \Phi_{k+1,y} )$, (ii) disjoint union of $\rho_{\alpha \to \beta}( \Phi_{k+1,y} )$  and $\Phi_{k+1,z}$, (iii) a join operation $\eta_{\alpha, \beta}$  on the graph obtained, and (iv) relabel all the vertices assigned the label $\beta$ to $\alpha$. 

    We now need to introduce two copies of $G_k$ and add edges between the introduced copies of $G_k$ and the vertices of $B(x)= B(y) \cup B(z)$. Towards this end, we (i) inductively construct two copies of $G_k$, (ii) relabel the vertices of these copies of $G_k$ to $\beta$, (iii) take a disjoint union 
    of these copies of $G_k$ and the graph constructed above on $B(x)$, and (iv) use a join operation
    $\eta_{\alpha, \beta}$ on the resulting graph.
    Finally, we relabel all the vertices assigned the label $\beta$ to $\gamma$ to maintain the 
    property (**). 
    \end{itemize}
\end{itemize}

Lastly, we show by induction that $G_{k+1}$ has no \cfcn{} coloring with only $k$ colors, for every $k\geq 1$. For the induction basis, consider $G_2$, a single edge.
There, a $1$-coloring is not possible.

For the induction step, $G_k \leadsto G_{k+1}$,
    suppose for a contradiction
    that there is a CFCN coloring $c: V(G_{k+1}) \to \{1,\dots,k\}$.
    
We first show that each set $S \in \BB$ contains a uniquely colored vertex $f(S)$.
%    \todo[]{“Formally, there is a mapping” is a bad transition, I
%would say “For doing so, we prove the existence of a mapping”. TH: sounds good} 
To be precise,
    the mapping $f: \BB \to B$
    such that for each set $S \in \BB$ there is a vertex $f(S)=v \in S$
    and $c(v)\neq c(v')$ for every other vertex $v' \in S \setminus \{v\}$. 

Recall that $G_{k+1}$ contains two copies $C_1,C_2$ of $G_k$
    where each $C_i, i \in \{1,2\}$ has $N[C_i] \setminus C_i = S$.
Now suppose for a contradiction that $S$ contains no uniquely colored vertex.
Let $c_i$ be the coloring $c$ restricted to vertices $V(C_i)$, for $i \in \{1,2\}$.
Then $c_i$ is a CFCN coloring of graph $C_i$, for $i \in \{1,2\}$.
Indeed by induction hypothesis, the restricted coloring $c_i$ 
is surjective.
%maps to all $k$ colors.
Hence in $V(C_1) \cup V(C_2)$ each of the $k$ colors occurs twice.
Then every vertex in $u \in S$ has every color from $\{1,\dots,k\}$ at least twice in its neighborhood.
This contradicts the claim that $u$ has a uniquely colored neighbor.
Therefore, each set $S \in \BB$ contains a uniquely colored vertex $f(S)$.

%\todo[]{“Now, by symmetry we may assume that the uniquely col-
%ored element of set B is f(B) = 2k = 1”. This is only a suggestion,
%but I think it is more understandable to say the uniquely colored element f(B) in the set B belong to {2k-1,...., 2
%k = 1} and is colored
%with the color k.}
Now, without loss of generality we may assume that the uniquely colored element of the set $B$ is $f(B)=b_{2^k-1}$ and that $b_{2^k-1}$ is colored with color $k$.
Then the subset $\{b_0,\dots,b_{2^{k-1}-1}\} \in \BB$ consists only of vertices of color $1,\dots,k-1$.
Again without loss of generality, we may assume that 
%this subset $\{b_0,\dots,b_{2^{k-1}-1}\}$ has uniquely colored element 
$f(\{b_0,\dots,b_{2^{k-1}-1}\}) = b_{2^{k-1}-1}$ and that the vertex $b_{2^{k-1}-1}$ is colored with ${k-1}$.
By repeating this argument, we eventually obtain that $b_0$ and $b_1$ 
%$f(\{b_0, b_1\})$ \todo[]{\bscomment{we should remove $f$ here ? and want to say both $b_0$ and $b_1$ take the color 1} TH: yes ``that $b_0$ and $b_1$ must''} 
must take the color $1$.
This contradicts the claim that $\{b_0,b_1\}\in \BB$ has a uniquely colored element.
Therefore, $G_{k+1}$ cannot be colored with just $k$ colors. 
%\todo[]{I think {b1, b2} in the two last sentences should be replaced by {b0, b1}. 
%\bscomment{Changed}}
\qed
\end{proof}

To show that the \cfon{} coloring number is also unbounded for graphs with clique-width three,
we use an analogous approach.
    We define a sequence of graphs $G_2',G_3',\dots$, such that each graph $G_{k+1}'$ for $k\geq 2$ has clique-width at most three and cannot be \cfon{} colored with  $k$ colors.
    Let $G_2'$ be a copy of $K_3$, which cannot be CFON colored with  one color
    and which has clique-width 2.
    We use the same inductive process to construct $G_{k+1}'$ from the copies of $G_k'$.
    Inductively it follows that $G_{k+1}'$ has clique-width at most $3$.
Again, by the same induction step as before, it follows that $G_{k+1}'$ cannot be \cfon{} colored with  $k$ colors.
We also provide an alternative construction in Lemma~\ref{lem:dist_here_ON}.

\begin{theorem}
\label{thm:unbounded_on}
For any given integer $k$ $\geq 2$, there exists  a graph $G_k$ of clique-width at most 3 with $\chi_{ON}(G_k) \geq k$.
%   \textcolor{red}{There is a series of graphs $G_k,k \geq 2,$ of clique-width $\leq 3$ but $\chi_{ON}(G_k) \geq k$.}
\end{theorem}

%This section derived several bounds on the conflict-free coloring number
%    for several graph classes of bounded clique-width.
%Notably, this is not 

%\input{trees.tex}
\newcommand{\CC}{\mathcal{C}}

\section{Graph classes of bounded clique-width}
\label{sec:dist_here}
% \bscomment{general comment: should we explain things for partial version and add a corollary saying that the full coloring variant can also be satisfied ? I see many places where partial and full is used. Especially in this section}
% \hhcomment{Right! I think we can say something in the second last paragraph in the preliminary section, when we talked about this. But I'm not very sure what exact statements we would like to make. In the meantime, I've added separate statements for the full coloring variant.}
% \skcomment{The statement is there in the preliminaries, maybe we can add a 
% remark/observation that full coloring chromatic number is at most 1 more than the partial one.}\bscomment{Added it in the preliminaries}
One consequence of Theorem \ref{thm:fpt} 
%the FPT results with clique-width and number of colors 
%Theorem~\ref{thm:cwonp} (or Theorem \ref{thm:cwcn}) 
is that if both the clique-width and the \cfonp{} (or \cfcnp{}) chromatic number
    of the input graph is bounded,
    then there exists a polynomial time algorithm to solve the \cfonp{} (or \cfcnp{}, respectively) coloring problem.
Theorems~\ref{thm:unbounded_cn} and~\ref{thm:unbounded_on} show that even when the clique-width is at most 3, the \cfonp{} and \cfcnp{} chromatic numbers can be unbounded.
Hence, this section explores some graph classes with clique-width at most 3,
    where the \cfonp{} or \cfcnp{} chromatic number is bounded.

Firstly, we consider the graphs with clique-width at most 2, which are exactly the cographs~\cite{courcellecw}.

\begin{definition}[Cograph \cite{corneil1981complement}]
A graph $G$ is a \emph{cograph} if it can be constructed recursively by the following rules.
An isolated vertex is a cograph, the disjoint union of two cographs is a cograph
    and the complement of a cograph is a cograph.
\end{definition}
We will show that cographs have \cfcnp{} and \cfonp{} chromatic numbers at most 2 (Lemmas~\ref{lem:dist_here_CN_restr} and~\ref{lem:cograph_ON}).

These graphs are a special case of distance hereditary graphs, whose clique-width is at most 3~\cite{DistHere_Golumbic}.

\begin{definition}[Distance hereditary graph \cite{DistHere_Howorka}]
A graph $G$ is \emph{distance hereditary} if for every connected induced subgraph $H$ of $G$, the distance (i.e., the length of a shortest path) between any pair of vertices in $H$ is the same as that in $G$.
\end{definition}

% Let $d_G(x,y)$ denote the length of the shortest path connecting vertices $x$ and $y$ in the graph $G$. 
% A graph $G$ is called \emph{distance hereditary} if for every connected induced subgraph $H$ of $G$, $d_G(x,y) = d_H(x,y)$ holds for every pair of vertices from $H$.

Bandelt and Mulder~\cite{DistHere_Bandelt} gave the following alternative definition of 
    connected distance hereditary graphs.
For a given ordering of the vertices $(v_1,v_2,\dots,v_n)$ of $V(G)$,
    let $G[i]$ be the induced subgraph of $G$ on $\{v_1, \dots, v_i\}$.
The sequence $(v_1, v_2, \dots, v_n)$ is a \emph{one-vertex extension sequence} if $G[2] = K_2$, and for every $i\geq 3$, $G[i]$ can be formed by adding $v_i$ to $G[i-1]$
    and edges incident to $v_i$
    such that for some $j < i$, one of the following holds:
\begin{itemize}
    \item $v_i$ is adjacent to $v_j$ and no other vertex (we say $v_i$ is a \emph{pendant} of $v_j$);
    \item $v_i$ is adjacent to all the neighbors of $v_j$ (we say $v_i$ is a \emph{false twin} of $v_j$); or
    \item $v_i$ is adjacent to $v_j$ and all the neighbors of $v_j$ (we say $v_i$ is a \emph{true twin} of $v_j$).
\end{itemize}
Then a connected graph is distance hereditary
    if and only if there exists a one-vertex extension sequence $(v_1, v_2, \dots, v_n)$.

Note that if the pendant operation is absent, then we obtain exactly the cographs.
In other words, cographs are exactly the distance hereditary graphs that can be constructed from a single vertex by the true twin and false twin operations~\cite{DistHere_Bandelt}.
If the true twin operation is absent, then we obtain bipartite distance hereditary graphs.
Lastly, if the false twin operation is missing,
    we obtain a graph class that contains block graphs as a subclass~\cite{block_markenzon}.

\begin{definition}[Block Graph~\cite{diestel2005graph}]
A \emph{block graph} is a graph in which every 2-connected component (i.e., a maximal subgraph which cannot be disconnected by the deletion of one vertex) is a clique. 
\end{definition}

%In this section, we discuss the~\cfcnp{} chromatic number of distance hereditary graphs and some subclasses.
%It turns out that the distance hereditary graphs in general have a bounded \cfcnp{} chromatic number but an unbounded~\cfonp{} chromatic number.  
%bounded.

\subsection{\cfcnp{} chromatic number}
We first show an upper bound for the \cfcnp{} chromatic number of distance hereditary graphs.

\begin{lemma}
\label{lem:dist_here_CN}
If $G$ is a distance hereditary graph, then $\chicnp \leq 3$.
\end{lemma}

\begin{proof}
Suppose $(v_1, v_2, \dots, v_n)$ is a one-vertex extension sequence of $G$.
We will give an iterative algorithm to provide a \cfcnp{} coloring with colors 0, 1, 2, 3.
% We will also specify a mapping $U: V(G) \to \{1, 2, 3\}$ and will show that this corresponds to the color of the uniquely colored neighbor of each vertex. 
% We will also specify a mapping $h: V(G) \to V(G)$ and will show that for every vertex $v$ of $G$, $h(v)$ is the uniquely colored neighbor of $v$. 

% Since we assume that the graph is connected, $v_1$ and $v_2$ are connected.
% We consider $v_2$ as a true twin of $v_1$.

% We introduce some terminologies.
% In other words, at the end of iteration $t-1$, $u$ is still a dangling pendant.
% If $v_i$ is a false twin of a dangling pendant of some vertex~$v_t$ in $G[i]$, then we treat~$v_i$ as a pendant of $v_t$ instead.
% In other words, the second neighbor of a pendant cannot be its false twin.
We use~$N_i(v)$ and~$N_i[v]$ to refer to the open and closed neighborhoods of a vertex~$v$ in the graph~$G[i]$, respectively, where $i \in [n]$ is the current iteration of the algorithm. 

% For a graph $H$ and a subset $S$ of $V(H)$, denote by $H - S$ the graph resulting from deleting vertices in $S$ from $H$.

% Consider the following directed tree $T$ whose vertex set is $V(G)$.
% For $u$ is a follower of $v$, we add a directed edge $(v, u)$ to the tree and label the edge either ``pendant", ``true twin", or ``false twin", depending on what kind of follower $u$ is.

% Consider the maximal subtree $T_1$ that contains $v_1$ without any pendant edges.
% If we delete all vertices of $T_1$ from $G$, then each connected component of the resulting graph $G - V(T_1)$ corresponds to a subtree (other than $T(v_1)$) of $F$.
% Indeed, two roots of two subtrees are two pendants, so by construction, they are not connected.
% Other vertices in the subtrees 

% For $v \in V$, denote by $T(v)$ the subtree that contains $v$ in the forest above.
% For a subtree~$T$ rooted at vertex $v_i$, we denote by $B(T)$ the tuple $(a,b) \in \{1,2,3\}^2$ such that $U(v_i) = a$, and $U(v) = b$ for $v \in T \setminus \{v_i\}$.

For each vertex~$v$, we specify a tuple $\CC(v) = (a,b)$
    with $a \in \{0, 1, 2, 3\}$ as the color of $v$,
    and $b \in \{1, 2, 3\}$ as the color of the uniquely colored neighbor of $v$.
We maintain the following two invariants for each iteration $i \in [n]$ of the coloring algorithm:

\begin{itemize}
	\item \textbf{Invariant 1}:
    For every vertex $v \in G[i]$,
        with $\CC(v)=(a,b)$, possibly $a=b$,
    there is a uniquely colored neighbor of $v$ in $G[i]$
        of color $b$.
	\item \textbf{Invariant 2}:
    For every vertex $v \in G[i]$,
    if $\CC(v) = (a,a)$, then condition (*) or condition (**) is true.
    \begin{itemize}
        \item Condition (*):
        There is a color $d\in\{1,2,3\}\setminus\{a\}$,
        such that every vertex $w \in N_i(v)$ has $\CC(w)=(0,d)$.
        \item Condition (**):
        There is a color $y\in \{1,2,3\} \setminus \{a\}$ that appears exactly once in $N_i(v)$.
    \end{itemize}
\end{itemize}

We are now ready to describe the coloring scheme.
%\skcomment{Can $v_2$ be a false twin of $v_1$? Does the defn of
%one vertex extension sequence insist that G[i] is connected?}
Recall that $v_1$ and $v_2$ are adjacent to one another.
We assign $\CC(v_1) = (1,2)$ and $\CC(v_2) = (2,1)$.
Clearly, the Invariants 1 and 2 hold.
For $i \geq 3$, consider $j$ such that $v_i$ is either a pendant, false or true twin of $v_j$.
Let $\CC(v_j) = (a,b)$ for some $a \in \{0, 1, 2, 3\}$ and $b \in \{1, 2, 3\}$.
We distinguish the following cases:

\begin{itemize}
    \item \textbf{Case 1a: $v_i$ is a pendant of $v_j$ and $a = b$.}
    \begin{itemize}
        \item
        \textbf{Case 1a$'$: $\CC(w) = (0,d)$ for all $w \in N_{i-1}(v_j)$ and $d \neq a$.} That is, condition (*) holds for $v_j$ in $G[i-1]$.\\
        We assign $\CC(v_i) = (x, a)$, where $x$ is the color in $\{1, 2, 3\} \setminus \{a, d\}$.
        \item
        \textbf{Case 1a$''$: otherwise.} That is, condition (**) holds for $v_j$ in $G[i-1]$.\\
        We assign $\CC(v_i) = (0,a)$.
    \end{itemize}
    \item \textbf{Case 1b: $v_i$ is a pendant of $v_j$ and $a \neq b$.}
    \begin{itemize}
        \item
        \textbf{Case 1b$'$: $a \neq 0$.}\\
        We assign $\CC(v_i) = (0,a)$.
        \item
        \textbf{Case 1b$''$: $a=0$.}\\
        We assign $\CC(v_i) = (x, x)$, for an arbitrary color $x$ in $\{1, 2, 3\} \setminus \{b\}$.
    \end{itemize}
    
    \item
    \textbf{Case 2a: $v_i$ is a true twin of $v_j$ and $a = b$.} 
    \begin{itemize}
        \item \textbf{Case 2a$'$: $\CC(w) = (0,d)$ for all $w$ in $N_{i-1}(v_j)$ and $d \neq a$.}
        That is, condition (*) holds for $v_j$ in $G[i-1]$.\\
         We assign $\CC(v_i) = (x, a)$, where $x$ is the color in $\{1, 2, 3\} \setminus \{a, d\}$.
         \item
         \textbf{Case 2a$''$: otherwise.} That is, condition (**) holds for $v_j$ in $G[i-1]$.\\
         We assign $\CC(v_i) = (0,a)$.
    \end{itemize}
    \item
    \textbf{Case 2b: $v_i$ is a true twin of $v_j$ and $a \neq b$.}\\
    We assign $\CC(v_i) = (0,b)$.
    
    \item
    \textbf{Case 3a: $v_i$ is a false twin of $v_j$ and $a = b$.}
    \begin{itemize}
        \item \textbf{Case 3a$'$: $\CC(w) = (0,d)$ for all $w$ in $N_{i-1}(v_j)$ and $d \neq a$.}
        That is, condition (*) holds for $v_j$ in $G[i-1]$.\\
         We assign $\CC(v_i) = (a,a)$.
         \item \textbf{Case 3a$''$: otherwise.} That is, condition (**) holds for $v_j$ in $G[i-1]$.\\
        % Notice that Invariant 2 for $v_j$ in $G[i-1]$ is satisfied by condition (**).
         That is, there is a vertex $w \in N_{i-1}(v_j)$
         with a unique color $y\in \{1, 2, 3\} \setminus \{a\}$ among the vertices in $N_{i-1}(v_j)$.
         We assign $\CC(v_i) = (0,y)$.
    \end{itemize}  
    \item
    \textbf{Case 3b: $v_i$ is a false twin of $v_j$, and $a \neq b$.}\\
    We assign $\CC(v_i) = (0,b)$.
\end{itemize}

We prove the invariants by induction.
Invariant 1 for iteration $i = n$ implies that the coloring above is a \cfcnp{} coloring for $G = G[n]$.

These invariants are trivially true for the base case of $i = 2$.
For the inductive step, observe that for any vertex $u \notin N_i[v_i]$,
    there is no change in the closed neighborhood of $u$, and hence the invariants hold for $u$ by the inductive hypothesis.
For $v_j$ and $N_i[v_i]$ we show that Invariants 1 and 2 are
    satisfied in $G[i]$.

%In the cases where $\CC(v_i) = (0,z)$ for $z \neq 0$,
%    Invariant 2 is (trivially) true for $v_i$. 
%Further, for a vertex $w \in N_i(v_i)$, the color 0 of $v_i$ does not interfere with the uniquely colored neighbor of $w$ (i.e., Invariant 1 is satisfied for $w$).
%In addition, if $\CC(w) = (y, y)$ for some $y$,
%    then regardless of whether condition (*) or (**) satisfies Invariant 2 of $w$ in $G[i-1]$,
%    the same is true for $w$ in $G[i]$.
%
%If condition (*) satisfies Invariant 2 of $w$ in $G[i-1]$,
%    then all neighbours of $w$ in $G[i-1]$ have color 0.
%Since $v_i$ is also colored 0,
%    condition (*) still holds for $w$ in $G[i]$.
%((Actually $\CC(w)=(0,d)$ needs to be show for every neighbor
%    has to be shown.))
%
%If condition (**) satisfies Invariant 2 of $w$ in $G[i-1]$,
%    then there is a unique neighbour $u$ of $w$ in $G[i-1]$ with color $x$.
%Since $v_i$ is colored 0, this neighbour $u$ remains the unique neighbour of $w$ with color $x$ in $G[i]$.
%Hence, condition (**) still satisfies Invariant 2 of $w$ in $G[i]$.
%
%Therefore, in summary,
%\begin{itemize}
%    \item If we assign the color 0 to the new vertex $v_i$, then we only need to check that Invariant 1 hold for $v_i$.
%    \item Otherwise, we need to verify the two invariants for all vertices in $N_i[v_i]$.
%\end{itemize}
\begin{description}
\item[Case 1a$'$.]
%We have to verify the invariants for $v_i$ and $v_j$.
%For Case 1a, in the first subcase (when $\CC(w) = (0,d)$ for all $w$ in $N_{i-1}(v_j)$ and $d \neq a$), we need to verify the invariants for $v_i$ and $v_j$.
\begin{description}
    \item[Vertex $v_i$:] 
Invariant 1 holds, since $v_j$ with color $a$ is the uniquely colored neighbor of $v_i$. Invariant 2 is vacuously true.
    \item[Vertex $v_j$:] Since $v_i$ has a different color than $a$, the uniquely colored neighbor of $v_j$ remains unchanged (in fact, it is $v_j$ itself), i.e., Invariant 1 holds. 
% In the first subcase (when $\CC(w) = (0,d)$ for all $w$ in $N_{i-1}(v_j)$ and $d \neq a$), 
Invariant 2 holds by condition (**),
    because $v_i$ becomes the only vertex in $N_i(v_j)$ with color $x$ assigned to it.
% Hence, the invariants hold for $v_j$.
\end{description}

\item[Case 1a$''$.]
\begin{description}
     \item[Vertex $v_i$:] Invariant 1 is satisfied since $v_i$ has $v_j$ with color $a$ as neighbor.
Invariant 2 is vacuously true.
    \item[Vertex $v_j$:] Invariant 1 is satisfied, since $v_j$ still has itself with color $a$ as its uniquely colored neighbor.
    %Condition (**) holds for $v_j$ in $G[i-1]$. That is, 
    By the case assumption, there exists a vertex $w\in N_{i-1}(v_j)$ that is assigned a color from $\{1, 2, 3\} \setminus \{a\}$ and $w$ is
    uniquely colored among the vertices in $N_{i-1}(v_j)$. The vertex $w$ continues to be uniquely colored
    in $N_{i}(v_j)$ since $v_i$ is colored 0. Hence Invariant 2 holds by condition (**) for $v_j$ in $G[i]$.
%    \item[Old, Do we need this:]
%    \textcolor{red}{If condition (**) was true for $v_j$ in $G[i-1]$,
%    then $a=d$ (as we are not in Case 1a$'$)
%    and also the $v_i$ has $\CC(0,d)$
%    such that condition (**) still holds for $v_j$ in $G[i]$.
%Otherwise, condition (*) was true for $v_j$ in $G[i-1]$.
%    Then by assigning color $0$ to $v_i$,
%    any uniquely non-zero colored vertex in $N_i(v_j)$ for $v_j$
%    remains having this property,
%    such that condition (*) still holds for $v_j$ in $G[i]$.}
\end{description}

%\bscomment{Argue the invariants 1 and 2 for $v_j$. Invariant 1 is trivial. 
%For invariant 2, (*) is not satisfied for $v_j$. 
%By induction, we know that (**) is true for $v_j$ in $G[i-1]$. 
%Hence the unique vertex remains the same now. 
%}
%\hhcomment{I have the blanket statement in the paragraph before that if $v_i$ has color 0, we only need to check Invariant 1 for $v_i$ and don't need to check the invariants for the other vertices (since they still hold by inductive hypothesis)}

\item[Case 1b$'$.]
\begin{description}
     \item[Vertex $v_i$:] Invariant 1 is satisfied, since $v_i$ has $v_j$ of color $a\neq 0$ as its neighbor. Invariant 2 is vacuously true.
    \item[Vertex $v_j$:] The vertex $v_j$ retains its uniquely colored neighbor,
    since $v_i$ is colored with $0$. Thus Invariant 1 holds. Invariant 2 is vacuously true.
\end{description}

\item[Case 1b$''$.]
\begin{description}
     \item[Vertex $v_i$:] Invariant 1 is satisfied, since $v_i$ has itself as its uniquely colored neighbor.
Further, condition (*) for Invariant 2 holds since $b\neq x$.
    \item[Vertex $v_j$:] Invariant 1 is satisfied,
    since $v_i$ is not colored $b$. Invariant 2 is vacuously true.
\end{description}

\item[Case 2a$'$.]
\begin{description}
    \item[Vertex $v_i$:] Invariant 1 is satisfied for $v_i$,
    since $v_j$ is its only neighbor with color $a$. Invariant 2 is vacuously true.
    \item[Vertex $v_j$:] Invariants 1 and 2 (due to condition (**)) 
    are true. The arguments are the same as those in the Case 1a$'$.
    \item[Vertices in $N_i(v_j) \setminus \{v_i\}$:] Invariant 1 remains true for these vertices, 
    since $v_i$ is colored with $x$, and $x \neq d$. Invariant 2 is vacuously true.
\end{description}

\item[Case 2a$''$.]
\begin{description}
    \item[Vertex $v_i$:] Invariant 1 holds since $v_j$ has color $a$
    and $v_j$ is uniquely colored among the vertices in $N_{i-1}[v_j]$.
    Invariant 2 is vacuously true.
    \item[Vertex $v_j$:] Invariants 1 and 2 (due to condition (**)) 
    are true. The arguments are the same as those in the Case 1a$''$.
    \item[Vertices in $N_i(v_j) \setminus \{v_i\}$:] Invariant 1 holds since $v_i$ is colored 0. 
    For $w \in N_i(v_j) \setminus \{v_i\}$, if $\mathcal C(w) = (x, x')$, where $x \neq x'$, the Invariant 2 is vacuously true.
    For $w \in N_i(v_j) \setminus \{v_i\}$, if $\mathcal C(w) = (x, x)$ for some $x \in \{1, 2, 3\}$, the Invariant 2 held by condition (**) in $G[i-1]$ (since $\mathcal C(v_j) = (a, a)$). 
    Invariant 2 continues to hold in $G[i]$ by condition (**) since $v_i$ is colored 0. 
\end{description}

\item[Case 2b.]
\begin{description}
    \item[Vertex $v_i$:] Invariant 1 holds since $N_i[v_j] = N_i[v_i]$ and since $v_i$ is colored 0. Invariant 2 is vacuously true.
    \item[Vertex $v_j$:] Invariant 1 holds since since $v_i$ is colored 0. Invariant 2 is vacuously true.
    \item[Vertices in $N_i(v_j) \setminus \{v_i\}$:] Invariant 1 holds since since $v_i$ is colored 0.
    For $w \in N_i(v_j) \setminus \{v_i\}$, if $\mathcal C(w) = (x, x')$, where $x \neq x'$, the Invariant 2 is vacuously true. Suppose that for $w \in N_i(v_j) \setminus \{v_i\}$, if $\mathcal C(w) = (x, x)$ for some $x \in \{1, 2, 3\}$. 
    
    If Invariant 2 held by condition (*) for $w$ in $G[i-1]$, it means that for all $w' \in N_{i-1}(w)$, we have $\mathcal C(w') = (0, b)$ (since $\mathcal{C}(v_j) = (a,b)$ by the case assumption). Invariant 2 continues to hold by condition (*) for $w$ in    $G[i-1]$ since $\mathcal C(v_i) = (0,b)$. 

    If Invariant 2 held by condition (**) for $w$ in $G[i-1]$, then it continues to hold  by condition (**) in $G[i]$ 
    since $v_i$ is colored 0.
\end{description}

\item[Case 3a$'$.]
\begin{description}
    \item[Vertex $v_i$:] Invariants 1 and 2 hold for $v_i$ since $N_i(v_i) = N_i(v_j) = N_{i-1}(v_j)$.
    \item[Vertex $v_j$:] Invariants 1 and 2 continue to hold as the neighborhood of $v_j$ is unaffected by the false twin operation. 
    \item[Vertices in $N_i(v_j)$:] Invariant 1 is true for $w \in N_i(v_j)$ since $a \neq d$. Invariant 
    2  is vacuously true.
\end{description}

\item[Case 3a$''$.]
\begin{description}
    \item[Vertex $v_i$:] By the case assumption, $w \in N_{i-1}(v_j) = N_i(v_i)$ is the unique vertex 
    in $N_i(v_i)$ that is colored $y$. Hence Invariant 1 holds.  Invariant 
    2  is vacuously true.
    \item[Vertex $v_j$:]  Invariants 1 and 2 continue to hold as the neighborhood of $v_j$ is unaffected by the false twin operation. 
    \item[Vertices in $N_i(v_j)$:]  Invariants 1 and 2 
    are true. The arguments are identical to those in the Case 2a$''$.
\end{description}

\item[Case 3b.]
\begin{description}
    \item[Vertex $v_i$:] Note that by the case assumption, there is a unique vertex $w \in N_{i-1}(v_j)$ that is colored $b$. Since $N_i(v_i) = N_{i-1}(v_j)$ and since $v_i$ is colored 0, Invariant 1 holds. Invariant 2 is vacuously true.
    \item[Vertex $v_j$:] Invariants 1 and 2 continue to hold as the neighborhood of $v_j$ is unaffected by the false twin operation. 
    \item[Vertices in $N_i(v_j)$:] Invariants 1 and 2 are true. The arguments are 
    identical to those in the Case 2b.
\end{description}
\end{description}
\qed

\end{proof}

Since the \cfcnp{} chromatic number  of distance hereditary graphs is at most 3,
    its \cfcn{} chromatic number is at most 4.
Hence the algorithm of Theorem~\ref{thm:fpt} combined with these bounds
    provides the following result:
%Together with the facts that distance hereditary graphs have clique-width at most 3 and that the chromatic numbers for the full and partial conflict-free coloring differ by at most 1, 
%Theorem \ref{thm:fpt} 
%Theorems~\ref{thm:cwcn},~\ref{thm:cwcn_fullcoloring} 
%and Lemma~\ref{lem:dist_here_CN} imply the following corollary.

\begin{corollary}
    For distance hereditary graphs, the \cfcnp{} and \cfcn{} coloring problems are polynomial time solvable. 
\end{corollary}

In the following lemma, we show that we need fewer colors when we restrict the operations used to construct the distance hereditary graphs.

\begin{lemma}
\label{lem:dist_here_CN_restr}
Let $G$ be distance hereditary graph
    defined by a one-vertex extension sequence
    which only uses two of the three operations:
adding a pendant vertex, adding a false twin and adding a true twin.
Then $\chicnp \leq 2$.
%For a distance hereditary graph $G$, if $G$ can be built from a vertex without one of the pendant, true twin, and false twin operations, then $\chicnp \leq 2$.
In particular, this holds for cographs and block graphs.
\end{lemma}
\begin{proof}
In principle, we use the same construction and case distinction
    as in Lemma~\ref{lem:dist_here_CN}.

If the pendant operation is absent (i.e., $G$ is a cograph), observe that Cases 1a and 1b do not occur.
The only remaining case that assigns $\CC(v_i)=(x,x)$
    for some color $x \in \{1,2,3\}$ is Case 3a$'$.
However, the prerequisite of this case is that $\CC(v_j)=(a,a)$ for some $a \in \{1,2,3\}$.
Hence, starting with $\CC(v_1) = (1,2)$ and $\CC(v_2) = (2,1)$,
   by induction, Case 3a never occurs
    and we never assign $\CC(u) = (x,x)$
    for some color $x \in \{1,2,3\}$ to any vertex $u$.
This means that only Cases 2b and 3b occur,
    and therefore all vertices other than $v_1$ and $v_2$ are colored 0.
Hence, 2 colors suffice for a \cfcnp{} coloring.
Notice that $v_1$ and $v_2$ have each other as their uniquely colored neighbor.
Further, all other vertices have $v_1$ or $v_2$
    as their uniquely colored neighbor.
Thus, this \cfcnp{} coloring is also a \cfonp{} coloring.

If the true twin operation is absent, the graph is bipartite.
We can color one part of the bipartition with color 1 and the other part with color 2.
Since all vertices with the same color are not adjacent to each other, each vertex is its own uniquely colored neighbor.

If the false twin operation is absent (this subclass includes the block graphs), we modify the coloring scheme as follows.
Recall that $v_1$ and $v_2$ are adjacent to one another.
We assign $\CC(v_1) = (1,2)$ and $\CC(v_2) = (2,1)$.
%\skcomment{Wondering why (1,2) and (2,1) won't work? It may 
%be better for the sake of exposition to be consistent with the proof 
%of Lemma 12.}
%\hhcomment{Sure. I've changed as you suggested and also add a remark below that these case distinctions are essentially similar to Case 1b and 2b in the above proof.}

For $i \geq 3$, we consider two cases, where we assume $\CC(v_j) = (a,b)$ for some $a \in \{0, 1, 2\}$ and $b \in \{1, 2\}$:

\begin{itemize}
    \item
    \textbf{Case 1: $v_i$ is a pendant of $v_j$.}
    \begin{itemize}
        \item
        \textbf{Case 1$'$: $a \neq 0$.}\\
        We assign $\CC(v_i) = (0,a)$.
        \item
        \textbf{Case 1$''$: $a=0$.}\\
        We assign $\CC(v_i) = (x, x)$, for an arbitrary color $x$ in $\{1, 2\} \setminus \{b\}$.
    \end{itemize}
    \item 
    \textbf{Case 2: $v_i$ is a true twin of $v_j$.}\\
    We assign $\CC(v_i) = (0,b)$.
\end{itemize}

Note that the color assignments above are similar to those in Case 1b and Case 2b in the proof of Lemma~\ref{lem:dist_here_CN}.

We prove by induction that at the end of every iteration $i \in [n]$, every vertex has a uniquely colored neighbor in $G[i]$.
This holds for the base case $i = 2$.
For the inductive step, it is easy to see that if $v_i$ has color 0, then we only need to show the claim for $v_i$, and otherwise, we have to show the claim also for all vertices in $N_i[v_i]$ (recall that this refers to the neighborhood of $v_i$ in $G[i]$).
%With this knowledge, we go through the two cases.

\begin{description}
    \item[Case 1$'$.] The vertex $v_j$ is the uniquely colored neighbor of $v_i$.

    \item[Case 1$''$.] The vertex $v_i$ is its own uniquely colored neighbor. 

    It remains to consider $v_j$, the only other vertex in $N_i[v_i]$.
    As $\CC(v_j)=(0,b)$ and $b\neq x$, $v_j$ retains its uniquely colored neighbor from $G[i-1]$.

    \item[Case 2.] In this case, $N_i[v_i] = N_i[v_j]$. Hence $v_i$ and $v_j$ share the same uniquely colored neighbor whose color is $b$.
\end{description}
\qed
\end{proof}

%Consider Case 1$'$, that $v_i$ is a pendant of $v_j$ and $a \neq 0$.
%Then $v_j$ is the uniquely colored neighbor of $v_i$.

%Consider Case 1$''$, that $v_i$ is a pendant of $v_j$ and $a=0$.
%Then $v_i$ is its own uniquely colored neighbor. 
%\textcolor{red}{Is there a reason why $v_j$ is not made the uniquely colored neighbor of $v_i$ ?} \hhcomment{In the case when $a = 0$, $v_j$ is uncolored and hence cannot be the uniquely colored neighbour of $v_i$.}

%Consider Case 2, that $v_i$ is true twin of $v_j$.
%Then $N_i[v_i] = N_i[v_j]$, and hence $v_i$ and $v_j$ share the same uniquely colored neighbor whose color is $b$.

From Lemma \ref{lem:dist_here_CN_restr}, we have that $\chi^*_{CN}(G)\leq 2$, when $G$ is a block graph. This bound is tight when $G$ is a bull graph, 
illustrated in Figure \ref{figure:bull-graph}. To see that there is no \cfcnp{} coloring that uses
only one color, observe that the degree 2 vertex in $G$ necessitates that exactly one of the vertices 
of the $K_3$ subgraph is colored. A simple case analysis completes the proof.

\begin{figure}
\vspace{-0.3cm}
\begin{center}
\begin{tikzpicture}
[scale=0.5,auto=left, node/.style={circle,fill=black, draw, scale=0.3}
	,max/.style={circle,fill=black, draw, scale=0.5}]

	\node[node] (m) at (-.5,3) {};
	\node[node] (l3) at (-3,1) {};
    \node[node] (xl3) at (-5,2) {};

	%\node[node] (r1) at (1,-1) {};
	\node[node] (r3) at (2.5,1) {};
	\node[node] (xr3) at (4.5,2) {};

\draw (m) -- (l3) -- (r3) -- (m); 
\draw (l3) -- (xl3); 
\draw (r3) -- (xr3);

\end{tikzpicture}
\end{center}
\vspace{-0.3cm}
\caption{Bull graph $G$ with $\chi^*_{CN}(G)=2$. }
\label{figure:bull-graph}
\vspace{-0.5cm}
\end{figure}

\subsection{\cfonp{} chromatic number}
%\skcomment{The previous section is titled CFCN* chromatic number, so I changed this section to CFON* too.}
%\hhcomment{Thanks. I wanted to change but forgot.}

In contrast to the closed neighborhood setting, the class of distance hereditary graphs has unbounded \cfon{} chromatic number and consequently also unbounded \cfonp{} chromatic number.
%In this section, we consider the \cfonp{} chromatic number of distance hereditary graphs and some subclasses.

\begin{lemma}
\label{lem:dist_here_ON}
For any $k \geq 1$, there exists a bipartite distance hereditary graph $G$ such that $\chi_{ON}(G) \geq k$.
\end{lemma}

\begin{proof}
\begin{figure}
    \centering
    \includegraphics{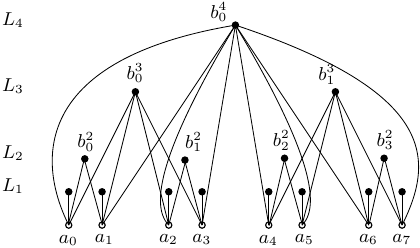}
    \caption{A bipartite distance hereditary graph $G_4$ with $\chi_{ON}(G_4) \geq 4$.}
    \label{fig:dist_here}
\end{figure}

%\skcomment{I am wondering if we can re-do the above figure without
%colors. Can we do it using dotted or dashed lines?}
%\hhcomment{Sure. I've changed the figure and caption.}

We define a family of graphs $G_2, G_3, \dots$ as follows.
Each graph $G_{k}$, for $k\geq 2$, is bipartite with the vertex sets 
%two parts 
$A$ and $B$ 
%of the bipartition 
that satisfy the following:
\begin{itemize}
    \item Set $A$ consists of $2^{k-1}$ vertices $a_0, \dots, a_{2^{k-1}-1}$.
    \item Set $B$ consists of vertices in $k$ levels $L_1, \dots, L_k$. Level $L_i$ contains $2^{k-i}$ vertices $b^i_0, \dots, b^i_{2^{k-i}-1}$, for $i \in [k]$.
    \item There are $2^{k-1}$ edges between each level $L_i$ and $A$ in a binary fashion. To be precise, the vertex $b^i_j$ is connected with vertices $a_t$ for $t = 2^{i-1} j, \dots, 2^{i-1} (j+1) - 1$.
\end{itemize}

Figure~\ref{fig:dist_here} illustrates the graph $G_4$. 
We construct $G_k$ recursively, starting from the graph 
of only one vertex called the root. 
%The construction assumes that (from the recursion) we are given a vertex $a_0$, which we denote as the \emph{root}.
%To construct $G_k$, we may start from $a_0$ and do the below construction with $a_0$ as the root.
Our construction satisfies the property $(\star)$: 
Every vertex of $A \setminus \{a_0\}$
    is indirectly a false twin of $a_0$;
    which means that it
    is created by a sequence of false twin operations
    on some vertices $u_0,u_1,\dots,u_z$, $z\geq 1$
    where $u_0=a_0$, and each $u_i$, $i \geq 1$, is introduced as a false twin of $u_{i-1}$.
    The construction is as follows:
\begin{itemize}
    \item For $k=2$, we add a pendant to the root, i.e.,
    $G_2$ is isomorphic to $K_2$.
    \item For $k \geq 3$, as mentioned above, we call $a_0$ the root.
    We add $b^{k}_0$ as a pendant of $a_0$.
    Next, we add a false twin of $a_0$, called $a_{2^{k-2}}$.
Then, recursively construct a $G_{k-1}$ with root $a_0$
    and another $G_{k-1}$ with root $a_{2^{k-2}}$.
The first $G_{k-1}$ introduces $a_1,\dots,a_{2^{k-2}-1}$,
    the second $G_{k-1}$ introduces $a_{2^{k-2}+1},\dots,a_{2^{k-1}-1}$.
The property $(\star)$ holds for these copies of $G_{k-1}$. That is, every vertex of $A\setminus \{a_0, a_{2^{k-2}}\}$ is created
    indirectly as a false twin of $a_0$ or $a_{2^{k-2}}$.
Since $a_{2^{k-2}}$ is created as a false twin of $a_0$,
    property $(\star)$ holds for $G_k$.
Further, because $(\star)$ is true,
    $b^k_0$ is adjacent to all of $A$.
Thus, we have constructed $G_{k}$.
\end{itemize}

We will show that the \cfon{} chromatic number of $G_{k}$ is at least $k$.
This holds trivially for $k = 2$. 
Consider the case where $k \geq 3$.
Observe that $b^k_0$ needs to have a neighbor with a unique color.
Without loss of generality, we color $a_{2^{k-1}-1}$ with the color $c_k$.
Next, $b^{k-1}_0$ also needs a neighbor with a unique color. Note that this color must be different than $c_k$, because all neighbors of $b^{k-1}_0$ are neighbors of $b^k_0$, while $a_{2^{k-1}-1}$ is not a neighbor of $b^{k-1}_0$.
Without loss of generality, we color $a_{2^{k-2}-1}$ with the color $c_{k-1}$.
Repeating the above argument, we can see that we need at least $k$ colors.
\qed
\end{proof}

% Observe that the construction in the proof above contains only pendants and false twins.
% It is naturally to ask whether the \cfon{} chromatic number is bounded, if one of the two types is absent.

% Firstly, if the pendant operation is absent, then we have exactly cographs.
% In other words, cographs are exactly distance-hereditary graphs that can be constructed from a single vertex by the true twin and false twin operations~\cite{DistHere_Bandelt}.

Notice that the number of vertices in the graph $G_k$ constructed above is $\Theta(2^{k})$. Since $\chi_{ON}(G_k)\geq k$, we have the following corollary. 

\begin{corollary}\label{cor:dist-here-lowerbound}
    There exists a distance hereditary graph $G$ on $n$ vertices for which 
    $\chi_{ON}(G)=\Omega(\log n)$. 
    %family of graphs that re
\end{corollary}

Although in general, a distance hereditary graph can have arbitrarily large \cfonp{} chromatic number, we show that this number is bounded for two subclasses, as in the following two lemmas.

\begin{lemma}
\label{lem:cograph_ON}
If $G$ is a cograph, then $\chionp{} \leq 2$. 
\end{lemma}
\begin{proof}
As observed in the proof of Lemma~\ref{lem:dist_here_CN_restr}, the coloring scheme there gives 
a \cfonp{}{} coloring with the colors $\{0, 1, 2\}$.
In this coloring, $v_1$ has color 1, $v_2$ color 2, and all other vertices color 0.
% Since it is a valid \cfcnp{} coloring, the uniquely colored neighbor of these other vertices must be either $v_1$ or $v_2$.
% Further, $v_1$ and $v_2$ are the uniquely colored neighbors of each other.
% Hence, this is also a valid \cfonp{} coloring.
% We have $v_1$ and $v_2$ are connected.
% We color $v_1$ with color 1 and $v_2$ with color 2 and leave the other vertices with color 0.
% %Note that $v_1$ is the uniquely colored neighbor of $v_2$ and $v_2$ is the uniquely colored neighbor of $v_1$.
% As observed earlier, only the true twin and false twin operations are required to construct $G$. 
% Consider a vertex $v$, such that there exists a sequence $(u_0 = v_1, u_1, u_2, \dots, v_{\ell} = v)$ for some $\ell$, such that $u_i$ is a twin of $u_{i-1}$ for $i \in [\ell]$.
% By an easy induction, we have $u_i$ is adjacent to $v_2$ for $i \in [\ell]$.
% Hence, $v_2$ is a uniquely colored neighbor for $v$.
% By a similar argument, we have $v_1$ is the uniquely colored neighbor for the other vertices, if such sequence above does not exist.
% Also observe that $v_1$ is also a uniquely colored neighbor of $v_2$.
% Therefore, $\chionp{}, \chicnp{} \leq 2$.
\qed
\end{proof}

\begin{lemma}\label{lem:blockon2}
If $G$ is a block graph, $\chion \leq 3$,
    hence $\chionp \leq 3$.
\end{lemma}
\begin{proof}
Our proof is by induction on $|V|$.
Trivially, if $|V| \leq 3$, then $\chion \leq 3$. 
For the inductive step, if $G$ is 2-connected, then by definition of a block graph, $G$ is a clique.
We can color two vertices with two different colors and all other vertices with the third color.
It is easy to see that this is a  \cfon{} coloring. 

Now suppose $G$ is not 2-connected.
Then there exists a vertex $v$ whose removal disconnects the graph, and a connected component $C$ satisfies that $V(C) \cup \{v\}$ induces a 2-connected component in $G$, i.e., a clique. 
(This component is sometimes called a leaf block, for example, in~\cite{block_west}.)

Consider the induced subgraph $G'$ of $G$ obtained by removing $V(C)$ from $G$.
It is easy to see that $G'$ is also a block graph.
Hence, applying the inductive hypothesis, we can obtain a \cfon{} coloring of $G'$ with 3 colors.
Let $c_1$ be the the color of $v$ and $c_2$ be the color of its uniquely colored neighbor.
We apply the same coloring of $G'$ to the vertices in $G$, where we additionally color all vertices in $C$ with the color other than $c_1$ and $c_2$.
Certainly, this does not invalidate the uniquely colored neighbor of $v$.
No other vertex in $G'$ is adjacent to a vertex of $C$ in $G$. 
Further, all vertices in $C$ have $v$ as their uniquely colored neighbor.
Hence, this is a \cfon{} coloring of $G$ with 3 colors.
\qed
\end{proof}
%\bscomment{Is it redundant ?}
%\hhcomment{The above is not needed in the narrative. I think it's a cute result, and I'm fine with it's either here or not. We could also leave it here and see what the reviewer says}
We show that the above result is tight.

\begin{lemma}
There is a block graph $G$ with $\chionp = 3$.
\end{lemma}

\begin{figure}
\vspace{-0.3cm}
\begin{center}
\begin{tikzpicture}
[scale=0.7,auto=left, node/.style={circle,fill=white, draw, scale=0.5}
	,max/.style={circle,fill=black, draw, scale=1}]

	\node[node] (m) at (0,0) {};
	\node[node] (l) at (-1,1) {};
	\node[node] (l1) at (-1,-1) {};
	\node[node] (l2) at (-3,-1) {};
	\node[node] (l3) at (-3,1) {};
	\node[node] (xl1) at (-1,-2) {};
	\node[node] (xl2) at (-4,-2) {};
	\node[node] (xl3) at (-4,2) {};

	\node[node] (r) at (1,1) {};
	\node[node] (r1) at (1,-1) {};
	\node[node] (r2) at (3,-1) {};
	\node[node] (r3) at (3,1) {};
	\node[node] (xr1) at (1,-2) {};
	\node[node] (xr2) at (4,-2) {};
	\node[node] (xr3) at (4,2) {};
	
	\node [above] at (m.north) {$m$};
	\node [above] at (l.north) {$\ell$};
	\node [right] at (l1.east) {$x^\ell_3$};
	\node [left] at (l2.west) {$x^\ell_2$};
	\node [above] at (l3.north) {$x^\ell_1$};
	\node [left] at (xl1.west) {$\overline{x}^\ell_3$};
	\node [left] at (xl2.west) {$\overline{x}^\ell_2$};
	\node [left] at (xl3.west) {$\overline{x}^\ell_1$};
	
	\node [above] at (r.north) {$r$};
	\node [left] at (r1.west) {$x^r_3$};
	\node [right] at (r2.east) {$x^r_2$};
	\node [above] at (r3.north) {$x^r_1$};
	\node [right] at (xr1.east) {$\overline{x}^r_3$};
	\node [right] at (xr2.east) {$\overline{x}^r_2$};
	\node [right] at (xr3.east) {$\overline{x}^r_1$};

	\foreach \s in {l,r}{
		\foreach \a in {m,\s,\s1,\s2,\s3}{
			\foreach \b in {m,\s,\s1,\s2,\s3}{
				\draw (\a) -- (\b);
			}
		}
		\foreach \a in {\s1,\s2,\s3}{
			\draw (\a) -- (x\a);
		}
	}
\end{tikzpicture}
\end{center}
\vspace{-0.3cm}
\caption{A block graph $G$ with $\chionp = 3$.}
\label{figure:block:graph}
\vspace{-0.5cm}
\end{figure}
\begin{proof}
Let $G$ have vertex set $V = \{\ell,m,r\} \cup \bigcup_{i \in \{1,2,3\}} \{x^\ell_i, {\overline{x}^\ell_i}, x^r_i, {\overline{x}^r_i} \}$, see also Fig. \ref{figure:block:graph}.
Let the edge set be defined by the set of maximal cliques $\{x^s_1, x^s_2, x^s_3, s, m \}$ and  $\{x^s_i, \overline{x}^s_i \}$ for every $s\in \{\ell,r\}$ and $i \in \{1,2,3\}$.
It is easy to see that $G$ is a block graph.
To prove that $\chionp > 2$, suppose for a contradiction
    that there is a $\chionps$ coloring $C : V \to \{0,1,2\}$.
Then there is a mapping $h$ on $V$ that assigns each vertex $v \in V$ its uniquely colored neighbor $w \in N(v)$.
Note that $x^s_i$, for $s\in \{\ell,r\}$ and $i \in \{1,2,3\}$, has to be colored $1$ or $2$, since it is the only neighbor of $\overline{x}^s_i$.
Without loss of generality,
we may assume that $h(m) \in \{\ell, x^\ell_1\}$.
%    because of the symmetry of the $\ell$-indexed vertices and $r$-indexed vertices.
Further, we may assume that $C(h(m))=2$.

First suppose that $h(m)=\ell$ and hence $C(\ell)=2$.
Then $(x^s_i)=1$ for every $s \in \{\ell,r\}$ and $i \in \{1,2,3\}$.
It follows that $h(\ell)=m$ and hence $C(m)=2$.
Then the vertex $x_1^\ell$ contains two neighbors colored $1$
    and two neighbors colored $2$,
    which is a contradiction. 

Thus it remains to consider that $h(m)=x^\ell_1$ and hence $C(x^\ell_1)=2$.
%\textcolor{red}{Then $C(x^s_i)=1$
 %   for every vertex $x^s_i \neq x^\ell_1$
 %   where $s \in \{\ell,r\}$
 %   and $i \in \{1,2,3\}$.} 
Then $C(w)=1$ for every vertex $w \in \{x_2^{\ell}, x_3^{\ell}, x_1^r, x_2^r, x_3^r\}$.
It follows that $h(r)=m$ and $C(m)=2$.
Then the vertex $\ell$ has two neighbors colored $1$
    and two neighbors colored $2$,
    which is a contradiction. 
%$\cm(N(\ell)) = \ms{1,1,2,2}$, also a contradiction.

Since both cases lead to a contradiction, it must be that $\chionp > 2$. 
\qed
\end{proof}

Together with the fact that distance hereditary graphs have clique-width at most 3, Theorem \ref{thm:fpt}, 
%Theorems~\ref{thm:cwonp},~\ref{thm:cwcn_fullcoloring}, 
Lemma~\ref{lem:cograph_ON} and Lemma~\ref{lem:blockon2} imply the following corollary.

\begin{corollary}
    For cographs and block graphs, the \cfonp{} and \cfon{} coloring problems are polynomial time computable.
\end{corollary}
\section{Interval Graphs}\label{sec:interval}

In this section, 
we consider interval graphs.
We prove that three colors are sufficient and sometimes necessary to \cfonp{} color an interval graph.
For \emph{proper interval graphs}, we show that two colors are sufficient. 
%Note that the results for \cfcnp{} coloring have been shown by Fekete and Keldenich~\cite{Sandor2017}.

%On the algorithmic side, we show that the \cfcnp{} problem is polynomial time solvable on interval graphs.  

\begin{definition}[Interval Graph]
A graph $G$, $G = (V,E)$, is an \emph{interval graph} if there exists a set $\mathcal I$ of closed intervals
on the real line such that there is a bijection $f: V \rightarrow \mathcal I$ satisfying the following: $v_1v_2 \in E$ if and only if $f(v_1) \cap f(v_2) \neq \emptyset$.
%Let $\mathcal I$ be a set of intervals. The graph formed by the intervals in $\mathcal I$ is called an \emph{interval graph} if there exists a vertex for each interval and two vertices are adjacent if their corresponding intervals intersect. 
\end{definition}
For an interval graph $G$, we refer to the set of intervals $\mathcal I$ as an 
\emph{interval representation} of $G$. An interval graph $G$ is a \emph{proper interval graph} if it has an interval 
representation $\mathcal I$ such that
no interval in $\mathcal I$ 
is properly contained in any other interval of 
$\mathcal I$. An interval graph $G$ is a \emph{unit interval graph} if it has an interval 
representation $\mathcal I$ where all the intervals are of unit length. It 
is known that the class of proper interval graphs and unit interval graphs are identical 
\cite{gardi}. 

For each interval $I \in \mathcal I$, we use
$L(I)$ and $R(I)$ to denote its left endpoint and right endpoint respectively. 
Throughout this section, we assume that no two intervals share the same endpoint (either left or right). If there exists two intervals that share an endpoint, 
%with one of the endpoint being the same for the other, then 
we can carefully adjust them such that they do not share the same endpoint. 
We use the terms ``vertex'' and ``interval'' interchangeably. 

It was shown in  \cite{Sandor2017} that $\chicnp{} \leq 2$, when $G$ is an interval graph.  
We use  similar ideas to show the bound for $\chionp{}$. 

\begin{lemma}\label{lem:intervalon}
If $G$ is an interval graph, then $\chionp{}\leq 3$. 
\end{lemma}

\begin{proof}
Let $G$ be an interval graph and $\mathcal I$ be an interval representation of $G$. 
%For each interval $I \in \mathcal I$, we use
%$R(I)$ to denote its right endpoint.
%arranged according to their left 
%endpoints. 
%Each interval $I_i$ has two endpoints namely the left endpoint and the right endpoint 
%denoted by $L(I_i)$ and $R(I_i)$ respectively. 
We use the function $C:\mathcal I\rightarrow \{1, 2, 3, 0\}$ to assign colors. We assign the colors 1, 2 and 3 cyclically. 
%\todo{Vinod: We are assigning colors to the intervals instead of vertices. May be we can add one sentence at the starting of the coloring.}
%$j\geq 1$. 
We start with the interval $I_1$ for which $R(I_1)$ is the least and assign $C(I_1)=1$. 
Then choose the interval $I_2$ such that $I_2 \in N(I_1)$ 
and $R(I_2)> R(I), \forall I \in N(I_1)$ and assign $C(I_2)=2$.  
%We color the remaining vertices in iterations. Let ${I_{i-1}}$ be the colored interval in the iteration $i-1$. 
For $j\geq 3$, 
we do the following: 
choose the interval $I_j$ such that $I_j \in N(I_{j-1})$ 
and $R(I_j)> R(I), \forall I \in N(I_{j-1})$ and 
assign the color $\{1, 2, 3\}\setminus \{C(I_{j-1}), C(I_{j-2})\}$ to the interval $I_j$.  
The procedure terminates at the value $j$ for which 
%Note that the interval $I_{\ell}$ is chosen in the last iteration $\ell$, such that
$R(I_{j})$ maximizes  $R(I)$ amongst all $I \in \mathcal{I}$. We refer to this value of $j$ as $\ell$.
%We continue till we have $I_j$ as the interval that maximizes $R(I)$, for all $I \in \mathcal I$.
Now we assign 0 to all the uncolored intervals.

Since $G$ is connected and because of the coloring procedure, 
the graph induced on the intervals $I_1, I_2, \ldots, I_{\ell}$ is a path. 
%there is a path of intervals assigned the colors 1, 2 and 3. 
For $1\leq j\leq \ell -1$, the interval $I_{j+1}$ will be a uniquely colored neighbor for the interval $I_{j}$. The  interval $I_{\ell-1}$ will be a uniquely colored neighbor for $I_{\ell}$. %\todo[]{Fix the word ``see''}

Consider an interval $I$ assigned the color 0. Recall that the intervals $I_1, I_2, \ldots, I_{\ell}$ induce a path.
This implies that $I$ is adjacent to an interval $I_j$, where  $1 \leq j \leq \ell$ and $R(I_j) > R(I)$. We claim that $I_j$ will be a uniquely colored
neighbor of $I$.
%Without loss of generality, let $C(I_j)=1$. 
%We note that $I$ is not adjacent to any other interval which assigned the same color as $I_j$. 
Assume for the sake of contradiction that $I$ is adjacent to $I_{j-3}$, the vertex that was assigned the color $C(I_j)$ immediately before $I_j$. This implies that $I$ is adjacent to $I_{j-2}$ and $I_{j-1}$ as well. Since the graph induced by 
$I_1, I_2, \ldots, I_{\ell}$ is a path, 
%$I_{j-3}$ is not adjacent to $I_{j-1}$ and  
$I_{j-2}$ is not adjacent to $I_{j}$. Since $I \in N(I_j)$, it follows
that $R(I) > R(I_{j-2})$. This contradicts the coloring 
procedure as we must have chosen $I$ in place of $I_{j-2}$. Thus $I$ is not adjacent to $I_{j-3}$.
By a similar argument, we can see that $I$ is not adjacent to $I_{j+3}$ as well.
%For $I$ to not have a uniquely colored neighbor, there exists an interval $I''\in N(I)$ such that $C(I'')=1$. Then $I$ must be neighboring intervals assigned the colors 2 (say $I'''$) and 3 as well. This will be a contradiction as the coloring procedure must have picked $I$ for assigning the color 2 and not $I'''$. 
%Each interval assigned the color 0 is adjacent to an interval assigned a color from $\{1, 2, 3\}$.  
%$I_j$, $1\leq j\leq \ell$, which will 
%serve 
%as its uniquely colored neighbor. 
\qed
\end{proof}

The bound of $\chionp{}\leq 3$ for interval graphs is tight.
In particular, there is an interval graph $G$ (see Figure~\ref{figure:interval:lower:bound}) 
that cannot be \cfon{} colored with three colors. 
%That shows the stronger result $\chion{}>3$, which
This implies that $\chionp{}>2$. 
%We present the proof of Lemma \ref{lem:int-lb} in Appendix \ref{app:inter_app}. 

%$\chion{} > 3$ implies $\chionp{} > 2$.

\newcommand{\igap}{0.1}
\newcommand{\itips}{0.08}
\newcommand{\intervallc}[4]{
	\draw (#1,#3+\itips)--(#1,#3-\itips) [thick,draw=#4];
	\draw (#1,#3) -- (#2,#3) [thick,draw=#4]; %actual itnerval
	\draw (#2,#3+\itips)--(#2,#3-\itips) [thick,draw=#4];
	}
\newcommand{\intervall}[3]{
	\intervallc{#1}{#2}{#3}{black}
	}
\newcommand{\intervalll}[3]{
	\intervall{#1+\igap}{#2-\igap}{#3}
	}
\newcommand{\intervallltriple}[3]{
	\intervalll{#1}{#2}{#3}
	\intervalll{#1}{#2}{#3+0.2}
	\intervalll{#1}{#2}{#3+0.4}
	}
\newcommand{\intervalllquadruple}[3]{
	\intervalll{#1}{#2}{#3}
	\intervalll{#1}{#2}{#3+0.2}
	\intervalll{#1}{#2}{#3+0.4}
	\intervalll{#1}{#2}{#3+0.6}
	}
	
\begin{figure}
%\vspace{-0.5cm}
\begin{center}
\begin{minipage}{0.4\textwidth}
\centering
%\begin{tikzpicture}[every node/.style={node distance=1.5cm,scale=0.5}, scale = 0.6]
\begin{tikzpicture}[scale = 0.6]
\tikzstyle{vertex}=[circle,draw, minimum size=3pt,node distance=1.5cm,scale=0.5]
\tikzstyle{edge} = [draw,thick,-,black]
\node[vertex]  (u1) at (1,0) {};
\node[vertex]  (u2) at (2,0) {};
\node[vertex]  (u3) at (3,0) {};
\node[vertex]  (u4) at (4,0) {};
\node[vertex]  (u5) at (5,0) {};
\node[vertex]  (u6) at (6,0) {};
\node[vertex]  (u7) at (7,0) {};
\node[vertex]  (u8) at (8,0) {};
\node[vertex]  (v1) at (2,1.5) {};
\node[vertex]  (v2) at (4.5,1.5) {};
\node[vertex]  (v3) at (7,1.5) {};
\draw[edge, color=black] (u1) -- (v1) (u2) -- (v1) (u3) -- (v1) ;
\draw[edge, color=black] (u3) -- (v2) (u4) -- (v2) (u5) -- (v2) (u6) -- (v2) ;
\draw[edge, color=black] (u6) -- (v3) (u7) -- (v3) (u8) -- (v3) ;
\draw[edge, color=black] (v1) -- (v2);
\draw[edge, color=black] (v2) -- (v3);
\node [above] at (v1.north) {$u$};
\node [below] at (u1.south) {$u'$};
\node [below] at (u2.south) {$u''$};
\node [below] at (u3.south) {$u^{\star}$};
\node [above] at (v2.north) {$v$};
\node [below] at (u4.south) {$v'$};
\node [below] at (u5.south) {$v''$};
\node [below] at (u6.south) {$w^{\star}$};
\node [below] at (u7.south) {$w'$};
\node [below] at (u8.south) {$w''$};
\node [above] at (v3.north) {$w$};
\end{tikzpicture}
\end{minipage}
\begin{minipage}{0.4\textwidth}
\centering
\begin{tikzpicture}[scale=1]
	\intervallltriple{1}{2.5}{1.9}
	
	\intervallltriple{3.5}{5}{1.9}
	
	\intervalllquadruple{1}{1.5}{1.1}
	\intervalllquadruple{1.5}{2}{1.1}
	\intervallltriple{2}{2.5}{1.3}
	\intervalllquadruple{2.5}{3}{1.3}
	\intervalllquadruple{3}{3.5}{1.3}
	\intervallltriple{3.5}{4}{1.3}
	\intervalllquadruple{4}{4.5}{1.1}
	\intervalllquadruple{4.5}{5}{1.1}

	\intervallltriple{2}{4}{0.7}
	
\end{tikzpicture}
\end{minipage}
%\vspace{-0.3cm}
\end{center}
\caption{On the left hand side, we have the graph $G'$, and on the right hand side we have
an interval representation of $G$, a graph in which $\chion{} > 3$.
The graph $G$ is obtained by adding two true twins each to the vertices $u,v,w,u^\star,v^\star$
of $G'$ %with a 3-clique 
and adding three true twins each to the vertices $u',u'',v',v'',w',w''$ of $G'$.}
\label{figure:interval:lower:bound}
\vspace{-0.3cm}
\end{figure}

\begin{lemma}\label{lem:int-lb}
There is an interval graph $G$ such that $\chion{} > 3$ (and thus $\chionp{} \geq 3$).
\end{lemma}

\begin{proof}
We define the graph $G = (V, E)$, an interval representation seen in Figure \ref{figure:interval:lower:bound}, with the help of a preliminary graph $G'=(V',E')$.
$V'$ consists of vertices $u,v,w$ and $u',u'',u^\star,v',v'',w^\star,\allowbreak w',\allowbreak w''$.
Let $E'$ be the edges which form the maximal cliques $\{u',u\}, \{u'',u\},\allowbreak \{u^\star,u,v\},\allowbreak \{v,v'\},\allowbreak \{v,v''\},\allowbreak \{w^\star,v,w\},\allowbreak \{w,w'\}, \{w,w''\}$.
By this ordering of maximal cliques, we observe that $G'$ is an interval graph.

For a vertex $z$, recall that a vertex $z'$ is said to be a {true twin} of $z$ if $z'$ is adjacent to $z$ and all the vertices in $N(z)$.
The graph $G$ is obtained by adding two true twins 
each to the vertices $u,v,w,u^\star,w^\star$
of $G'$ and 
adding three true twins each to the vertices  $u',u'',v',v'',w',w''$ of $G'$.
%Formally, 
In other words, 
%That is, 
$V= \bigcup_{x \in V'} \{x_1,x_2,x_3\} \cup \{u_4',u_4'',v_4',v_4'',w_4',w_4''\}$ and $E = \bigcup_{pq \in E', \; i,j \in [4]} p_i q_j$ (for those where vertices $p_i$ and $q_j$ exist).
Since $G$ is an interval graph, $G'$ is also an interval graph.

Now we show that $G$ cannot be \cfon{} colored with 3 colors.
Suppose there is a \cfon{} coloring $C: V \rightarrow \{1,2,3\}$.
Let $h$ map each vertex $x \in V$ to a uniquely colored neighbor $y \in N(x)$.
%We use the notation $L_{uv} = \bigcup_{1\leq i \leq 3} \{u_i,v_i\}$ and 
% $L_{vw} = \bigcup_{1\leq i \leq 3} \{v_i,w_i\}$.

\begin{claim}
$|C(\{u_1,u_2,u_3\})| = |C(\{v_1,v_2,v_3\})| = |C(\{w_1,w_2,w_3\})| = 2$.
\end{claim}
\begin{proof}%[Claim's Proof]
Suppose for a contradiction that $|C(\{u_1,u_2,u_3\})| = 1$. Without loss of generality, 
we may assume $C(u_1) = C(u_2) = C(u_3) = 1$. 
Note that the neighborhood $N(\{u_1',u_2',u_3', u_4'\}) = \allowbreak \{u_1,u_2,u_3, u_1',u_2',u_3',u_4'\}$. 
%\todo[]{the inclusion is an equality right? If yes, I would replace ⊆
%by =. \bscomment{I think we can use = when we talk about N(.) and not at "h" instance. I have not changed it. You may change if you think if it is necessary. } TH: yes ``=''.}
It follows that the uniquely colored neighbors 
for each vertex in $\{u_1',u_2',u_3', u_4'\}$ belong to the same set $\{u_1',u_2',u_3', u_4'\}$.
%$h(\{u_1',u_2',u_3', u_4'\}) \subseteq \allowbreak \{u_1',u_2',u_3', u_4'\}$, 
%This implies that 
%then must satisfy coloring 
%$ \{2, 3\} \subseteq C(\{u_1',u_2',u_3', u_4'\})$.
	%$\cm(\{u_1',u_2',u_3', u_4'\}) = \ms{2, 2, 3, 3}$.
This implies further that
a vertex amongst $u'_1 , u_2', u_3',  u'_4$ is colored 2, and another one is colored
3.
Analogously it follows that there are a vertex colored 2 and one colored 3 amongst
 $u''_1 , u_2'', u_3'',  u''_4$. 
%\todo[]{We should write the uniquely colored neighbors of each of the vertex is a subset of rather than $h(...)$}
%$\cm(\{u_1'',u_2'',u_3'', u_4''\}) = \ms{2,2,3,3}$.
Then %we have the contradiction that 
each vertex in $\{u_1,u_2,u_3\}$ has two neighbors each of the colors from $\{1, 2, 3\}$, which is a contradiction to the fact that $C$ is a CFON coloring using three colors. 
%$C_{\{\!\!\{ \}\!\!\}}(N(u_1)) \supseteq \{\!\!\{ 1,1,2,2,3,3 \}\!\!\}$.
By symmetry we can also show that 
$|C(\{v_1,v_2,v_3\})| \neq 1$ and  $|C(\{w_1,w_2,w_3\})| \neq 1$.  

%same for the vertices $v$ and $w$. 
%the claim also follows for $x \in \{v,w\}$.

It remains to show that $|C(\{u_1,u_2,u_3\})| \neq 3$, $|C(\{v_1,v_2,v_3\})| \neq 3$ and $|C(\{w_1,w_2,w_3\})| \neq 3$. 
For the sake of contradiction, assume without loss of generality that $C(u_1) = 1, C(u_2) = 2$, and $C(u_3) = 3$.
These vertices are adjacent to $\{v_1,v_2,v_3\}$. As shown in the previous paragraph, we have $|C(\{v_1,v_2,v_3\})| \geq 2$.
%Thus $C(L_{uv})$ has at least two colors that occur at least twice, in contradiction to the previous claim.
If $|C(\{v_1,v_2,v_3\})| = 3$, 
then each of the colors $\{1, 2, 3\}$ appear twice in the neighborhood of the vertex $u_1^{\star}$.
%each vertex in $\{u_1,u_2,u_3\}$ 
%$\cm(N(u^\star_1)) \supset \ms{1,1,2,2,3,3}$
%$\cm(\{u_1,u_2,u_3,v_1,v_2,v_3\}) = \cm(N(u^\star_1) \setminus \{u^\star_2, u^\star_3 \}) = \ms{1,1,2,2,3,3}$, 
Hence, $u^\star_1$ does not have a uniquely colored neighbor. 
If $|C(\{v_1,v_2,v_3\})| = 2$, then without loss of generality, we assume $C(v_1)=1, C(v_2)=1, C(v_3)=2$.
Then each of the colors $\{1, 2\}$ appear twice in the neighborhood of  $u_3$  
%$\cm(N(u_3)) \supset \ms{1,1,2,2}$ 
and thus $C(h(u_3)) = 3$. 
However, since $h(u_3) \in N(u_1)$, the vertex $u_1$  has two neighbors each of the colors from 
%$\cm(N(u_1)) \supset \ms{1,1,2,2,3,3}$ 
$\{1, 2, 3\}$ 
    and $u_1$ cannot have a uniquely colored neighbor. 
By symmetry, we can show that $|C(\{v_1,v_2,v_3\})| \neq 3$ and  $|C(\{w_1,w_2,w_3\})| \neq 3$. 
%Hence the claim. 
%the claim also follows for $x \in \{v,w\}$.
\qed
\end{proof}

Without loss of generality, we may now assume that $C(v_1)=1, C(v_2) = 2, C(v_3) = 2$.
If $3 \notin C(\{u_1, u_2, u_3\})$, then $C(h(u^\star_1))=3$ and 
%must be colored 3 and the corresponding vertex be 
%be 
%either 
$h(u^\star_1)\in \{u^\star_2, u^\star_3\}$.
Without loss of generality, let $h(u^\star_1) = u^\star_2$. This means $C(u^\star_2) = 3$ and $C(u^\star_3) \in \{1, 2\}$.
By a similar reasoning $C(h(u^\star_2)) = 3$. This forces $h(u^\star_2) = u^\star_1$ and $C(u^\star_1) = 3$. However now 
$u^\star_3$ has at least two neighbors from each of the colors in 
$\{1, 2, 3\}$. 
%$\cm(N(u^\star_3)) \supset \ms{1,1,2,2,3,3}$, 
Therefore, $u^\star_3$ does not have a uniquely colored neighbor. Hence, $3 \in C(\{u_1, u_2, u_3\})$, and analogously, $3 \in C(\{w_1, w_2, w_3\})$. 

However, $v_1$ is now adjacent to at least two vertices of color 3 and two of color 2.
Hence, $v_1$ must be adjacent to exactly one vertex with color 1.
This implies either $1 \notin C(\{u_1, u_2, u_3\})$ or $1 \notin C(\{w_1, w_2, w_3\})$.
Without loss of generality, suppose $1 \notin C(\{u_1, u_2, u_3\})$.
By the claim above, $C(\{u_1, u_2, u_3\}) = \{2, 3\}$. 

However, $v_2$ is then adjacent to two vertices with color 1 (i.e., $v_1$ and $h(v_1)$), two vertices of color 2 (i.e., $v_3$ and at least one in $\{u_1, u_2, u_3\}$), two vertices of color 3 (i.e., a vertex in $\{u_1, u_2, u_3\}$ and one in $\{w_1, w_2, w_3\}$).
That means $v_2$ does not have a uniquely colored neighbor, a contradiction. 
Therefore, $G$ cannot be \cfon{} colored with 3 colors.
\qed
\end{proof}

%Due to space constraints, we present the proofs in 
%related to block graphs, interval graphs and split graphs in 
%Appendix \ref{app:inter_app}. 

\begin{lemma}\label{lem:non_contained_interval}
If $G$ is a proper interval graph, then $\chionp{}\leq 2$. 
%$1\leq \chionp{}\leq 2$. 
\end{lemma}

\begin{proof}
Let $G$ be a proper interval graph and $\mathcal I$ be a unit interval representation of $G$. 
We use the function $C:\mathcal I\rightarrow \{1, 2, 0\}$ to assign colors as follows. 
%and $C:\mathcal I \rightarrow \{1, 2, 0\}$ be a CFON$^*$ coloring that is constructed 
%as follows. 
%iteratively like in Lemma \ref{lem:intervalon}. 

%where the intervals are sorted in ascending order according to their left endpoints. 
%We denote the left endpoint 
%and the right endpoint 
%of an 
%interval $I$ by $L(I)$. 
%and $R(I_j)$ respectively. 
%We construct    
%$C:\mathcal I \rightarrow \{1, 2, 0\}$ which will be a \cfonp{} coloring.

At each step $i\geq 1$,  
%\todo[]{We have the word ``iteration'' in the proof}
we pick two intervals $I_1^{i}, I_2^{i} \in \mathcal I$.
%for which $C$ has not been assigned.
The interval $I_1^{i} $ is chosen such that 
%the interval 
%whose 
$L(I_1^{i})$ is the least among intervals for which $C$ has not yet been assigned. 
% The interval  
% $I_2^{i} $ is a neighbor of $I_1^{i}$, 
% whose $L(I_2^{i})$ is the greatest. 
% It might be the case that $C$ has been already assigned for
% all neighbors of $I_1^{i}$. This can happen only in the very last
% iteration of the algorithm.
% Depending on this, we have the following two cases. 
%This is because all the intervals except $I_1^{j}$ are colored and 
%$I_1^{j}$ is the lone interval uncolored. 
The choice of $I_2^{i}$ depends on the  following two cases. 

\begin{itemize}
    \item \textbf{Case 1:} \textit{$I_1^{i}$ has a neighbor for which $C$ is unassigned}. 

    %For each iteration $i$, 
    We choose $I_2^{i}$ such that $R(I_2^i)$ is the largest 
    %that has the greatest right endpoint 
    amongst the intervals in $N(I_1^{i})$ for which $C$ is yet to be assigned.
    %without an assigned color. 
    Notice that $L(I_1^i)<L(I_2^i)$ and hence $R(I_1^i)<R(I_2^i)$. 
    %(as they are unit intervals). 
    We assign $C(I_1^{i})=1$ and $C(I_2^{i})=2$. 
    We assign the color 0 to all the other intervals 
    adjacent to 
    %which are neighbors to 
    either $I_1^{i}$ or $I_2^{i}$.  
    %are assigned the color 0. 

    \item \textbf{Case 2:} \textit{$C$ is already assigned for all the
    neighbors of $I_1^{i}$}. 
    
    % As mentioned before, this can happen only during the last iteration $i=j$.
    % In this case, $I_1^{j}$ is the only interval for which $C$ is yet to be assigned. 
    This cannot happen for $i = 1$, because otherwise the graph has an isolated vertex. 
    Let $\widehat{I}$ be an interval in $N(I_2^{i-1}) \cap N(I_1^i) $. 
    Such an $\widehat{I}$ exists, because otherwise $G$ is disconnected. 
    %\todo[]{$G$ is disconnected, right?}
%    For each iteration 
%    $1\leq i \leq j-2$, 
%    we assign $C(I_1^{i})=1$ and $C(I_2^{i})=2$. 
    We reassign $C(I_1^{i-1})=0$, $C(I_2^{i-1})=1$, 
     $C(\widehat{I})=2$ and assign $C(I_1^i)=0$. 
\end{itemize}

%Notice that whenever Case 1 occurs, the coloring procedure in each step  assigns non-zero colors to two intervals. All the other neighbors of either of these two non-zero colored intervals are assigned the color 0. When Case 2 occurs, 
By the choice of $I^i_1$,
%it is easy to see that 
%for any step $i$, after we assign a color for $I^i_1$, 
all intervals whose left endpoints are smaller than $L(I^i_1)$ have been assigned a color (which may be the color 0).
Therefore, Case 2 can only occur at the last step. Let the last step be the $j$-th step of the coloring process.

We prove by induction on $i$ that $C$ is a \cfonp{} coloring for the induced subgraph containing 
$N[I_1^{i}\cup I_2^{i}]$ 
%their neighbors, 
and all the intervals whose left endpoints are less than $L(I_1^i)$. 
%to the left of any of them. 
For the base case $i = 1$, the subgraph only contains $I_1^{1}$, $I_2^{1}$, and their neighbors.
The claim then holds by construction.
Since Case 1 applies for the base case, the intervals $I_1^1$ and $I_2^1$ see each other as their uniquely colored neighbors, and the  vertices colored 0 see $I_2^1$ as their uniquely
colored neighbor.

For the inductive step for $i > 1$, we first consider the situation when Case 1 applies at step $i$.
Note that the intervals $I_1^{i}$ and $I_1^{i-1}$ have the same color, and so do $I_2^{i}$ and $I_2^{i-1}$. 
However, because of the unit length of the intervals and the choice of the two intervals in each step, it is easy to see that no interval intersects both $I_1^{i}$ and $I_1^{i-1}$.
We have the following cases based on whether $N(I_2^i)\cap N(I_2^{i-1})$ is empty.
%Thus an interval that relies on the intervals a colored in the previous iterations still 
\begin{itemize}
    \item $N(I_2^i)\cap N(I_2^{i-1})=\emptyset$. 

         All intervals colored in the previous steps (till $i-1$) retain their uniquely colored neighbors.
    \item %There is an the base case?interval $I\in N(are for the vertices in I_2^i)\cap N(I_2^{i-1})$. 
There is an interval $I\in N(I_2^i)\cap N(I_2^{i-1})$.

    Notice that by construction, $I_2^{i-1}$ and $I_2^{i}$ are disjoint, and $L(I_2^{i-1})<L(I_2^{i})$.
    Hence, $L(I_2^{i-1})<L(I)$. 
    Further, $I\notin N(I_1^{i-1})$, because otherwise we would have chosen the interval $I$ in place of $I_2^{i-1}$.

    Moreover, $I \in N(I_1^i)$. We have $I_1^i$ as the uniquely colored neighbor for the interval $I$. This argument holds for all the intervals in $N(I_2^i)\cap N(I_2^{i-1})$. For all the other intervals colored in the previous steps (till $i-1$), the uniquely colored neighbors remain the same. 
\end{itemize}

%(resp. $I_2^{i}$ and $I_2^{i-1}$).
%Therefore, all intervals colored in the previous iterations still keep their uniquely colored neighbors.
Further, the intervals $I_1^{i}$ and $I_2^{i}$ act as the uniquely colored neighbors for each other.
Lastly, as every interval has unit length, all neighbors of $I_1^{i}$ that are assigned 0 in step $i$ are also neighbors of $I_2^{i}$.
Therefore, $I_2^{i}$ is the uniquely colored neighbor of all intervals that are assigned 0 in this step.
    
Now suppose that Case 2 applies to step $i$, i.e., we are at the last step $i = j$. That is, there is no $I_2^j$.  
In the $j$-th step, we reassign $C(I_1^{j-1}) = 0$, $C(I_2^{j-1}) = 1$ and $C(\widehat I) = 2$.
% The assignment of colors in iterations $1 \leq i \leq j-2$ are unchanged.
As argued above, before the reassignment in this step, the set of intervals that relied on $I_1^{j-1}$ for their uniquely colored neighbor is  $\{I_2^{j-1}\} \cup \{I\mid I\in N(I_2^{j-1})\cap N(I_2^{j-2})\}$. 
%and 
%the intervals $I\in N(I_2^{i-1})\cap N(I_2^{i-2})$. 
%
%
After the reassignment, the set of intervals $\{I\mid I\in N(I_2^{j-1})\cap N(I_2^{j-2})\}$ depend on $I_2^{j-1}$ (which is reassigned the color 1) for the uniquely colored neighbor. This is fine as $I\notin N(I_1^{j-2})$. The intervals $I_2^{j-1}$ and $I_1^j$ rely on $\widehat{I}$  while all other intervals that relied on $I_2^{j-1}$ previously will continue to rely on $I_2^{j-1}$. 
%depends on 
%depends only on $I_1^{j-1}$ for its uniquely colored neighbor.
%Therefore, the color reassignment of $I_1^{j-1}$ only affects $I_2^{j-1}$.
%    Though $C(I_1^{j-1})$ is changed to 0, this does not affect 
%    any interval colored in the first $j-2$ iterations, since there are no intervals which depend only on $I_1^{j-1}$
%    for their uniquely colored neighbor. 
%    If there was such an interval, this would contradict the choice of $I_1^{j-1}$.
 %   At the iteration $j-1$, we have the following situation. 
 %   There exists no interval contained inside $I_1^{j-1}$. 
 %   So every interval $I_\ell$ except $I_2^{j-1}$ and $I_1^j$, 
%    whose $L(I_\ell)>I_1^{j-1}$ will have 
%    $I_2^{j-1}$ as their uniquely colored neighbor. 
%However, this is not an issue, because the interval $I_m$ is now the new uniquely colored neighbor of $I_2^{j-1}$.
%Further, it is also the uniquely colored neighbor of $I^1_{j}$.
%What remains to be shown is that the color reassignment of $I_m$ does not interfere with the uniquely colored neighbor of any other vertex.
%Because $I^1_j$ is not a neighbor of $I^2_{j-1}$ but $I_m$, and because the intervals have unit length, we must have $L(I_m) > L(I^2_{j-1})$.
%Therefore, no neighbor of $I_m$ is adjacent to another vertex colored 1 in an iteration before $i-1$.
\qed
%We arbitrarily pick a neighbor of $I_1^{j}$, say $I_\ell$ and 
%assign $C(I_1^{j})=1$ and reassign $C(I_\ell)=2$. 
%among the neighbors of $I_1^{i}$.
\end{proof}

\subsection{Algorithmic Status of Conflict-free Coloring on Interval Graphs}
%\todo[]{Added this section. Hung: Since this is quite short, I demoted this to a pragraph heading 
%\bscomment{OK}}
Fekete and Keldenich~\cite{Sandor2017} studied \cfcnp{} coloring on intersection graphs. They showed that for an interval graph $G$, $\chi_{CN}^*(G)\leq 2$. 
%From Lemma \ref{lem:intervalon}, we have that $\chi_{ON}^*(G)\leq 3$. 
%, , and $\chi_{CN}^*(G)\leq 2$  from \cite{Sandor2017}. 
%It was shown in \cite{iwoca-sri} that 
The \cfcnp{} coloring problem was shown to be polynomial time solvable on interval graphs 
in \cite{iwoca-sri}. 

From Lemma \ref{lem:intervalon}, we have that $\chi_{ON}^*(G)\leq 3$. Bhyravarapu, Kalyanasundaram and Mathew in \cite{interval-sri-mfcs} showed that \cfonp{} coloring problem is solvable in time $O(n^{20})$ using the structural properties of interval graphs. 
Independently, Gonzalez and Mann in \cite{DBLP:journals/corr/abs-2203-15724} showed that all the four variants of conflict-free coloring 
%\textcolor{red}{WHICH VARIANTS}
%can be formulated as a local checkable vertex partioning problem and that conflict-free coloring 
can be solved in time $n^{O(w)}$, 
%using some framework, 
where $w$ is the \emph{mim-width} of the graph. 
Interval graphs 
%Coupled with the fact that interval graphs 
have mim-width one.
The algorithms resulting from the formulation in \cite{DBLP:journals/corr/abs-2203-15724} 
result in a running time of 
$O(n^{300})$  on interval graphs. Thus the complexity status of the problem on interval graphs is settled. 
%polynomial time on interval graphs. However the running time of the algorithms is $O(n^{300})$. 

%We show that  the \cfcnp{} problem
%is polynomial time solvable on interval graphs using a characterization.
%given a graph $G$, a \cfcnp{} coloring can be assigned 
%to $G$ in polynomial time. 

\begin{corollary}\label{cor:interval_bound}
\cfcnp{} and \cfonp{} coloring problems are polynomial time solvable on interval graphs. 
\end{corollary}

\section{Unit Square and Unit Disk Intersection Graphs}\label{sec:unitsquare}

Unit square (or unit disk) intersection graphs are intersection graphs of closed unit sized axis-aligned squares (or disks, respectively) in the Euclidean plane. 
Figure \ref{fig:unit_square_example} is a unit square and unit disk graph. 
It is shown in \cite{Sandor2017} that $\chicnp{}\leq 4$ for a unit square intersection graph $G$. They also showed that $\chicnp{}\leq 6$ for a unit disk intersection graph $G$. We study the \cfonp{} coloring problem on these graphs and obtain 
constant upper bounds. 
To the best of our knowledge, no upper bound for \cfonp{} chromatic number was previously known
%To the best of our knowledge, these are the first known upper bounds 
on unit square and unit disk graphs. 
%for \cfonp{} coloring. 

\subsection{Unit Square Intersection Graphs}\label{subsec:unitsquare}
We first discuss the unit square intersection graphs.
Consider a unit square representation of such a graph.
Each square is identified by its center, which is the intersection point of its diagonals. 
By unit square, we mean that the distance between its center and its sides is one, i.e., the length of each side is two.
Sometimes we interchangeably use the term ``vertex'' for unit square.
Throughout, we denote the $X$-coordinate and the $Y$-coordinate of a vertex $v$ with $v_x$ and $v_y$ respectively. 
A \emph{stripe} is the region between two horizontal lines, and the \emph{height} of the stripe is the distance between these two lines. 
We consider a unit square as \emph{belonging} to a stripe 
if its center is contained in the stripe. 
%Each horizontal line is adjacent to two stripes, one above the line and the one below 
%If a unit square has its center on a horizontal line that separates two stripes 
For a unit square whose center lies on a horizontal line, we consider it belonging to the stripe that is immediately below the horizontal line. 
%that is common to two stripes (adjacent stripes)
%then it is considered in the stripe below the line. 
%The bottom most stripe is identified as the region between the horizontal line passing through center of the unit square with the least Y-coordinate and another horizontal line at height 2 from existing horizontal line. 
We say that a unit square intersection graph has \emph{height} $h$, if 
the centers of all the squares lie in a stripe of height $h$. 
%\todo[]{height is given in italics twice}

\begin{lemma}\label{lem:unit_sq_heig2}
If $G$ is a unit square intersection graph of height 2, then  $\chionp{}\leq 2$. 
%Unit square intersection graphs of height 2 are \cfonp{} 2-colorable. 
\end{lemma}
\begin{proof}
Let $G$ be a unit square intersection graph of height 2.
Note that vertices $u$ and $v$ are adjacent if and only if 
$|u_x -v_x|\leq 2$. 
We may represent $G$ as a unit interval graph (with each interval of length 2) by mapping every vertex $v$ to an interval from $v_x-1$ to $v_x+1$. 
It is easy to note that two vertices in $G$ are adjacent if and only if the corresponding vertices 
in the interval representation are adjacent.
%Since $|u_x -v_x|\leq 2$, we have that two vertices $u$ and $v$ of $G$ are adjacent if and only if their corresponding intervals overlap. 
%\todo[]{Change this}
%their $X$-coordinates 
%differ by at most $2$. 
%A graph is an interval graph if and only if there is no induced cycle $C_4$ in it. \todo[]{Incorrect characterization} We will now show that $G$ does not contain an induced $C_4$. Suppose that there is an induced cycle $p-q-r-s-p$ in $G$. WLOG let $p_x\leq q_x \leq r_x \leq s_x$. Since $G$ is a unit square intersection graph of height 2 and 
%the vertices $p$ and $s$ are adjacent we get that $s$ and $q$ are adjacent which is a contradiction. 
%Thus $G$ is an interval graph. 
%Moreover, we may represent $G$ as a unit interval graph (with each interval of length 2) by replacing every vertex $v$ 
%$X$-coordinate $x_v$ 
%by an interval from $v_x-1$ to $v_x+1$. 
By Lemma \ref{lem:non_contained_interval}, we obtain that $\chionp \leq 2$. 
\qed
\end{proof}
\begin{theorem}\label{thm:unit_sq_ub}
If $G$ is a unit square intersection graph, then  $\chionp{}\leq 27$. 
\end{theorem}

\begin{proof}
% We first show that two colors are sufficient to \cfonp{} color unit square intersection graphs of height at most 2.
%The approach %for the general case  
%to arriving at the bound 
%is to divide the graph into horizontal stripes of height 2 and color
%the vertices in two phases.
%We first prove that 
%two colors are sufficient to \cfonp{} color all the unit square intersection graphs of height 2. 
%Throughout, we denote the $X$-coordinate and the $Y$-coordinate of a vertex $v$ with $v_x$ and $v_y$ respectively. 
\begin{figure}[t!]
\centering
\begin{tikzpicture}[scale = 0.6]
\tikzstyle{vertex}=[circle,draw, minimum size=3pt,node distance=1.5cm,scale=0.5]
\tikzstyle{edge} = [draw,thick,-,black]
\node[vertex]  (u1) at (-2,0) {};
\node[vertex]  (u2) at (-1,1.75) {};
\node[vertex]  (u3) at (1,1.75) {};
\node[vertex]  (u4) at (2,0) {};
\node[vertex]  (u5) at (1,-1.75) {};
\node[vertex]  (u6) at (-1,-1.75) {};
\node[vertex]  (v1) at (-1.8,3.15) {};
\node[vertex]  (v2) at (3.6,0) {};
\node[vertex]  (v3) at (-1.8,-3.15) {};
\draw[edge, color=black] (u1) -- (u2) -- (u3) -- (u4) -- (u5) -- (u6) -- (u1) ;
\draw[edge, color=black] (u2) -- (v1) (u4) -- (v2) (u6) -- (v3);
\node [below] at (u2.south east) {$x$};
\node [left] at (u4.west) {$y$};
\node [above] at (u6.north east) {$z$};
\end{tikzpicture}

  \caption{A unit square graph $G$ for which $\chionp{} \geq 3$. The vertices $x, y, z$ 
  have to be assigned distinct non-zero colors. Note that $G$ is also a unit 
  disk graph.}
  \label{fig:unit_square_example}
    \vspace{-15pt}
\end{figure}

%\begin{theorem}\label{thm:unit_sq_ub}
%Let $G$ be a unit square intersection graph. Then  $\chionp{}\leq 27$. 
%\end{theorem}
%\begin{proof}

We assign colors
for  the vertices of $G$, $G = (V,E)$, in two phases --- in phase 1, we color all the
vertices and in phase 2, we modify the coloring to ensure that all the vertices have a uniquely colored neighbor.
In phase 1, we use 6 colors 
$C: V \to \{0\} \cup \{c^i_{0}, c^i_{1} \mid  i \in \{0,1,2\}\}$. 
%$\{ c_{i,0}, c_{i,1} \mid i \in \{0,1,2\} \}$.
Without loss of generality, we assume that the centers of all the squares have positive $Y$-coordinates.
We partition the plane into stripes $S_{\ell}$ for $\ell \in \mathbb{N}$
where each stripe is of height 2. 
%consists of the points whose $Y$-coordinate lies between $2(i-1)$ and $2i$.
We assign vertex $v$ with $Y$-coordinate $v_y$ to $S_{\ell}$ if $2(\ell-1) < v_y \leq 2\ell$.
Let $G[S_{\ell}]$ be the graph induced by the vertices belonging to the stripe $S_\ell$. 
Then $G[S_\ell]$ has height 2. Notice that $G[S_\ell]$ may be disconnected. We apply Lemma \ref{lem:unit_sq_heig2} on each of the connected components, and 
%we can 
color vertices in $S_{\ell}$ using colors $c^i_{0}$ and $c^i_{1}$ where $i = \ell \bmod 3$. 
Then every vertex $u \in S_{\ell}$ that is not isolated in $G[S_{\ell}]$ has a uniquely colored neighbor $v$ in $G[S_{\ell}]$. %\bscomment{What does isolated mean ?}
The isolated vertices in $G[S_{\ell}]$ are assigned the color 0. 
Every $w \notin S_{\ell}$ with color $C(w)=C(v)$ must be in a stripe $S_{\ell^\star}$ with $|\ell-\ell^\star|\geq 3$.
Thus $w \notin N(u)$ and hence $v$ is also a uniquely colored neighbor of $u$ in $G$.
It remains to determine uniquely colored neighbors for the vertices $u \in S_\ell$ 
%for $i \in \mathbb{N}$ 
which are isolated in $G[S_{\ell}]$.
Let $I$ be the set of all such vertices, which belong to all the stripes in the graph.

In phase 2, we reassign colors to some of the vertices of $G$ to ensure a uniquely colored neighbor for each vertex in $I$.
For each vertex $v\in I$, choose an arbitrary \textit{representative vertex} $r(v)\in N(v)$.
Let $R = \{r(v) \mid v \in I\}\subseteq V$ be the set of representative vertices.
We update the coloring $C$ by recoloring the vertices in $R$ using the colors  $\{ c^i_{j} \mid i \in \{0,1,2\}, j \in \{2, 3, \ldots, 8\}\}$.
    %that replaces the color assigned in phase 1.
Consider a stripe $S_{\ell}$ for $\ell \in \mathbb{N}$.
We order the vertices $S_{\ell} \cap R$ non-decreasingly by their $X$-coordinate and
    sequentially color them with $c^i_{2},\dots,c^i_{8}$ where $i = \ell \bmod 3$. 
%In phase 2, we order the vertices of a stripe in non-decreasing way (based on its x-coordinates) and color them sequentially. 
%the vertices of a stripe in sequence based on their coordinates \textcolor{red}{In phase 2, we order the vertices of G' in the stripe in  non-decreasing order}. 

\medskip
\noindent
\textbf{Total number of colors used:} The numbers of colors used in phase 1 and phase 2 are 6 and 21 respectively, giving a total of 27.
%In phase 1, each stripe $S_i$ uses 2 colors. we used 2 colors (plus a dummy color 0) in each stripe $S$. In phase 2, each stripe $S$ uses 7 new colors. 
%The color of a vertex $v\in S$ is the cartesian product $(c,c')$ where $c$ and $c'$ are the colors assigned to $v$ in phase 1 and phase 2 respectively. 
%So, each stripe uses (2+1)*7=21 colors. 
%Each stripe uses 9 colors. Taking into consideration three stripes, the algorithm uses 27 colors. 
%Together with 6 colors used during phase 1, 
%69 colors are sufficient to \cfonp{} a unit square intersection graph. 

\medskip
\noindent
\textbf{Correctness:} 
%Let $G=(V,E)$ be the unit square intersection graph. After the phase 1 coloring, let $I$ denote the set of vertices for which there exists no uniquely colored neighbor. 
%For each vertex $v\in I$, we chose a representative vertex $r_v$. Let the set of representative vertices of $I$ be denoted by the set $R$. 
%In order words, $I$ is the set of vertices participating in phase 2 coloring. 
We now prove that the assigned coloring is a valid \cfonp{} coloring. For this we need to prove the following, 
\begin{itemize}
    \item Each vertex in $I$ has a uniquely colored neighbor. 
    
    \item The coloring in phase 2 does not upset the uniquely colored neighbors (determined in phase 1) of the vertices in $V\setminus I$. 
\end{itemize}

\begin{figure}[t!]
\centering
\begin{tikzpicture}[every node/.style={node distance=1.8cm,scale=0.9}, scale = 0.5]

\tikzstyle{vertex}=[dot, draw, minimum size=8pt]
\draw[gray] (-6,2.4) -- (7,2.4);
\draw[gray] (-6,0.4) -- (7,0.4);

\draw[thick, thin] (4,2) rectangle (6,4); 
\draw[thick, thin] (2.7,-1.8) rectangle (4.7,0.2); 
\draw[thick, thin] (0.5,-1.8) rectangle (2.5,0.2); 
\draw[thick, thin] (-1.7,-1.8) rectangle (0.3,0.2); 
\draw[thick, thin] (-3.9,-1.8) rectangle (-1.9,0.2); 
\draw[thick, thin] (-4,2) rectangle (-2,4); 
\draw[black, double=black] (2,0) rectangle (4,2); 
\draw[black, double=black] (-2,0) rectangle (0,2); 
\draw[black, double=black] (0,2) rectangle (2, 4); 
%\draw[orange, thick, dotted] (-1,-4) -- (-1,6);
%\draw[orange, thick, dotted] (3,-4) -- (3,6);

\node at (1,3) {$v$};
%\node[above, blue]  at (5,4) {$u_2$};
\node at (3,1) {$w$};
\node  at (-1,1) {$u$}; 
\node at (-6,4) {$S_{\ell+1}$};
\node at (-6,-1) {$S_{\ell-1}$}; 
\node at (-6.2, 1.4) {$S_\ell$}; 

%\node[above, blue] at (-3, 4) {$u_1$}; 
%\node[below, blue] at (-3, -2) {$u_3$}; 
%\node[below, blue] at (-0.6, -2) {$u_4$};
%\node[below, blue] at (1.4, -2) {$u_5$};
%\node[below, blue] at (3.6, -2) {$u_6$};

%\node[draw, red]  at (-1,-7.5) {u};

%\node[draw, red]  at (-10,-8) {Colors $\{4,5,6\}$};
%\node[draw, black]  at (-10,-4) {Colors $\{7,8,9\}$};
%\node[draw, blue]  at (-10,0) {Colors $\{1,2,3\}$};
%\node[draw, black]  at (-10,4) {Colors $\{4,5,6\}$};
%\node[draw, black]  at (-10,8) {Colors $\{7,8,9\}$};

\end{tikzpicture}
\caption{The vertex $v\in S_{\ell +1}$ is adjacent to two vertices $u$ and $w$ in $S_\ell$,  which are representative vertices for some isolated vertices. 
In the worst case, $|u_x-w_x|=4$. The picture describes the positions of the isolated vertices whose representative vertex $r$ is such that $u_x\leq r_x\leq w_x$. 
}
\label{fig:unitsquare}
\end{figure}

We first prove the following claim. 

\begin{claim}\label{cla:unit_sq}
%Let $G$ be a unit square intersection graph. 
For each vertex $v\in V$, all vertices in $N(v) \cap R$ are assigned distinct colors in phase 2.
%if there exists a vertex  $u\in N(v)$ such that $|$C(u)$\cap \{c_1, \dots, c_{21}\}|=1$, then $C(u)\neq C(w)$, for each $w\in N(v)\setminus u$. 
\end{claim}
\begin{proof}[Claim's proof]
Let $v \in S_{\ell+1}$ (see Figure \ref{fig:unitsquare}).
Suppose for a contradiction that there are two vertices $u, w\in N(v) \cap R$ such that $C(u)=C(w)$.
Then $u$ and $w$ have to be from the same stripe that neighbors $S_{\ell +1}$.
%\thcomment{a common stripe which neighbors $S_{\ell+1}$} 
%a neighboring stripe 
Without loss of generality, we may assume that $u, w\in S_\ell$, $\ell=0 \bmod 3$ and $u_x \leq w_x$. 
We may further assume that $C(u) = C(w) = c^0_{2}$.
Then there are eight vertices (including $u$ and $w$),  $R' \subseteq R \cap S_\ell$, that are assigned the colors $c^0_{2},c^0_{3},\dots, c^0_{8}, c^0_{2}$ and have their $X$-coordinates between $u_x$ and $w_x$.
Note that $|u_x - v_x| \leq 2$ and $|w_x - v_x| \leq 2$.
Vertices $R'$ are the representative vertices of some eight vertices ${I}' \subseteq I$.
By definition, $I' \subseteq S_{\ell+1} \cup S_{\ell-1}$.

First, let us consider $I' \cap S_{\ell+1}$.
We claim that there is at most one vertex $u' \in {I}' \cap S_{\ell+1}$ 
such that $u'_x < v_x$. %with a lower $X$-coordinate than $v$.
Indeed any such vertex $u' \in I'$ must be adjacent to some representative $r \in R'$ with $|r_x-v_x| \leq 2$.
Thus the distance between $u_x'$ and $v_x$ is at most $4$ and hence there is at most one vertex in ${I}' \cap S_{\ell+1}$ with lower $X$-coordinate than $v$.
Analogously, there is at most one vertex $w' \in {I}' \cap S_{\ell+1}$ such that $w'_x > v_x$.
Considering the possibility that $v \in I'$, we have $|I' \cap S_{\ell+1}|\leq 3$.

Now, consider the vertices in $I' \cap S_{\ell-1}$.
Again any vertex in $I' \cap S_{\ell-1}$ must be adjacent to some representative $r \in R'$ with $|r_x-v_x| \leq 2$.
Thus the $X$-coordinates of the vertices in $I' \cap S_{\ell-1}$ differ by at most 8.
Since the vertices in $I' \cap S_{\ell-1}$ are non-adjacent, we have that $|I' \cap S_{\ell-1}|\leq 4$.
This contradicts the assumption that $|I'| = 8$. 
Thus all vertices $N(v) \cap R$ are assigned distinct colors.
\qed
\end{proof}

We now proceed to the proof of correctness. 
\begin{itemize}
    \item 
\textbf{Every vertex $v\in I$ has a uniquely colored neighbor.} 

Let $v\in S_{\ell+1} \cap I$. 
%be an isolated vertex in stripe $S_{i-1}$ 
%and $r_v$ be its representative vertex in $S_i$. 
%The vertex $r_v$ is the uniquely colored neighbor for $v$. 
By the above claim, no two vertices in $N(v) \cap R$ are assigned
the same color in phase 2.
Since $v$ is not isolated in $G$, we have that $|N(v) \cap R|\geq 1$ and $v$ has a uniquely colored neighbor.
\item \textbf{The coloring in phase 2 does not upset the uniquely colored neighbors of vertices in $V\setminus I$.} 

Let $v\in V\setminus I$ and $u$ be its uniquely colored neighbor after the phase 1 coloring. 
If $u$ is no longer the uniquely colored neighbor of $v$ after phase 2, it has to be the case that $u$ was recolored in phase 2, and $v$ had another vertex $w \in N(v)$ which was assigned the same color as $u$ in phase 2. 
%For $v$ to not have a uniquely colored neighbor after phase 2 coloring, 
%there exists a vertex $w\in N(v)$ such that $C(u)=C(w)$. 
This implies that both $u$ and $w$ are representative vertices for some vertices in 
$I$ and they are recolored in phase 2.
This contradicts the above claim.
%By the above claim, it can not happen that $C(u)=C(w)$ where $C(u)\in \{c_1,\cdots, c_{21}\}$. 
\qed
\end{itemize}
\end{proof}

\subsection{Unit Disk Intersection Graphs}\label{subsec:unitdisk}
In this section, we prove an upper bound for the \cfonp{} chromatic number 
on  
unit disk intersection graphs.
%\subsection{Proof of Theorem \ref{thm:unitdisk}}
Consider a unit disk
representation of such a graph. 
Each disk is identified by its center. By unit disk, we mean that its radius is 1. 
%the distance between its center and its sides is 1 i.e., length of each side is 2. 
Sometimes we interchangeably use the term ``vertex'' for unit disk.
% We first show that two colors are sufficient to \cfonp{} color unit disk
% intersection graphs of height at most $\sqrt{3}$.
%A stripe is the region between two horizontal lines of infinite length. 
We consider a unit disk as belonging to a stripe 
if its center is contained in the stripe. 
If a unit disk has its center on the horizontal line that separates two stripes 
then it is considered in the stripe below the line. 

%The bottom most stripe is identified as the region between the horizontal line passing through center of the unit square with the least Y-coordinate and another horizontal line at height 2 from existing horizontal line. 
We say that a unit disk intersection graph has \emph{height} $h$, if 
the centers of all the disks lie in a horizontal stripe of height $h$. 
% We first prove that 
%two colors are sufficient to \cfonp{} color all the unit disk intersection graphs of height $\sqrt{3}$.
%The approach for the general case  
%to arriving at the bound 
The approach is to divide the graph into horizontal stripes of height $\sqrt{3}$ and color
the vertices in two phases. 
Throughout, we denote the $X$-coordinate and the $Y$-coordinate of a vertex $v$ with $v_x$ and $v_y$ respectively.

\begin{theorem}\label{thm:unitdisk}
If $G$ is a unit disk intersection graph, then $\chionp{}\leq 54$. 
\end{theorem}

The proof of this theorem is similar to the proof of Theorem \ref{thm:unit_sq_ub},
but different in the following three aspects: 

\begin{itemize}
    \item In Theorem \ref{thm:unit_sq_ub}, we used the result that unit square graphs
of height 2 are CFON$^*$ 2-colorable. In this theorem, we will use the result that
unit disk intersection graphs of height $\sqrt{3}$ are \cfcnp{} 2-colorable, and not 
\cfonp{} 2-colorable.

    \item In Theorem \ref{thm:unit_sq_ub}, the set $I$ for which we needed to identify 
    the uniquely colored neighbor was the set of isolated vertices in the respective stripe.
    In this theorem, the set $I$ will be the set of vertices colored in phase 1.
    
    \item In Theorem \ref{thm:unit_sq_ub}, the phase 2 coloring involved considering the representative
    vertices in the order of their $X$-coordinate. For the phase 2 coloring of this theorem,
    we consider the vertices in $I$ in the order of their $X$-coordinate and then color their 
    representative vertices.
\end{itemize}

We will use the following lemma from \cite{Sandor2017}.

\begin{lemma}[Theorem 5 in \cite{Sandor2017}]
\label{lem:disk_sqrt3}
If $G$ is a unit disk intersection graphs of height~$\sqrt{3}$, then 
$\chicnp{}\leq 2$. 
Further, the horizontal distance between two colored vertices is greater than 1.
\end{lemma}

Note that the above lemma pertains to CFCN$^*$ coloring and not CFON$^*$ coloring. The second sentence in the above lemma is not stated in the statement of Theorem 5 in \cite{Sandor2017}, but rather in its proof. We will use the CFCN$^*$ coloring used in the lemma stated above to 
obtain
%in the process of arriving at 
a CFON$^*$ coloring for unit disk intersection graphs. 
Below we reproduce the coloring process
used in the proof of the above lemma in \cite{Sandor2017}.

\medskip
\noindent \textbf{Coloring process used in the proof of Lemma \ref{lem:disk_sqrt3}:}
Let $G=(V,E)$ be a unit disk intersection graph such that
the centers of all the disks in $G$ lie in a stripe of height $\sqrt{3}$. 
The vertices in $V$ are colored in the order of their non-decreasing $X$-coordinates. 
A vertex $v$ is \emph{covered} if and only if it is colored or has a colored neighbor.
%neighbor assigned a color from $\{1, 2\}$.
In each step of the algorithm, we choose a vertex $v$ 
whose $v_x$ is the maximum and 
that covers all uncovered vertices 
to its left. We assign the color 1 (or 2) to $v$ if the previously colored vertex was assigned the color 2 (or 1).  At the end, each uncolored vertex is assigned the color 0. 
It follows from the algorithm that the horizontal distance between any two colored vertices is greater than 1. 
%The reader is directed to \cite{Sandor2017} for the correctness of the algorithm. 

\begin{proof}[Proof of Theorem \ref{thm:unitdisk}]
We assign color  
$C(v)$ to each unit disk $v$ of $G$ in two phases. 
In phase 1, we use 6 non-zero colors to color the vertices of $G$, i.e., 
$C: V \to \{0\} \cup \{c^i_{0}, c^i_{1} \mid  i \in \{0,1,2\}\}$. 
WLOG we assume that the centers of all the disks have positive $Y$-coordinates.
We partition the plane into horizontal stripes $S_{\ell}$ for $\ell \in \mathbb{N}$
where each stripe is of height $\sqrt{3}$. 
We assign vertex $v$ with $Y$-coordinate $v_y$ to $S_{\ell}$ if $\sqrt{3}(\ell-1) < v_y \leq \sqrt{3}\ell$. 
Let $G[S_{\ell}]$ be the graph induced by the vertices belonging to the stripe $S_\ell$. 
Then $G[S_\ell]$ has height $\sqrt{3}$.
We \cfcnp{} color vertices in $S_{\ell}$ accordingly using (nonzero) colors $c^i_{0}, c^i_{1}$ where $i = \ell \bmod 3$, according to Lemma~\ref{lem:disk_sqrt3}.
Let $I$ be the set of all colored vertices after this phase.
Our goal is to \cfonp{} color all the vertices. 
After phase 1, each vertex not in $I$ has a uniquely colored neighbor that is not itself.
Hence we only need
to identify uniquely colored neighbors for vertices in $I$.
%We will show later on that any vertex not in $I$ has a uniquely colored neighbor in $G$.

In phase 2, we modify the colors assigned to some vertices of $G$ to ensure a uniquely colored neighbor for each vertex in $I$.
For each vertex $v\in I$, choose an arbitrary \textit{representative vertex} $r(v)\in N(v)$.
%\skcomment{What if such a representative does not exist? Can we claim that $v$ already has 
%a uniquely colored neighbor? I checked the proof in Fekete Keldenich, and they 
%don't rule out a colored vertex seeing two other vertices of the same color-- THIS ISSUE NEEDS TO BE FIXED}
Note that two vertices in $I$ may share the same representative vertex.
% Let $I_1$ and $I_2$ be the vertices in $I$ that are colored $c^i_{0}$ and $c^i_{1}$, respectively, for $i \in \{0, 1, 2, 3\}$.
Let $R = \{r(v) \mid v \in I\}$ be the set of representative vertices.
We use the coloring function $D: R \to \{d^i_{j} \mid  i \in \{0,1,2,3,4,5\}, j \in \{0,1, \ldots, 7\} \}$ to assign colors to vertices in $R$. 
Consider a stripe $S_{\ell}$ for $\ell \in \mathbb{N}$.
We order the vertices $S_{\ell} \cap I$ non-decreasingly by their $X$-coordinate.
We consider the vertices sequentially in that order. 
If the representative vertex of the current vertex has not yet been colored in phase 2, we color 
it 
%the representative vertex 
with a color in $\{d^i_{0},\dots,d^i_{7} \mid i \equiv \ell \bmod 6 \}$ in a cyclic manner (i.e., the first vertex to be colored will take color $d^i_{0}$, and the next $d^i_{1}$, and so on). %Note that we consider the vertices as per the order in $S_{\ell} \cap I$, and not as per 
%the order in $S_{\ell} \cap R$.

\medskip
\noindent
\textbf{Total number of colors used:} The number of colors used in phase 1 and phase 2 are 6 and 48 respectively, giving a total of 54.

\medskip
\noindent
\textbf{Correctness:} 
We now prove that the assigned coloring is a \cfonp{} coloring, by showing that every vertex has a uniquely colored neighbor.

Firstly, we consider a vertex~$v$ in $G[S_\ell]$, for some $\ell$,  and  not in~$I$.   
By definition of the set~$I$, $v$ is adjacent to its uniquely colored neighbor $u$ colored by $C$, after phase 1. 
Suppose that $u$ is not recolored in phase 2.
Then since~$C$ is a \cfcnp{} coloring of $G[S_{\ell}]$, $v$ is adjacent to a uniquely colored neighbor~$u \neq v$ in $G[S_\ell]$.
By the coloring, the distance between $v$ and other vertices in another stripe with the same color as $u$ is at least $2\sqrt{3} > 2$.
%\skcomment{I suppose the previous sentence follows from 
%the correctness of the CFCN* coloring. We can either mention this, 
%or skip it (as it is evident from the fact that C is a CFCN* coloring.}
Hence, $u$ is the uniquely colored neighbor of~$v$ in~$G$.

Now suppose $u$ is recolored in phase 2 to some color $d^i_{j}$.
For a contradiction, suppose that~$v$ is also adjacent to another vertex~$w$ with the same color~$d^i_{j}$.
Then~$u$ and~$w$ must be the representative vertices of two vertices~$a$ and~$b$ in $I$ that are in stripes $S_{\ell'}$ and $S_{\ell''}$, respectively, such that $\ell' \equiv \ell'' \bmod 6$.
Since $(a,u,v,w,b)$ forms a path in $G$, and since two adjacent vertices in $G$ have Euclidean distance at most 2, we conclude that $a$ and $b$ are at the distance of at most 8.
If $\ell' \neq \ell''$, then $|\ell' - \ell''| \geq 6$.
This implies $|a_y - b_y| \geq 5\sqrt{3} > 8$, a contradiction.
Hence, $a$ and $b$ are in the same stripe.
Because $u$ and $w$ have the same color, there must be 7 other vertices in $I$ between $a$ and $b$ in terms of the $X$-coordinate.
By Lemma~\ref{lem:disk_sqrt3}, this implies $|a_x - b_x| > 8$, another contradiction.
Hence, $v$ cannot be adjacent to two vertices of the same color.

%\skcomment{I think we need to argue that $u$ and $w$ are also in the same stripe. This follows easily %from $u, w$ being common neighbors of $v$ and hence distance between $u$ and $w$ is at most 4. Since 
%we repeat colors in groups of 4, this implies that $u, w$ have to be 
%in the same stripe. Else $|u_y - w_y| > 3\sqrt{3} > 5$.
%Second, in the case when $a$ and $b$ are in the same stripe, we can 
%talk about $u$ and $w$ being colored with the same color. 
%This means that $|u_x - w_x| > 8$. But this contradicts the fact that
%distance between $u$ and $w$ is at most 4. 
%We may be able to bring down the number of distinct stripes to 4, 
%and colors per stripe to 5 (or 4).
%}

Lastly, we consider a vertex~$v$ in $I$.
Then the representative~$u$ of~$v$ is colored by $D$. For $v$ to not have a uniquely colored neighbor, there should exist another representative vertex $w$ such that $D(u)=D(w)$. 
As in the above paragraph, we can use distance arguments to note that two neighbors of $v$ cannot 
be assigned the same color in phase 2. 
Thus
we can conclude that~$v$ is not adjacent to any other vertex with the same color as $D(u)$.
\qed
\end{proof}

%\noindent \textbf{Remark:} 
%We believe that the upper bound 
%for unit disk graphs is loose 
%presented for unit square and unit disk graphs 
%and can be improved. \todo[]{We can move this the R\&D section?}

\section{\NP-completeness on Unit Square and Unit Disk Intersection Graphs %\todo[]{new}
}\label{sec:np-hard}
In this section, we show that the \cfonp{} coloring problem is \NP-hard 
%the algorithmic result 
for unit disk and unit square intersection graphs. 
The idea of the proofs 
is similar to the 
{\sf NP}-completeness proofs in \cite{Sandor2017, Bod2019} for the \cfcnp{} coloring problem. 
%We prove the following for unit disk intersection graphs. A similar result can be obtained for unit square intersection graphs. 

% \hhcomment{I've added the explanation of \textsc{Positive Planar 1-in-3-SAT} and expanded the description of the clause gadgets. By the way, in~\cite{Sandor2017}, they use 4 * 3 * k = 12k vertices per variable gadget, where 4 here possibly refers to 4 sides of the square shape of the gadget. Here we only use 2 * 4 * k = 8k. I'm not sure if this is intentional, but I think it should be enough.}\bscomment{I think that is the idea with which 8k is used. }

\begin{theorem}\label{thm:udisk-hard}
It is \NP-complete to determine if a unit disk intersection graph can be \cfonp{} colored using one color. 
%\todo[]{Comment on this}
\end{theorem}

\begin{proof}
%[Proof of Theorem \ref{thm:udisk-hard}] 
Given a unit disk intersection graph $G=(V,E)$, 
and a partial vertex-coloring using one color, 
we can verify in polynomial time whether the coloring is a \cfonp{} coloring. 
To show that the problem is {\sf NP}-hard,  we give a reduction from \textsc{Positive Planar 1-in-3-SAT}.
%We now prove that the problem is \NP-hard by giving a reduction from \textsc{Positive Planar 1-in-3-SAT}. 
The input to \textsc{Positive Planar 1-in-3-SAT} is a Boolean CNF formula 
%set of clauses and variables 
where each clause has exactly three literals with each literal being positive, 
 and the clause-variable incidence graph is planar.  The objective is to check if there exists an assignment of Boolean values to the variables 
 such that each clause has exactly one true literal. This problem is known to be {\sf NP}-hard, see Mulzer and Rote~\cite{Mulzer2008} for more details. 

 %This version of \textsc{3-SAT} has the additional conditions that all literals are positive, the clause-variable incidence graph is planar, and the satisfiability (called \emph{1-in-3-satisfiability}) requires each clause to have exactly one variable assigned to true; 
 
Given $\phi$, an instance of \textsc{Positive Planar 1-in-3-SAT}, 
we  construct a unit disk intersection graph $G(\phi)$ as follows.

\medskip
\noindent
\textbf{Construction of the graph $G(\phi)$:} 
Let $\{x_1,x_2,\dots, x_n\}$ be the variables and 
$\{t_1,t_2, \dots, t_k\}$ be the clauses of the formula~$\phi$. 
For each variable $x_j$, $1\leq j\leq n$, we introduce a variable gadget which 
%is represented by a  gadget which 
is isomorphic to a   
cycle of length $8k$, where $k$ is the number of clauses. 
We start with an arbitrary vertex and name the vertex as $a_{1,j}$. 
The next three consecutive vertices (in anti-clockwise direction) 
are called $b_{1,j}$, $c_{1,j}$ and $d_{1,j}$. 
The vertices are named in sets of 4. After $a_{1,j}$, $b_{1,j}$, $c_{1,j}$ and $d_{1,j}$, the next four vertices are named $a_{2,j}$, $b_{2,j}$, $c_{2,j}$ and $d_{2,j}$, then $a_{3,j}$, $b_{3,j}$, $c_{3,j}$ and $d_{3,j}$, and so on till $a_{2k,j}$, $b_{2k,j}$, $c_{2k,j}$ and $d_{2k,j}$. See Figure~\ref{fig:unitdisk_gadjets} (right) for an illustration.
%Every fourth vertex from $a_{1j}$, $b_{1j}$, $c_{1j}$ and $d_{1j}$ are denoted by $a_{ij}$ $b_{ij}$, $c_{ij}$ and $d_{ij}$ respectively, where $i\in \{2, \dots, 2k\}$, see Figure~\ref{fig:unitdisk_gadjets} (right). 

The clause gadget is illustrated in Figure~\ref{fig:unitdisk_gadjets} (left). Each clause $t_\ell$, $1\leq \ell \leq k$ is represented by a \emph{clause vertex} $c_{\ell}$. % (illustrated by the thick vertex in the figure).
The  vertex $c_{\ell}$ is adjacent to a tree of five vertices, depicted below $c_{\ell}$ in the figure. %Figure~\ref{fig:unitdisk_gadjets} (left). 
Additionally, there are three paths connecting the clause vertex~$c_{\ell}$ with the three corresponding variable gadgets; for each variable~$x_j$ of the clause~$t_{\ell}$, a path connects~$c_{\ell}$ with the vertex~$a_{y, j}$ of the corresponding gadget, for a suitable $y \in [2k]$. While choosing $y$, we ensure that a vertex $a_{y,j}$ connects to at most one clause gadget.
The length of each path (defined as the number of vertices excluding $c_{\ell}$ and the vertex~$a_{y, j}$ in the variable gadget) is a multiple of 4. 
For illustration, we show such a path  in Figure~\ref{fig:unitdisk_connections}. 
%\todo[]{vriables notation should be changed?}

We now argue that 
the graph $G(\phi)$ is a unit disk intersection graph and can be constructed in polynomial time.  The arguments are similar to those in \cite{Sandor2017}. 
%for the sake of completeness, we explicitly provide them here. 
We start with a planar embedding of the clause-variable incidence graph of $\phi$. We transform all the curved edges in the embedding into straight line segments 
%embedding 
with vertices placed on 
an $O(n+k)\times O(n+k)$ grid. Fraysseix, Pach and Pollack~\cite{fray} showed that such a straight line segment embedding can be obtained in polynomial time. We spread out the vertices in 
this embedding  to ensure that the clause and variable gadgets can be accommodated with adequate distance between them.
%when the vertices of the embedding are replaced by their respective gadgets, the 
%clause gadgets and variable gadgets 
%gadgets are sufficiently far apart.  
The clause vertex in the embedding is replaced by the 
clause  gadget and the variable vertex is replaced with variable gadget.
%the center of the variable gadget marked by a cross ($\times$) in Figure \ref{fig:unitdisk_gadjets}. 

The edges between variables and clauses are replaced by paths whose lengths are multiples of 4. We  perform some local shifting (we move the vertices of the path by a small distance, retaining the adjacencies)  to ensure that the path lengths are multiples of 4.
%In case of multiple options for $a_{y,j}$ 
When connecting a clause gadget to a variable gadget $x_j$, we choose an $a_{y,j}$   that is not already connected to any clause gadget. 
Note that we may have to bend some paths while trying to make the connections, and  ensure that the connecting paths between clause gadgets and variable gadgets do not intersect. 

Below, we show that $G(\phi)$ is \cfonp{} colorable using one color if and only if 
$\phi$ is 1-in-3-satisfiable. 

\begin{Observation}\label{obs:unitdisk}
In any CFON$^*$ coloring of $G(\phi)$ using one color, each  clause vertex $c_\ell$, where $1\leq \ell \leq k$, remains uncolored.  Further, the uniquely colored neighbor of $c_\ell$ is not from the tree adjacent to it. 
\end{Observation}
\begin{proof}
%The purpose of this tree is to ensure that $c_{\ell}$ is not colored, and at the same time, the uniquely colored neighbor of $c_{\ell}$ is not from the tree adjacent to it. 
Consider the tree of 5 vertices adjacent to $c_{\ell}$. In any CFON$^*$ coloring, the shaded vertices of the tree (See Figure \ref{fig:unitdisk_gadjets}) are forced to be assigned the non-zero color because of its pendant neighbors. This forces the remaining three vertices of the tree and $c_{\ell}$ to remain uncolored. Since the neighbor of $c_{\ell}$ in the tree is uncolored, its uniquely colored neighbor does not belong to the tree. 
\qed
\end{proof}

%We now show that $G[\phi]$ is \cfonp{} colorable using one color  if and only if $\phi$ is 1-in-3-satisfiable. 

%If the variable $x_j$ is set to true in $\phi$, then all the $a_{ij}$ vertices are colored and exactly one of the following holds: either all the $d_{ij}$ are colored or all the $b_{ij}$ are colored, where $i\in [2k]$. 
%Else if $x_j$ is set to false in $\phi$, then all the $c_{ij}$ are colored and exactly one of the following hold: either all $d_{ij}$ are colored or all $b_{ij}$ are colored. The remaining vertices are left uncolored in both the cases. 

\begin{figure}[t!]
\begin{center}
    
\begin{minipage}{0.4\textwidth}
\centering

\begin{tikzpicture}[every node/.style={node distance=1.8cm,scale=0.9}, scale = 0.5]

\tikzstyle{vertex}=[dot, draw, minimum size=8pt]
%\draw[gray] (-6,2.4) -- (7,2.4);
%\draw[gray] (-6,0.4) -- (7,0.4);

\draw[dotted, thick] (-4,5) circle (1cm);
\draw[dotted, thick] (-2,3) circle (1cm);
\draw[dotted, thick] (-6,3) circle (1cm);
\draw[thick] (-4,3) circle (1cm);
\draw[thin] (-4,1) circle (1cm); 
\draw[thin, shade] (-4,-1)  circle (1cm);
\draw[thin] (-6,-1)  circle (1cm);
\draw[thin, shade] (-2,-1)  circle (1cm);
\draw[thin] (0,-1)  circle (1cm);

\node at (-4,3) {$c_{\ell}$};
%\node at (1.2,2.6) {$v$};

\end{tikzpicture}

\end{minipage}
\begin{minipage}{0.4\textwidth}
\centering

\begin{tikzpicture}[every node/.style={node distance=1.8cm,scale=0.9}, scale = 0.5]

\tikzstyle{vertex}=[dot, draw, minimum size=8pt]
%\draw[gray] (-6,2.4) -- (7,2.4);
%\draw[gray] (-6,0.4) -- (7,0.4);

\draw[thin] (-4,5) circle (1cm);
\draw[thin] (-3,6) circle (1cm);
\draw[thin] (-4,3) circle (1cm);
\draw[thin] (-1,6) circle (1cm);

\node at (-4,5) {$c_{1,j}$};
\node at (-3,6) {$b_{1,j}$};

\node at (-4,3) {$d_{1,j}$};
\node at (-1,6) {$a_{1,j}$};

\draw[thin] (-3,-3) circle (1cm);
\draw[thin] (-1,-3) circle (1cm);
\draw[thin] (-4,-2) circle (1cm);
\draw[thin] (-4,0) circle (1cm);
%\node at (-3,-3) {$b_{21}$};
%\node at (-4,-2) {$a_{21}$};

\draw[thin] (5,-2) circle (1cm);
\draw[thin] (2,-3) circle (1cm);
\draw[thin] (4,-3) circle (1cm);
\draw[thin] (5,0) circle (1cm);

\node at (5,-2) {$c_{i,j}$};
\node at (2,-3)  {$a_{i,j}$};

\node at (4,-3) {$b_{i,j}$};
\node at (5,0) {$d_{i,j}$};

\draw[thin] (5,5) circle (1cm);
\draw[thin] (2,6) circle (1cm);
\draw[thin] (4,6) circle (1cm);
\draw[thin] (5,3) circle (1cm);

\node at (-4,1.5) {$\dots$};
\node at (0.5,-3) {$\dots$};
\node at (5,1.5) {$\dots$};
\node at (0.5,6) {$\dots$};
%\node at (0.5,1.5) {$\times$};

%\node at (1.2,2.6) {$v$};

\end{tikzpicture}
\end{minipage}
\end{center}
\caption{The clause gadget is on the left.  
The dotted disks around the clause vertex $c_{\ell}$ indicate the connection with the variable gadgets. 
The shaded vertices force the clause vertex to not draw its uniquely colored neighbor from within the clause gadget. \\
On the right side, we have the variable gadget for $x_j$. 
%whose center is marked by a cross ($\times$). \textcolor{red}{SUBRUK: Do we really need the $\times$?}
}
\label{fig:unitdisk_gadjets}
\end{figure}

\begin{figure}[t!]
\begin{center}
    
%\begin{minipage}{0.4\textwidth}
\centering{

\begin{tikzpicture}[every node/.style={node distance=1.8cm,scale=0.9}, scale = 0.5]

\tikzstyle{vertex}=[dot, draw, minimum size=8pt]
%\draw[gray] (-6,2.4) -- (7,2.4);
%\draw[gray] (-6,0.4) -- (7,0.4);

\draw[thin, shade] (8,1) circle (1cm);
\draw[thin, shade] (8,3) circle (1cm);
\draw[thin] (6,3) circle (1cm);
\draw[thin] (4,3) circle (1cm);
\draw[thin, shade] (0,3) circle (1cm);
\draw[thin, shade] (-2,3) circle (1cm);
\draw[thick] (-4,3) circle (1cm);
\draw[thin] (-6,3) circle (1cm);
\draw[thin] (-4,5) circle (1cm);
\draw[dotted, thick] (-4,1) circle (1cm); 
\draw[dotted, thick] (-4,-1)  circle (1cm);
\draw[dotted, thick] (-6,-1)  circle (1cm);
\draw[dotted, thick] (-2,-1)  circle (1cm);
\draw[dotted, thick] (0,-1)  circle (1cm);
\draw[thin] (8,5) circle (1cm);
\draw[thin] (8,-1) circle (1cm);

\node at  (2,3) {$\ldots$};
\node at  (8,3) {$a_{y_1, j_1}$};
\node at  (8,1) {$b_{y_1, j_1}$};
\node at  (8,-1) {$c_{y_1, j_1}$};
\node at  (8,5) {$d_{{y_1-1}, j_1}$};
\node at (6,3) {$g_\ell^1$};
%\node at (4,3) {$f_\ell^i$};
%\node at (0,3) {$e_\ell^i$};
\node at (-2,3) {$m_\ell^1$};
\node at (-6,3) {$m_\ell^3$};
\node at (-4,5) {$m_\ell^2$};
\node at (-4,3) {$c_{\ell}$};
\node at (8,7) {$\vdots$};
\node at (8,-3) {$\vdots$};
\node at (-4,7) {$\vdots$};
\node at (-8,3) {$\dots$};
%\node at (1.2,2.6) {$v$};

\end{tikzpicture}
}
%\end{minipage}
\end{center}
\caption{
Illustration of a \cfonp{} coloring of the path connecting the clause vertex $c_{\ell}$ to the vertex $a_{y_1,j_1}$ in the variable gadget of $x_{j_1}$. 
The dotted vertices are a part of the clause gadget. }
%The path connecting the clause vertex $c_{\ell}$ to the vertex $a_{y,j}$ in the variable gadget of $x_{j}$. The dotted vertices are a part of the clause gadget. 
%The clause vertex $c_\ell$ is uncolored for any \cfonp{} coloring of $G(\phi)$, and exactly one of 
%$m_\ell^1, m_\ell^2, m_\ell^3$ is colored. 
%The vertices $m_\ell^1$ and $e_\ell^1$ are colored if the variable $x_{j}$ is true.} %Else, the vertices $e_\ell^i$ and  $f_\ell^i$ are colored. Every fourth vertex from the initial colored vertices are colored along the path, starting from either $\{m_\ell^i,e_\ell^i\}$ or $\{e_\ell^i,f_\ell^i\}$. }
\label{fig:unitdisk_connections}

\end{figure}

\begin{lemma}\label{lem:col-sat}
   If  $G(\phi)$ is \cfonp{} colorable using one color,  then $\phi$ is 1-in-3-satisfiable. 
\end{lemma}
\begin{proof}
   Let $G(\phi)$ have a \cfonp{} coloring  using one color. 
   We first consider the clause gadgets. %Consider a clause vertex $c$. 
%We first note that in any \cfonp{} coloring using one color, exactly one of the vertices $m_\ell^1, m_\ell^2, m_\ell^3\in N(c_\ell)$ is colored. 
Let $m_\ell^1, m_\ell^2, m_\ell^3 \in N(c_{\ell})$ be the first vertices of the paths that connect the clause vertex $c_\ell$ to each of the variable gadgets. 
Let the paths from $c_\ell$ to the variable gadgets terminate at the vertices 
$a_{y_1, j_1}$, $a_{y_2, j_2}$ and $a_{y_3, j_3}$ respectively. 
%This is because the shaded vertices in the clause gadget are forced to be colored which thereby forces $c_\ell$ to not have a uniquely colored neighbor from within the gadget. 
As noted in Observation \ref{obs:unitdisk}, one of these vertices has to be the uniquely colored neighbor of $c_{\ell}$.
Without loss of generality, let $m_\ell^1$ be the uniquely colored neighbor of   $c_{\ell}$. This  implies that  $m_\ell^2$ and $m_\ell^3$ are not colored.
%and $d_2, d_3$ be uncolored. 
%In any \cfonp{} coloring using one color,  
Let $g_\ell^1$ be the last vertex in the path that connects $c_{\ell}$ to the variable gadget.
Recall that the length of the path from $m_\ell^1$ to $g_\ell^1$ 
%(the last vertex on the path) 
connecting $c_{\ell}$ to 
 $a_{y_1, j_1}$, %for some $\ell \in [2k]$, 
 in the corresponding variable gadget, is a multiple of four. 
%In any CFON$^*$ coloring using one color of this path, starting from a colored vertex  $m_\ell^1$, there will be two colored vertices, followed by two uncolored vertices, again followed by two colored vertices and so on.  
Consider the path starting from  the colored vertex $m_\ell^1$ to  $g_\ell^1$. 
Since  $m_\ell^1$  is colored, it follows that in any  CFON$^*$ coloring using one color, the first two vertices have to be colored, the next two vertices have to be  uncolored, the next two have to be colored and so on.
%there will be two  
%the colored and uncolored vertices come in pairs. 
%That is, 
%there are two colored vertices followed by two uncolored vertices followed by two colored vertices and so on starting from $m_\ell^1$. 
%$a_{ij}$ be the vertex in the variable gadget connecting the path from $d_1$. 
Consequently, the vertex~$g_\ell^1$ is uncolored and~$a_{y_1, j_1}$ is colored (to be the uniquely colored neighbor of~$g_{\ell}^1$). An illustration is given in Figure \ref{fig:unitdisk_connections}. 

Now consider the path from $m_\ell^2$ (resp. $m_\ell^3$) to its corresponding variable gadget. 
The starting vertex $m_\ell^2$ (resp. $m_\ell^3$) is not colored. The next two vertices have to be colored, followed by two uncolored vertices, then followed by two colored vertices and so on. 
The last vertex $g_\ell^2$ (resp. $g_\ell^3$) in the path  will be uncolored. 
%starting from $N(m_\ell^2)\setminus \{c_\ell\}$ (resp. $N(m_\ell^3)\setminus \{c_\ell\}$) in the path to the corresponding variable gadget. 
%In the path from $m_\ell^1$ to $g_\ell^1$, the gadget $a_{yj}$ is colored. 
This implies that the vertices  $a_{y_2,j_2}$ and  $a_{y_3,j_3}$ in the 
corresponding variable gadgets are uncolored.
%the last vertex~$g_\ell^i$ in the path where $i\in \{2, 3\}$ is uncolored and its neighbor,  $a_{y_i,j_i}$, in the corresponding variable gadget is also uncolored. 
%path $f_i$ as its uniquely colored neighbor. 
%Hence its neighbor $a_{ij}$ remains uncolored and should have its uniquely colored neighbor from within the variable gadget. 

Notice that for $i \in \{1, 2, 3\}$ the vertex $g_{\ell}^i$ is uncolored. 
Hence the vertex $a_{y_i j_i}$, for $i \in \{1, 2, 3\}$, has its uniquely colored neighbor within the variable gadget of~$x_{j_i}$.
Because of this, we also have the same coloring pattern along the cycle of any variable gadget, i.e. a pair of colored vertices followed by a pair of uncolored vertices repeating.
This implies that in any \cfonp{} coloring using one color, for each variable gadget of a variable~$x_j$, 
the vertices $a_{y,j}$, for $1 \leq y \leq 2k$, 
are either all colored or all uncolored. 
In addition, as argued above, if $a_{y,j}$ is connected to a clause vertex~$c_{\ell}$, then~$a_{y,j}$ is colored if and only if the corresponding adjacent vertex~$m_{\ell}^i$ of~$c_{\ell}$ is colored.

%This ensures that  for each $1\leq y\leq 2k$, the vertex $a_{yj}$ in the gadget $x_j$ is colored.  Hence in $x_j$, for each $1\leq y\leq 2k$ either all of $d_{yj}$ are colored or all of $b_{yj}$'s are colored but not both. 

%In any \cfonp{} coloring, for each $1\leq y\leq 2k$,  the vertices $c_{yj'}$ is colored and 
%either all the vertices $b_{yj'}$ are colored or all $d_{yj'}$ are colored but not both. 
%Note that all the vertices in $a_{yj'}$ are uncolored. 

To obtain the desired 1-in-3-satisfying assignment of $\phi$, we consider the vertices
$a_{y,j}$ in the clause gadget corresponding to $x_j$. 
If the vertices $a_{y,j}$  for $1\leq y\leq 2k$ are colored in the  \cfonp{} coloring,
we set $x_j$ to true. Else we set $x_j$ to false. 
The arguments above imply that for each clause $t_{\ell}$, exactly one of the variables $x_j$ in the clause will be set to true. This implies that 
$\phi$ is 1-in-3-satisfiable. 
%
%Therefore, in the \cfonp{} coloring, if all the vertices $a_{y,j}$ are colored in the gadget $x_j$ for each $1\leq y\leq 2k$,  
%in a variable gadget $x_j $ if all $a_{ij}$'s are colored, 
%we set $x_j$ to be true in $\phi$. Else, we set $x_j$  to be false. 
%The arguments above imply that for each clause gadget $c_{\ell}$, exactly one of the variable gadgets $x_j$ connecting to~$c_{\ell}$ has all vertices $a_{yj}$ colored. 
%This implies that exactly one variable is set to true in each clause, making $\phi$ 1-in-3-satisfiable. 
\qed
\end{proof} 
%Thus it forces the When $d_1$ is uncolored, it forces the vertices $g_i$ and $a_ij$ to be uncolored. Thus in any \cfonp{} coloring either $b_{ij}$ is colored or $d_{i-1}j$ is colored thereby forcing the vertices $c_{ij}$ or $c_{i-1}j$ to be colored. 
%As argued in the earlier case, because of the type of vertices colored in this coloring, this gadget implies $x_j$ should be set to false in $\phi$. 

%One of its neighbors is colored The coloring ensures that $a_{ij}$ is colored and one of $b_{ij}$ and $d_{i-1j}$ is colored. 
%If $b_{ij}$ is colored, it forces $c_{ij}$, $d_{ij}$ to be uncolored and $a_{i+1 j}$, $b_{i+1 j}$ to be colored. 
%This ensures that all $a_{ij}$ and $b_{ij}$   where $\forall 1\leq i\leq 8k$ are colored in the gadget $x_j$. 
%Else if $d_{i-1j}$ is colored, then it forces $b_{ij}$, $c_{ij}$, $b_{i-1j}$, $c_{i-1j}$ to be uncolored and $a_{ij}$, $d_{i-1 j}$ to be colored for all $i\leq 8k$. 
%Since all $a_{ij}$'s are colored and either all $b_{ij}$'s or all $d_{ij}$'s are colored, we have that $x_j$ is true in $\phi$. 

%$d_{ij}$ are uncolored and 
%viceversa. Since it is a \cfonp{} coloring, we have all the $a_{ij}$ to be colored and either all the $b_{ij}$ or all the $d_{ij}$ are colored. 

\begin{lemma}\label{lem:sat-col}
  If  $\phi$ is 1-in-3-satisfiable, then $G(\phi)$ is \cfonp{} colorable using one color. 
\end{lemma}
\begin{proof}
Consider a 1-in-3-satisfying assignment of $\phi$.     
For each variable $x_j$ in $\phi$ that is set to true, 
color all vertices $a_{y,j}$ and $b_{y,j}$ for each $1\leq y\leq 2k$. 
%to be colored. 
Else, 
%if $x_j$ is false in $\phi$, we set all 
color all vertices $b_{y,j}$ and $c_{y,j}$ for each $1\leq y\leq 2k$. %to be colored. 
In either case, the remaining vertices in the gadget are left uncolored. 
Such a coloring ensures that every vertex in the variable gadget associated with $x_j$ has a uniquely colored neighbor from the gadget.
%for every vertex in $x_j$. 
Figure \ref{fig:unitdisk_connections} illustrates the case when all vertices $a_{y,j}$ and $b_{y,j}$ are colored.

Consider the case when all the vertices $a_{y,j}$ are colored for $1\leq y\leq 2k$.   
Suppose a vertex $a_{y',j}$ is adjacent to a vertex $g_\ell^i$ along the path to a clause gadget representing the clause $t_\ell$. 
%Since $a_{y',j}$ is colored and already has a uniquely colored neighbor, 
We leave $g_\ell^i$  uncolored along with its other neighbor in $N(g_\ell^i)\setminus \{a_{y',j}\}$ on the path. 
We now color the next two vertices on the path, leave the next two vertices uncolored and so on till we reach $m_\ell^i$ (which is the vertex adjacent to $c_\ell$). 
%Starting from $g_i$ to the vertex $d_i$, two vertices are uncolored and next two vertices are colored and so on. 
Since the length of the path is a multiple of four, the vertices $m_\ell^i$ and $N(m_\ell^i)\setminus \{c_\ell\}$ will be colored.  %Else if all $a_{ij}$'s are uncolored, then its neighbor $g_i$ (if any) is uncolored and starting from $f_i$ to $d_i$, we have two vertices colored and next two vertices uncolored and so on. 
%Hence the other vertices adjacent to $c_\ell$ are left uncolored. 

The other case is that all $a_{y,j}$ are left uncolored in the variable gadget of $x_j$. Consider the connecting paths to the clause gadgets starting from a vertex say $g_\ell^i$ (adjacent to $a_{y',j}$) and ending at $m_\ell^i$, where $m_\ell^i$ is adjacent to the 
vertex $c_\ell$ from the clause gadget.
In this case we leave $g_\ell^i$ uncolored, coloring the next two vertices in the 
path, leaving the next two vertices uncolored and so on. The vertex  $m_\ell^i$
is hence uncolored. 
%will be uncolored and have a uniquely colored neighbor in $N(m_\ell^i)\setminus \{c_\ell\}$. %Hence $c_\ell$ is forced to be left uncolored. 

Since $\phi$ is positive planar 1-in-3-satisfiable, the  assignment 
assigns true to exactly one variable of
each clause $t_{\ell}$. 
%has exactly one variable set to true. 
This ensures that there is exactly one colored neighbor of $c_\ell$.
The rest of the vertices in the clause gadget can easily be \cfonp{} colored. 
%and that neighbor 
%leads the path to the variable gadget which is set to true. 
We have a \cfonp{} coloring using one color according to the above rules. 
\qed
\end{proof}

%\textcolor{red}{Now, all that remains to show is that $G(\phi)$ is a planar unit disk intersection graph.} 
%\bscomment{But this is exactly the same arguments from \cite{Sandor_2017_interval} where they proved that the \cfcnp{} problem is \NP-hard on unit disk and unit square intersection graphs. Do we have to write them explicitly or point to that paper ?}

%Given an instance $\phi$ of \textsc{Positive Planar 1-in-3-SAT}, 
%we construct an instance $G(\phi)$ in polynomial time as discussed above. 
Lemmas \ref{lem:col-sat} and \ref{lem:sat-col} imply that $G(\phi)$ is \cfonp{} colorable using one color  if and only if $\phi$ is 1-in-3-satisfiable.
\qed
\end{proof}

\begin{theorem}\label{thm:usquare-hard}
It is \NP-complete to determine if a unit square intersection graph can be  \cfonp{} colored using one color. 
\end{theorem}
\begin{proof}
The reduction is from \textsc{Positive Planar 1-in-3 Sat}, and similar to the reduction in the proof of Theorem \ref{thm:udisk-hard}.
%where we consider unit squares instead of unit disks. 
The graphs corresponding to the clause and the variable gadgets are  the same as the ones used in the proof of Theorem \ref{thm:udisk-hard}.
%in construction as shown in Figure \ref{fig:unitdisk_gadjets}. 
The clause  and  variable gadgets can be realised as unit square graphs. For instance, see  Figure \ref{fig:unitsquare_gadjets} for an illustration of the clause gadget.
%Similar variable gadget as in Figure \ref{fig:unitdisk_gadjets} can be constructed. 
%We give a reduction from \textsc{positive planar 1-in-3 sat}. 
%Given a formula $\phi$, $G(\phi)$ is constructed using gadgets and the paths similar to the above when each unit disk is replaced with unit square. 
\qed
\end{proof}

\begin{figure}[t!]
\begin{center}
    
\begin{minipage}{0.4\textwidth}
\centering

\begin{tikzpicture}[every node/.style={node distance=1.8cm,scale=0.9}, scale = 0.4]

\tikzstyle{vertex}=[dot, draw, minimum size=8pt]
%\draw[gray] (-6,2.4) -- (7,2.4);
%\draw[gray] (-6,0.4) -- (7,0.4);

\draw[thick] (0,0) rectangle (2,2);
\draw[dotted, thick] (1.25,1.25) rectangle (3.25,3.25);
\draw[dotted, thick] (-1.25,1.25) rectangle (0.75,3.25);
\draw[thin] (1.25,-1.25) rectangle (3.25,0.75);
\draw[dotted, thick] (-1.25,-1.25) rectangle (0.75,0.75);
\draw[thin, shade] (2,-2.5) rectangle (4,-0.5);
\draw[thin] (0.75,-4) rectangle (2.75,-2);
\draw[thin, shade] (3,-4) rectangle (5,-2);
\draw[thin] (4,-5) rectangle (6,-3);

\node at (1,1) {$c_{\ell}$};
%\node at (1.2,2.6) {$v$};

\end{tikzpicture}

\end{minipage}

\end{center}
\caption{The clause gadget corresponding to the clause $t_{\ell}$. 
%with 
%is on the left.  
The dotted squares around the clause vertex $c_{\ell}$ indicate the connection with the variable gadgets. 
The shaded vertices force $c_{\ell}$ to not draw its uniquely colored neighbor from within the clause gadget. 
%\\
%On the right side, we have the variable gadget for $x_j$. 
%whose center is marked by a cross ($\times$). \textcolor{red}{SUBRUK: Do we really need the $\times$?}
}
\label{fig:unitsquare_gadjets}
\end{figure}

\section{Kneser Graphs}\label{sec:kne_split}
%\subsection{Kneser Graphs}\label{sec:kneser}

In this section, we study 
%complete %conflict-free 
%both the problems on 
\cfonp{} and \cfcnp{} colorings
of Kneser graphs. 
We shall use the words \emph{$\kappa$-set} or \emph{$\kappa$-subset} to refer 
to a set of size $\kappa$. We shall sometimes refer to the $\kappa$-subsets of $[n]$ and the vertices
of $K(n,\kappa)$ in an interchangeable manner. We also use the symbol $\binom{S}{\kappa}$
to denote the set of all $\kappa$-subsets of a set $S$.

\begin{definition}[Kneser graph]
The Kneser graph $K(n,\kappa)$ is the graph whose vertices are $\binom{[n]}{\kappa}$, the $\kappa$-sized subsets of $[n]$, and the vertices $x$ and $y$ are adjacent if and only if $x\cap y = \emptyset$ (when $x$ and $y$ are viewed as sets).
\end{definition}

Observe that for $n < 2\kappa$, $K(n,\kappa)$ has no edges, and for $n = 2\kappa$, $K(n,\kappa)$ is a perfect matching.
Since we are only interested in connected graphs, we assume $n \geq 2\kappa + 1$.
For this range of values of $n$, we show that $\chi^*_{ON}(K(n,\kappa)) \leq \kappa+1$.
Further, we prove that this bound is tight for $n \geq 2\kappa^2 + \kappa$.
We conjecture that this bound is tight for all $n \geq 2\kappa + 1$.
% For Kneser graphs $K(n,\kappa)$, we show that $\chi^*_{ON}(K(n,\kappa)) = \kappa+1$ 
% when $n\geq \kappa(\kappa+1)^2 + 1$ and show bounds for $\chi^*_{CN}(K(n,\kappa))$. 
In addition, we also show an upper bound for $\chi^*_{CN}(K(n,\kappa))$.

%We prove the following results on Kneser graphs $K(n,\kappa)$ in Appendix \ref{app:kneser}. 
%Due to space constraints, the proofs are presented in Appendix \ref{app:kneser}. 

\begin{theorem}\label{thm:knesercfon}
%Suppose $K(n,\kappa)$ is a Kneser graph. Then 
%$\chi^*_{ON}(K(n,\kappa))  \leq \kappa+1$, for $n\geq 3\kappa -1$. Further 
For 
% $n\geq \kappa(\kappa+1)^2 + 1$, $\chi^*_{ON}(K(n,\kappa)) = \kappa+1$.
$n\geq 2\kappa^2 + \kappa$, $\chi^*_{ON}(K(n,\kappa)) = \kappa+1$.
\end{theorem}

The above theorem is an immediate corollary of the two lemmas below.
%During the discussion, we shall use the words \emph{$\kappa$-set} or \emph{$\kappa$-subset} to refer 
%to a set of size $\kappa$. We shall sometimes refer to the $\kappa$-subsets of $[n]$ and the vertices
%of $K(n,\kappa)$ in an interchangeable manner. We also use the symbol $\binom{S}{\kappa}$
%to denote the set of all $\kappa$-subsets of a set $S$.

\begin{lemma}\label{lem:cfonk+1}
%$\chions(K(n,\kappa)) \leq \kappa+1$ when $n\geq 3\kappa-1$.
 $\kappa+1$ colors are sufficient to 
\cfonp{} color %a 
$K(n,\kappa)$ for $n \geq 2\kappa + 1$.
%when $n\geq 3\kappa-1$. 
%$\forall n \geq 3\kappa-1$, $\chi(K(n,\kappa))\leq \kappa+2$.
\end{lemma}

\begin{proof}
Consider the following assignment of colors to the vertices of $K(n,\kappa)$:
%\footnote{In this coloring, 
%the uniquely colored neighbor is not colored $\kappa+2$ for any 
%of the vertices. Thus, by recoloring the color class $\kappa+2$
%with the color 0, we get a partial coloring that uses $\kappa+1$ colors.}
\begin{itemize}
    \item For any vertex ($\kappa$-set) $v$ that is a subset of $\{1, 2, \dots, 2\kappa\}$, we assign $C(v) = \max_{\ell \in v} \ell - (\kappa-1)$. 
    %and we color the vertex $[2\kappa] \setminus v$ with color $\kappa+1$.
    % \item The set $\{2\kappa, 2\kappa+1, \dots, 3\kappa-1\}$ is assigned the color $\kappa+1$.
    \item All the remaining vertices are assigned the color 0.
\end{itemize}
For example, for the Kneser graph $K(n,3)$, 
we assign the color 1 to the vertex $\{1, 2, 3\}$, 
color 2 to the vertices $\{1, 2, 4\}, \{1, 3, 4\}, \{2, 3, 4\}$, 
color 3 to the vertices $\{1, 2, 5\}, \{1, 3, 5\}, \{1, 4, 5\}, \{2, 3, 5\}, \{2, 4, 5\}, \{3, 4, 5\}$,
color 4 to the vertices 
$\{1, 2, 6\}, \{1, 3, 6\}, \{1, 4, 6\}, \{1, 5, 6\},
\{2, 3, 6\}, \{2, 4, 6\},$ $\{2, 5, 6\}, \{3, 4, 6\},$
$\{3, 5, 6\},$ $\{4, 5, 6\}$, and
%that are subsets of $[6]$ that contain the entry 6, and 
color 0 to all the remaining vertices. 

Now, we prove that the above coloring is a \cfonp{} coloring. 
Let $C_i$ be the set of all vertices assigned the color $i$.
%,i.e., the color class of the color $i$.
Notice that 
% $C_1 \cup C_2 \cup \dots \cup C_k = \binom{[2\kappa-1]}{\kappa}$.
$C_1 \cup C_2 \cup \dots \cup C_{\kappa+1} = \binom{[2\kappa]}{\kappa}$.
In other words, all the colored vertices induce a  $K(2\kappa,\kappa)$, which, as observed at the beginning of this section, is a perfect matching.
%By construction, each edge in the matching is between a vertex with color $\kappa+1$ and a vertex with a different color.
%Hence, they are the uniquely colored neighbor of one another.
Thus each colored vertex has exactly one colored vertex as its neighbor, which serves as its uniquely colored neighbor.
% Let $w_{\kappa+1}$ denote the $\kappa$-set $\{2\kappa, 2\kappa+1, \dots, 3\kappa-1\}$.
% Any vertex $v \in C_1 \cup C_2 \cup \dots \cup C_k$ is a neighbor of $w_{\kappa+1}$.
% Since $w_{\kappa+1}$ is the lone vertex colored $\kappa+1$, it serves as the uniquely colored
% neighbor for any $v \in C_1 \cup C_2 \cup \dots \cup C_k$.

Now we have to show the presence of uniquely colored neighbors for vertices that have some elements not contained in 
%from outside 
$[2\kappa]$.
Let $v$ be such a vertex. 
That is, 
$v \cap [2\kappa] \neq v$. Let $t$ be the smallest nonnegative integer such that $\left | [\kappa+t] \setminus v \right | = \kappa$. Since $v$ has at least one element 
not contained in $[2\kappa]$, $t$ is at most $\kappa-1$.

By construction, the vertex $u = [\kappa+t] \setminus v$ has color $t+1$ and is adjacent to $v$.
Also by construction, $[\kappa+t]$ contains exactly $\kappa$ elements not in $v$ and all these $\kappa$ elements are in $u$.
Hence, for another vertex with color $t+1$, all of its $\kappa$ elements are in $[\kappa+t]$ and at least one of them is contained in $v$.
This implies that no other neighbors of $v$ have color $t+1$, and $u$ is the uniquely colored neighbor of $v$.
\qed
\end{proof}

%\subsection{Lower bound for CF-ON coloring of Kneser graph $K(n,\kappa)$}
%Now we show that $\kappa+1$ colors are necessary to \cfonp{} color $K(n,\kappa)$, when $n$ is large enough. 
% To show this, we need the full coloring variant of 
% the problem defined in Definition \ref{def:cfon}, where every vertex is assigned a color
% that is not 0.

% %This full coloring variant allows for a more fine-grained analysis.
% A \cfon{}-coloring with $\kappa$ colors, clearly implies a coloring with $\kappa$ colors for the partial variant \cfcn{}.
% A \cfonp{}-coloring with $\kappa$ colors implies a \cfon{}-coloring with $\kappa+1$ colors,
%     by simply using the color $\kappa+1$ instead of the dummy-color $0$.
% Thus, if a graph $G$ has $\kappa < \chionp{}\leq \kappa+1$,
%     it is still open whether $\chion{}=\kappa+1$.
% Of course, the same relationship between the partial and full coloring variant is true for the closed neighborhood version.

\begin{lemma}\label{lem:lowerbound}
$\kappa+1$ colors are necessary to %complete 
\cfonp{} %conflict-free 
color %a Kneser graph 
$K(n,\kappa)$ %for $n \geq 2\kappa + 1$.
when $n\geq 2\kappa^2 + \kappa$.
\end{lemma}

\begin{proof}
We prove this by contradiction. Suppose that $n\geq 2\kappa^2 + \kappa$ and $K(n,\kappa)$ can be colored using the $\kappa$
colors $1, 2, 3, \dots, \kappa$, besides the color 0.
For each $i$, $1 \leq i \leq \kappa$, let $C_i$ denote the set of all vertices colored with 
the color $i$. 
%For any vertex $v$, let $d_i(v) = | N(v) \cap C_i|$ denote the number of neighbors of $v$ 
%that are colored $i$. 

We will show that there exists a vertex $x$ 
that does not have a uniquely colored neighbor, i.e.,  
%$d_i(x) \neq 1$, for all $1 \leq i \leq \kappa+1$.
$| N(x) \cap C_i| \neq 1$, for all $i$, $1 \leq i \leq \kappa$.
We construct the vertex ($\kappa$-set) $x$, by choosing elements in it as follows. Suppose that there
are $C_i$'s that are singleton, i.e., $|C_i| = 1$. 
For every $i$, $1 \leq i \leq \kappa$ such that $|C_i| = 1$ add to $x$ an arbitrary element from the lone
vertex in $C_i$. 
In other words, we choose elements in $x$ so as to ensure that $x$ intersects
with the vertices in all the singleton $C_i$'s. This partially constructed $x$ may also 
intersect with vertices in other color classes. Some of the other $C_i$'s might become 
``effectively singleton'', that is $x$ may intersect with all the vertices in those $C_i$'s 
except one. We now choose further elements in $x$ so that $x$ intersects these effectively singleton
$C_i$'s too. Finally, we terminate this process when all the   
remaining $C_i$'s are not singleton. 

At this stage, if $x$ has exactly $\kappa$ elements, then it must be the case that $x$ intersects with all  the vertices in all the $C_i$'s. Hence no colored vertices are adjacent to $x$.
%Hence, $x$ has no colored neighbors.
 
Otherwise, the number of elements in $x$ is $t < \kappa$. There are two possible subcases. The first subcase is when $x$ intersects with all the colored vertices. In this case, we add $\kappa-t$ arbitrary elements to $x$ from  $[n] \setminus x$. This vertex $x$ is not adjacent to any colored vertex. 
The second subcase is when there are color classes that do not become effectively singleton. This is because  each of these color classes contain at least two vertices that do not intersect with $x$.
%To fill up the remaining entries of $x$, %(if any), 
%we consider the color classes(s) $C_j$ that have not become effectively singleton. 
For each of these color class(es) $C_j$, we choose
two distinct vertices, say $y_j, y_j' \in C_j$. We choose the remaining elements of $x$ so that $x \cap y_j = \emptyset$ and $x \cap y_j' = \emptyset$. The number of 
such sets $C_j$ is $\kappa-t$. So for choosing the remaining $\kappa-t$ elements of $x$, we have
at least $n -t - 2\kappa(\kappa-t)$ choices. The $t$ elements already present in $x$ cannot be used again.
There could be a maximum of $\kappa-t$ color classes $C_j$ which do not become effectively singleton. In each of these color classes, we want to avoid intersecting two vertices each, which forbids 
%means we have to avoid 
%choosing 
a maximum of $2\kappa(\kappa-t)$ elements. 
Because $n \geq 2\kappa^2 + \kappa$, it follows that the available $n -t - 2\kappa(\kappa-t)$ choices suffice to fill up the remaining $\kappa-t$ elements in $x$. Thus in this subcase, by construction, 
we ensure that $x$ has no neighbors in the color classes that become effectively singleton, and has at least two neighbors in the remaining color classes.
\qed
\end{proof}

Next, we consider \cfcnp{} coloring of Kneser graphs.
% It is easy to see that a proper 
% coloring (a coloring where no two vertices with the same color are adjacent) of any graph $G$ is also a 
% \cfcn{} %closed neighborhood conflict-free 
% coloring. 
%That is, $\chi^*_{CN}(G) \leq \chi(G)$ for all 
%graphs $G$.
%This is because in a proper coloring, for every vertex 
%$v$, none of its 
%neighbors will have the same color as $v$.
%The chromatic number (with respect to proper coloring) of a Kneser graph $K(n,\kappa)$ is
Observe that since the chromatic number of $K(n,\kappa)$ is $n - 2\kappa + 2$ \cite{lovasz}, we have that
$\chi_{CN}(K(n,\kappa)) \leq n - 2\kappa + 2$. %, where $\chi(G)$ represents the chromatic number of the graph $G$.
We show the following:

\begin{theorem}\label{thm:knesercfcn}
%Suppose $K(n,\kappa)$ is a Kneser graph. Then 
When $2\kappa+1\leq n\leq 3\kappa -1$, we have $\chi^*_{CN}(K(n,\kappa))  \leq n-2\kappa+1$. When $n\geq 3\kappa$, we have $\chi^*_{CN}(K(n,\kappa))  \leq \kappa$. 
\end{theorem}

\begin{lemma}\label{cfcnk+1} 
When $n\geq 2\kappa+1$, we have $\chi^*_{CN}(K(n,\kappa)) \leq \kappa$.
\end{lemma}

\begin{proof}
We assign the following coloring to the vertices of $K(n,\kappa)$: 

\begin{itemize}
    \item For any vertex ($\kappa$-set) $v$ that is a subset of $\{1, 2, \dots, 2\kappa-1\}$, we assign $C(v) = \max_{\ell \in v} \ell - (\kappa-1)$.
    \item All the remaining vertices are assigned the color 0. 
    %uncolored vertices are assigned color $\kappa+1$.
\end{itemize}

For $1 \leq i \leq \kappa$, let $C_i$ %, C_2, \dots, C_{\kappa+1}$ 
be the color class of the color $i$. 
%for the above coloring. 
Notice that $C_1 \cup C_2 \cup \dots \cup C_{\kappa} = \binom{[2\kappa-1]}{\kappa}$. Since any two $\kappa$-subsets of $\{1, 2, \dots, 2\kappa-1\}$ intersect,
it follows that 
$\binom{[2\kappa-1]}{\kappa}$ is an independent set. Hence each of the color classes $C_1, C_2, \dots, C_{\kappa}$ are independent sets, and each colored vertex 
%So if $v$ is colored with color $i$, where $1 \leq i \leq \kappa$, it has no neighbors of its own color. 
serves as its own uniquely colored neighbor.

If $v$ is %uncolored, 
assigned the color 0,
then $v \not \subset [2\kappa-1]$. That is, $v$ has some elements not contained in $[2\kappa-1]$. 
%= \{1, 2, \dots, 2\kappa - 1\}$. 
Let $t$
be the smallest nonnegative integer such that $\left | [\kappa + t] \setminus v \right | = \kappa$. Since $v$ has at least one element 
not contained in $[2\kappa-1]$, $t$ is at most $\kappa-1$.
We claim that the vertex $w=[\kappa + t] \setminus v$ is the only 
neighbor of $v$ with color $t+1$. 

First note that $\kappa+t \notin v$, because otherwise, the minimality of $t$ would not hold. It follows that the vertex $w$ is colored $t+1$. 
To show that $w$ is the only neighbor of $v$ with color $t+1$, assume the contrary. Let $w'$ be another neighbor of $v$ that is colored $t+1$. By the coloring used, $w' \subseteq [\kappa+t]$. Since $w \neq w'$, it follows that $|w \cup w'| \geq \kappa+1$, and hence $|[\kappa+t] \setminus v | \geq \kappa+1$. This again contradicts the choice of $t$. Thus $w$ is 
a uniquely colored neighbor of $v$.
\qed
\end{proof}

\begin{lemma}\label{2kappa+1lemma}
 $\chi_{CN}(K(2\kappa+1,\kappa)) = 2$, for all $\kappa\geq 1$.
\end{lemma}
\begin{proof}
Consider a vertex $v$ of $K(2\kappa+1, \kappa)$.
If $v\cap \{1,2\} \neq \emptyset$, we assign color 1 to $v$. Otherwise, we assign color 2 to $v$.

Let $C_1$ and $C_2$ be the sets of vertices colored $1$ and $2$ respectively. Below, we 
discuss the unique colors for every vertex of $K(n,\kappa)$.
\begin{itemize}
    \item If $v\in C_1$ and $\{1, 2\} \subseteq v$, then $v$ is the uniquely colored neighbor of itself. This is because all the
    vertices in $C_1$ contain either 1 or 2 and hence $v$ has no neighbors in $C_1$.
    
    \item Let $v\in C_1$ and $|v \cap \{1, 2\}| = 1$. WLOG, let $1 \in v$ and $2 \notin v$.
    In this case, $v$ has a uniquely colored neighbor $w\in C_2$. This vertex $w$ is the
    $\kappa$-set $w = [2\kappa +1] \setminus (v \cup \{2\})$.

    \item If $w\in C_2$, $w$ is the uniquely colored neighbor of itself. This is because 
    $C_2$ is an independent set. For two vertices 
    $w,w'\in C_2$ to be adjacent, we need $|w \cup w'| = 2\kappa$, 
    but vertices in $C_2$ are subsets of $\{3, 4, 5, \ldots, 2\kappa+1\}$, which has
    cardinality $2\kappa-1$.
\end{itemize}
\qed
\end{proof}

\begin{lemma}\label{cfcnd+1}
$\chi_{CN}(K(2\kappa + d,\kappa)) \leq d+1$, for all $\kappa\geq 1$. 
%When\footnote{When $\kappa= 1$, the graph becomes the complete graph on d+2 vertices %, $K_{d+2}$ for 
%for which 2 colors are sufficient to CF-CN color the kneser graph $K(2\kappa+d,\kappa)$. 
%\bscomment{Added the foot note.}} 
%$\kappa\geq 2$, $\chi_{CF-CN}(K(2\kappa + d,\kappa)) \leq d+1$,
\end{lemma}
\begin{proof}
We prove this by induction on $d$. 
The base case of $d=1$ is true by Lemma~\ref{2kappa+1lemma}. 
Suppose $K(2\kappa+d, \kappa)$ has a \cfcn{} coloring that uses $d+1$ colors.
Let us consider $K(2\kappa+d+1, \kappa)$. For all the vertices
of $K(2\kappa+d+1, \kappa)$ that appear in $K(2\kappa+d, \kappa)$
we use the same assignment as in $K(2\kappa+d, \kappa)$.
The new vertices (the vertices that contain 
$2\kappa + d +1$) are assigned the new color $d+2$.
As all the new vertices contain $2\kappa+d+1$, they 
form an independent set. Hence each of the new 
vertices serve as their own uniquely colored
neighbor. 

The vertices of $K(2\kappa+d+1, \kappa)$ already present in $K(2\kappa+d, \kappa)$ get new neighbors, but all the new neighbors 
are colored with the new color $d+2$. Hence the 
 unique colors of the existing vertices
are retained. 
\qed
\end{proof}

Lemma \ref{cfcnd+1} implies that $\chi^*_{CN}(K(n,\kappa)) \leq \chi_{CN}(K(n,\kappa)) \leq n - 2\kappa + 1$, when $n \geq 2\kappa + 1$.
So, from Lemma~\ref{cfcnk+1} and Lemma~\ref{cfcnd+1} we get 
Theorem \ref{thm:knesercfcn}. 

  \[
    \chi^*_{CN}(K(n,\kappa))  \leq \left\{\begin{array}{ll}
        n-2\kappa+1, & \text{for } 2\kappa+1\leq n\leq 3\kappa - 1 %\text{ and } d=n-2\kappa
        \\
        \kappa, & \text{for } n\geq 3\kappa\\
        %x(n-1), & \text{for } 0\leq n\leq 1
        \end{array}\right\}.% = xy
  \]

\section{Split Graphs}\label{sec:split}
%\subsection{Kneser Graphs}\label{sec:kneser}

In this section, we study  
%complete %conflict-free 
%both the problems on 
\cfonp{} and  \cfcnp{} colorings
of split graphs. 
We show that the \cfonp{} coloring problem is \NP-complete and  the \cfcnp{} coloring problem is polynomial time solvable.

\begin{definition}[Split Graph]
A graph $G$, $G = (V, E)$, is a split graph if there exists a partition of its vertex set $V = K \cup I$ 
such that %$G[K]$ is a clique and $G[I]$ is an independent set. 
the graph induced by $K$ is a clique and the graph induced by 
$I$ is an independent set.
\end{definition}
%\bscomment{Change the clique notation to K}

%We study \cfcnp{} and \cfonp{} problems on 
%split graphs. 
%The following are the results. 

%We have the following results on split graphs. 

\begin{theorem}\label{thm:split_cfonp_nphard}
The \cfonp{} coloring problem is \NP-complete on split graphs. 
\end{theorem}
\begin{proof}

We give a reduction from the classical \textsc{Graph Coloring} problem.
%\footnote{The coloring where every vertex is colored and adjacent vertices are colored differently.}. 
Given an instance $(G=(V,E),k)$ of \textsc{Graph Coloring}, %graph instance 
%$G=(V,E)$, 
we construct an auxiliary graph $G_1$, $G_1 = (V_1, E_1)$ from $G$ such that 
$V_1=V\cup \{x, y\}$ and $E_1=E\cup \{xy\} \cup \bigcup_{v \in V} \{xv , yv\}$. 
Note that $N(x)=V\cup\{y\}$ and $N(y)=V\cup \{x\}$. 
%are connected to every vertex in We get that 
%This makes the vertices 
%$x$ and $y$ are universal vertices in $G_1$. 
Now we construct the graph $G_2$, $G_2=(V_2, E_2)$,  from $G_1$ such that 

$$V_2=V_1\cup \{I_{uv} \; | \; uv \in E_1\} \cup \{I_v \; | \; v\in V_1\}, \mbox{ and}$$  
$$E_2=\{uv\;|\; u,v\in V_1\}\cup \{uI_{uv},  vI_{uv} \;|\; uv \in E_1\} \cup \{uI_u \;|\; u\in V_1\}.$$

Note that $G_2$ is a split graph $(K,I)$ with the clique 
$K=V_1$ and $I=V_2\setminus V_1$. See Figure \ref{figure:split} for an illustration. 
The construction of the graph $G_2$ from $G$ can be done in polynomial time. Let $I=I_1 \cup I_2$ where 
$I_1$ and $I_2$ represent the set of degree one vertices %of $V_1$ 
and 
%$I_2$ represents 
the set of 
degree two vertices in $I$ respectively. 

Now, we argue that $\chi(G) \leq k$ if and only if $\chi^*_{ON}(G_2) \leq k+2$,
%$G_2$ is \cfonp{} colorable using $k+2$ colors, 
where $k\geq 3$. 
We first prove the forward direction. Given a $k$-coloring $C_G$ of $G$, we extend $C_{G}$ to the coloring $C_{G_2}$ for $G_2$ using $k+2$ colors. 
For all vertices $v\in V$, $C_{G_2}(v)=C_G(v)$. 
We assign $C_{G_2}(x)=k+1$, $C_{G_2}(y)=k+2$. 
%With this all vertices in $K$ are colored. 
All vertices in $I_1\cup I_2$ are left uncolored. 
%For each vertex $I_v\in I_1$, assign $C_{G_2}(I_v)=C_{G}(v)$ and the vertices in $I_2$ are left uncolored. 
Every vertex $v\in K\setminus \{x\}$ has $x$ as its uniquely colored neighbor whereas the vertex $y$ is the uniquely colored neighbor for $x$. 
%and since all vertices 
%in $I_1$ are colored,  it acts as the uniquely colored neighbor for the vertices in $K$. 
For each vertex $I_{uv}\in I_2$, we have $N(I_{uv})=\{u,v\}$ and $C_{G_2}(u)\neq C_{G_2}(v)$. Hence the vertices $u$ and $v$ act as the uniquely colored neighbors for $I_{uv}$. 
%the vertices $u$ and $v$ as the uniquely colored neighbors. 
Each vertex $I_u\in I_1$ will have the vertex $u$ as its uniquely colored neighbor. 

Now, we prove the converse. Given a \cfonp{} $(k+2)$-coloring $C_{G_2}$ 
of $G_2$, we show that the restriction of $C_{G_2}$ 
%when restricted 
to the vertices of $G$ 
gives a proper $k$-coloring $C_G$ of $G$. 
Observe that each vertex in $K$ receives a non-zero color 
%is colored 
in any \cfonp{} coloring of $G_2$, because it is adjacent to a degree-one vertex in $I_1$. 
For every edge $uv \in E_1$, we have $C_{G_2}(u)\neq C_{G_2}(v)$ as $N(I_{uv})=\{u,v\}$.
This implies that $x$ and $y$ do not share the same color with each other nor with other vertices in $V$.
It also implies that for every edge $uv \in E$, we have $C_{G_2}(u)\neq C_{G_2}(v)$.
% Further as $x$ and $y$ are universal vertices in $G_2$, we have $C_{G_2}(x)\neq C_{G_2}(w)$ for all $w\in K\setminus \{x\}$ and $C_{G_2}(y)\neq C_{G_2}(w)$ for all $w\in K\setminus \{y\}$. 
% For every edge $uv \in E(G)$, we have $C_{G_2}(u)\neq C_{G_2}(v)$ as $N(I_{uv})=\{u,v\}$. 
Hence, the coloring $C_{G_2}$ when restricted to the set $K\setminus \{x,y\}=V$ is a $k$-coloring of $G$. 
\qed
\end{proof}

\begin{figure}
\vspace{-0.3cm}
\begin{center}
\begin{tikzpicture}
[scale=0.5,auto=left, node/.style={circle,fill=white, draw, scale=0.5}
	,max/.style={circle,fill=black, draw, scale=1}]

	\node[node] (a1) at (-3,-1) {};
	\node[node] (a2) at (-5,-1) {};
	\node[node] (a3) at (-4,1) {};

	\node [above] at (a3.north) {$a$};
	\node [below] at (a1.south) {$c$};
	\node [below] at (a2.south) {$b$};

	\node [above] at (-4, -4) {$G$};

	\node [above] at (2, -4) {$G_1$};
	\node [above] at (12, -7) {$G_2$};

	\node[node] (b1) at (2,-1) {};
	\node[node] (b2) at (0,-1) {};
	\node[node] (b3) at (1,1) {};
    \node[node] (b4) at (4, 1) {};
    \node[node] (b5) at (4,0) {};
 
	\node [above] at (b3.north) {$a$};
	\node [below] at (b1.south) {$c$};
	\node [below] at (b2.south) {$b$};
	\node [above] at (b4.north) {$x$};
	\node [below] at (b5.south) {$y$};

	\node[node] (c1) at (10,0) {};
	\node[node] (c2) at (11,0) {};
	\node[node] (c3) at (12,0) {};
    \node[node] (c4) at (13,0) {};
    \node[node] (c5) at (14,0) {};

    \node[node] (i1) at (10,2) {};
	\node[node] (i2) at (11,2) {};
	\node[node] (i3) at (12,2) {};
    \node[node] (i4) at (13,2) {};
    \node[node] (i5) at (14,2) {};

  \node[node] (iab) at (7,-4) {};
	\node[node] (iac) at (8,-4) {};
	\node[node] (iax) at (9,-4) {};
    \node[node] (iay) at (10,-4) {};
%    \node[node] (ibc) at (11,-4) {};
    \node[node] (ibx) at (11.5,-4) {};
    \node[node] (iby) at (12.5,-4) {};
	\node[node] (icx) at (14,-4) {};
	\node[node] (icy) at (15,-4) {};
    \node[node] (ixy) at (16.5,-4) {};

	\node [left] at (c3.north) {$c$};
	\node [left] at (c1.north) {$a$};
	\node [left] at (c2.north) {$b$};
	\node [left] at (c4.north) {$x$};
	\node [left] at (c5.north) {$y$};

	\node [above] at (i3.north) {$I_c$};
	\node [above] at (i1.north) {$I_a$};
	\node [above] at (i2.north) {$I_b$};
	\node [above] at (i4.north) {$I_x$};
	\node [above] at (i5.north) {$I_y$};

	\node [below] at (iab.south) {$I_{ab}$};
	\node [below] at (iac.south) {$I_{ac}$};
	\node [below] at (iax.south) {$I_{ax}$};
	\node [below] at (iay.south) {$I_{ay}$};
	%\node [below] at (ibc.south) {$I_{bc}$};
    \node [below] at (ibx.south) {$I_{bx}$};
	\node [below] at (iby.south) {$I_{by}$};
	\node [below] at (icx.south) {$I_{cx}$};
	\node [below] at (icy.south) {$I_{cy}$};
	\node [below] at (ixy.south) {$I_{xy}$};

   \node [right] (cli) at (15,0) {Clique $K$};

    \draw (11.8,0) ellipse (3cm and 1cm);

    \draw (a1) -- (a3) -- (a2); 
	\draw (b1) -- (b3) -- (b2);
    \draw (b3) -- (b4) -- (b5) -- (b1); 
    \draw  (b1) -- (b4) -- (b2) -- (b5) -- (b3); 
    \draw  (c1) -- (iab) -- (c2);
    \draw  (c1) -- (iac) -- (c3);
    \draw  (c1) -- (iax) -- (c4);
    \draw  (c1) -- (iay) -- (c5);
%    \draw  (c2) -- (ibc) -- (c3);
    \draw  (c2) -- (ibx) -- (c4);
    \draw  (c2) -- (iby) -- (c5);
    \draw  (c3) -- (icx) -- (c4);
    \draw  (c3) -- (icy) -- (c5);
    \draw  (c4) -- (ixy) -- (c5);

    \draw  (c1) -- (i1);
    \draw  (c2) -- (i2);
    \draw  (c3) -- (i3);
    \draw  (c4) -- (i4);
    \draw  (c5) -- (i5);

    %\draw [->] (14.8, 0) -- (cli);

\end{tikzpicture}
\end{center}
\vspace{-0.3cm}
\caption{Illustration of the graphs $G$ (on the left), $G_1$ (in the middle) and $G_2$ (on the right). The vertices $\{a,b,c,x,y\}$ of $G_2$ drawn inside the ellipse form the clique $K$.}
\label{figure:split}
\vspace{-0.5cm}
\end{figure}

\begin{theorem}\label{thm:split_cfcn}
The \cfcnp{} coloring problem is polynomial time solvable on split graphs. 
\end{theorem}

The proof of Theorem \ref{thm:split_cfcn} is through a characterization. We first show that for split graphs $G$, $ \chicnp{}\leq 2$. Then we characterize split graphs $G$
for which $\chicnp{} = 1$
thereby proving Theorem \ref{thm:split_cfcn}. 

%we We first show that two colors are sufficient to 
%\cfcnp{} color a split graph using Lemma \ref{lem:split_coloring_cfcn}. 
%Then we give characterization of split graphs that requires 1 color using Lemma \ref{lem:split_cfcn_charac}. 
%Thus settling Theorem \ref{thm:split_cfcn}. 

%\subsection{Proof of Theorem \ref{thm:split_cfcn}}
%are polynomial time 
%solvable for the \cfcnp{} problem. 

\begin{lemma}\label{lem:split_coloring_cfcn}
If $G$  is a split graph, %\textcolor{red}{ with  $|E|\geq 1$,} 
%with at least one edge, 
then $\chicnp{} \leq 2$.
\end{lemma}

\begin{proof}
Let $V =K\cup I$ be a partition of vertices into a clique $K$ and an independent set $I$. 
%a split graph. 
We use $C:V \rightarrow \{1,2,0\}$ to assign colors to the vertices of $V$. 
%color the vertices of $G$ with two colors 1 and 2, that is indeed a \cfcnp{} coloring. 
Choose an arbitrary vertex $u\in K$ and assign $C(u)=2$. The remaining vertices (if any) in $K\setminus \{u\}$ are assigned the color 0. For every vertex $v\in I$, we assign $C(v)=1$. 
%Choose a vertex $u\in K$ and assign the color $C(w)=2$. The remaining vertices in $K$
%are left uncolored. 
Each vertex in $I$ will have itself as the uniquely colored neighbor and every vertex in $K$ will have the vertex $u$ as the uniquely colored neighbor. 
%sees itself as the uniquely colored neighbor while every vertex in $K$ sees $u$ as the uniquely colored neighbor. 
\qed
\end{proof}

We now characterize all the split graphs that are \cfcnp{} colorable using one color. 
%requires only 1 color 
%for a \cfcnp{} coloring. 

\begin{lemma}\label{lem:split_cfcn_charac}
Let $G$ be a split graph %without a universal vertex and 
with $V=K\cup I$, where $K$ and $I$ are the clique and the independent set respectively. 
We have $\chicnp{} =1$ if and only if at least one of the following is true: 
(i) $G$ has a universal vertex, or (ii) $\forall v\in K, |N(v)\cap I | =1$. 
%For a split graph $G=(C,I)$ without a universal vertex, 
%\begin{enumerate}
%\item Suppose $K=\{u,v\}$.
%and $I=V\setminus C$. 
%If $|N(u)|= 2$ and $|N(v)|=2$ then $\chicnp{}=1$. 

%If $|C|=2$ then $\chicnp{} =1$.
%    \item Suppose $|K|> 2$. We have $\chicnp{} =1$ if and only if $\forall v\in C, |N(v)\cap I | =1$. 
%\end{enumerate}
\end{lemma}

\begin{proof}

We first prove the ``if'' statement.
If there exists a universal vertex $u\in V$, then we assign the color 1 to $u$ and assign the 
color 0 to all the other vertices. This is 
a \cfcnp{} coloring. 

Suppose that for each  vertex $v\in K$, $|N(v)\cap I|=1$.
(Note that $K$ cannot be empty because we assume $G$ to be connected.)
% \footnote{This case also captures the case when $K$ is empty.}
We assign the color 1 to each vertex in $I$ and color 0 to the vertices in $K$. Each vertex in $I$ acts as the uniquely colored neighbor for itself and for its neighbor(s) in $K$. 

For showing the ``only if'' statement, let $C:V\rightarrow \{1,0\}$ be a \cfcnp{} coloring of $G$.
We further assume that there exists $y \in K$ such that $|N(y)\cap I|\neq 1$
%has no universal vertices, 
and show that there exists a universal vertex. 
%This means that for each $v\in K$, either $|N(v)\cap I| = 0$
% or 
% $|N(v)\cap I|\geq 2$. 
We assume that $|K|\geq 2$  and $|I|\geq 1$ 
(if either assumption is violated, $G$ has a universal vertex).
We first prove the %and 
following claim. 

%We assume that $|I|\geq 1$ and 
%prove the following claim. 
\begin{claim}
Exactly one vertex in $K$ is assigned the color 1. 
\end{claim}

\begin{proof}
Suppose that there are two vertices $v,v'\in K$ such that $C(v)=C(v')=1$. Then none of the vertices in $K$ have a uniquely colored neighbor. 

Suppose that all vertices in $K$ are assigned the color 0. For vertices in $I$ to have a uniquely colored neighbor, each vertex in $I$ has to be assigned the color 1. 
By assumption, there is a vertex $y \in K$ such that  $|N(y)\cap I|\neq 1$.
This means that $y$ does not have a uniquely colored neighbor.
%should have a uniquely colored neighbor. 

In either case, it follows that $C$ is not a \cfcnp{} coloring of $G$, which is a contradiction.
\qed
\end{proof}

%Now we show that there is a universal vertex in $K$. 

%\textcolor{red}{I think the below paragraph can be skipped.}
%By the above claim, there is a unique vertex $v\in K$ such that $C(v) = 1$.
%We first show that  $|N(v)\cap I| > 0$.
%Assume for the sake of contradiction that  $|N(v)\cap I|=0$.
%For $v$ to have a uniquely colored neighbor, 
%exactly one vertex in $K$ is assigned the color 1. 
%We first argue that $C(v)=0$. Suppose not. Let $C(v)=1$ and from the above claim, all vertices in  $K\setminus \{v\}$ are assigned the color 0. 
%If all the vertices in $I$ are colored 0, then none of the vertices in $I$ have a uniquely colored neighbor. If a vertex $w\in I$ is such that $C(w)=1$, then its neighbor(s) in $K$ does not have a uniquely colored neighbor because of the vertices $v$ and $w$. 
%So $|N(v)\cap I|>0$.

By the above claim, there is a unique vertex $v\in K$ such that $C(v) = 1$.
We will show that $v$ is a universal vertex.  If not, there is a $w' \in I$
such that $w'\notin N(v)\cap I$. For $w'$ to have a uniquely colored neighbor, either $w'$ or one of its neighbors in $K$ has to be assigned the color 1. 
The latter is not possible because $v$ is the lone vertex in $K$ that is colored 1. If $C(w')=1$, then its neighbor(s) in $K$ does not have a uniquely colored neighbor because of the vertices $w'$ and $v$. 
Hence, $v$ is a universal vertex. 
 \qed
\end{proof}
By Lemmas \ref{lem:split_coloring_cfcn}, \ref{lem:split_cfcn_charac}, and the fact that conditions in the latter lemma can be checked in polynomial time, we obtain Theorem \ref{thm:split_cfcn}.

\section{Conclusion}
%\todo{One of the comments of Reviewer 4 is to remove the conclusion. Should we do it?}
In the preliminary version of our paper \cite{iwoca-sri}, we had shown that the conflict-free coloring problem is FPT when parameterized by combined parameters clique-width $w$ and number of colors $k$.
Since the problem is {\sf NP}-hard for any $k\geq 3$, the problem is not FPT when parameterized by $k$ unless {\sf P} $=$ {\sf NP}. 
%Since the problem is {\sf NP}-hard for constant number of colors $k$,
%it is unlikely to be FPT with respect to $k$ only.
%it is not in FPT unless {\sf P} $=$ {\sf NP}.
As we have shown in Theorems~\ref{thm:unbounded_cn} and~\ref{thm:unbounded_on}, the conflict-free chromatic numbers are not bounded by a function of the clique-width.
% If there exists such a bound, our algorithm would also be a fixed-parameter tractable algorithm for parameter $w$ only.
%So it is unlikely that the results in \cite{iwoca-sri} can be strengthened to an FPT algorithm for parameter clique-width $w$ alone.
Therefore it remains an open question if there exists an FPT algorithm 
with only clique-width as a parameter. 
%\todo[]{Para 1: Conveyed the same as the first bullet point in Results and Discussion. 

%This can probably be pushed (if we decide to not include conclusion) after the 1st bullet point. 

%The last paragraph may be ignored. The open question can be added in the R\&D sections for unitdisk graphs. }
%\todo[]{Changed this}

Recently, Gonzalez and Mann~\cite{DBLP:journals/corr/abs-2203-15724} 
%Graphs of bounded clique-width have bounded mim-width. 
%Hence it is interesting to see if there exists a FPT algorithm parameterized by mim-width and $k$, which generalizes our FPT algorithm parameterized by clique-width and $k$. 
%The results
%in \cite{DBLP:journals/corr/abs-2203-15724} 
showed that both open neighborhood and closed neighborhood variants are polynomial time solvable when 
mim-width and the number of colors are bounded. 
In particular, they design {\sf XP} algorithms in terms of mim-width and $k$. 
%the running time of their algorithms is  of the form $n^{O({\sf mw(G)\cdot k})}$, where ${\sf mw(G)}$ represents the mim-width of $G$. 
Since mim-width generalizes clique-width, 
%Hence their results do not imply our results. Therefore 
it is interesting to see if there exists an FPT algorithm parameterized by mim-width and $k$. 
%which generalizes our FPT algorithm parameterized by clique-width and $k$. 

%but are not FPT algorithms, 
%and hence are incomparable
%to our results on clique-width and $k$. 
%The algorithm in  \cite{DBLP:journals/corr/abs-2203-15724}
%shows that the problem is in XP, and not FPT.

Further, we presented an upper bound of conflict-free chromatic numbers for several graph classes.
For most of them we established graph classes that match or almost match the upper bounds for their respective conflict-free chromatic numbers.
For unit square and square disk graphs there is still a wide gap, 
and it would be interesting to improve those bounds.

\medskip

\noindent \textbf{Acknowledgments:} We would like to thank Stefan Lendl, Rogers Mathew, and Lasse Wulf for helpful
discussions.
We would also like to thank Alexander Hermans for his help on finding a lower bound example for interval graphs. 
We would also like to thank the anonymous reviewers for their helpful comments. 
%suggesting changes that improved the presentation of the paper. 
The first author,  the fourth author and the last author 
acknowledge SERB-DST  for funding to support this research via grants PDF/2021/003452 (SB), MTR/2020/000497 (SK), CRG/2022/009400 (SK) and SRG/2020/001162 (IVR).
The third author acknowledges support from the Austrian Science Foundation (FWF, project Y 1329 START-Programm).
%The fourth author acknowledges DST-SERB (MTR/2020/000497) for supporting this research.
%The last author acknowledges DST-SERB (SRG/2020/001162) for funding to support this research. 
%The first author acknowledges SERB (PDF/2021/003452) for funding to support this research. 

\noindent \textbf{Competing Interests:} The authors have no competing interests to declare that are relevant to the content of this article.
\bibliography{BibFile}
%\bibliography{main}

%\appendix
%\input{block_charac.tex}
%\input{inter_app.tex}
%\input{unitsquareanddisk-app.tex}

%\bibliographystyle{elsarticle-num}

%\bibliographystyle{splncs}

\end{document}